\documentclass[11pt]{article}

\usepackage{fullpage}
\usepackage{amsmath}
\usepackage{amsthm}
\usepackage{amssymb}
\usepackage{amsfonts}
\usepackage{hyperref}
\usepackage{color}
\usepackage{xcolor}

\newcommand{\ignore}[1]{}

\newtheorem{theorem}{Theorem}

\newtheorem{lemma}[theorem]{Lemma}

\newtheorem{assm}[theorem]{Assumption}
\newtheorem{remark}[theorem]{Remark}

\renewcommand{\Pr}{{\bf Pr}}
\newcommand{\E}{{\bf E}}
\newcommand{\D}{{\cal D}}
\newcommand{\T}{{\cal T}}

\newcommand{\Va}{{\rm{Va}}}
\newcommand{\Var}{{\rm{Var}}}
\def\etal{{\it et~al.}}
\newcommand{\prob}{{\rm Pr}}
\newcommand{\e}{\epsilon}
\newcommand{\ext}{{\tt ext}}

\newcommand{\ev}{\nu}
\newcommand{\si}{\sigma}
\newcommand{\bitset}{\{0,1\}}
\newcommand{\xth}{{\rm th}}
\newcommand{\poly}{{\rm poly}}

\newcommand{\ov}{\overline}
\newcommand{\BE}{\begin{enumerate}}
\newcommand{\EE}{\end{enumerate}}

\begin{document}

\title{Almost Optimal Testers for\\ Concise Representations}
\author{{\bf Nader H. Bshouty}\\ Dept. of Computer Science\\ Technion,  Haifa, 32000\\
}

\maketitle
\begin{abstract}
We give improved and almost optimal testers for several classes of Boolean functions on $n$ inputs that have concise representation in the uniform and distribution-free model. Classes, such as $k$-Junta, $k$-Linear Function, $s$-Term DNF, $s$-Term Monotone DNF, $r$-DNF, Decision List, $r$-Decision List, size-$s$ Decision Tree, size-$s$ Boolean Formula, size-$s$ Branching Program, $s$-Sparse Polynomial over the binary field and functions with Fourier Degree at most $d$.

The approach is new and combines ideas from
Diakonikolas et al.~\cite{DiakonikolasLMORSW07}, Bshouty~\cite{Bshouty19}, Goldreich et al.~\cite{GoldreichGR98}, and learning theory. The method can be extended to several other classes of functions over any domain that can be approximated by functions that have a small number of relevant variables.
\color{black}
\end{abstract}

\newpage
\tableofcontents
\newpage

\section{Inroduction}

Property testing of Boolean function was first considered in the seminal works of Blum, Luby and Rubinfeld~\cite{BlumLR93} and Rubinfeld and Sudan~\cite{RubinfeldS96} and has recently become a very active research area. See for example,~\cite{AlonKKLR05,BaleshzarMPR16,BelovsB16,BhattacharyyaKSSZ10,BlaisBM11,BlaisK12,Bshouty19,ChakrabartyS13,ChakrabartyS16a,ChakrabortyGM11,ChenDST15,ChenST14,ChenWX17,ChenWX17b,DiakonikolasLMORSW07,FischerKRSS02,
GoldreichGLRS00,GopalanOSSW11,KhotMS15,KhotS16,MatulefORS10,MatulefORS09,ParnasRS02,Saglam18} and other works  referenced in the surveys~\cite{GoldreichSurvey10,Ron08,Ron09}.

A Boolean function $f:\{0,1\}^n\to \{0,1\}$ is said to be $k$-junta if it depends on at most $k$ coordinates. The class $k$-Junta is the class of all $k$-juntas. The class $k$-Junta has been of particular interest to the computational learning theory community~\cite{Blum03,BlumL97,BshoutyC18,Damaschke00,GuijarroTT98,LiptonMMV05,MosselOS04}. A problem closely related to learning $k$-Junta is the problem of learning and testing subclasses $C$ of $k$-Junta and classes $C$ of Boolean functions that can be approximated by $k$-juntas~\cite{BlaisK12,BlumL97,DiakonikolasLMSW11,ChakrabortyGM11,DiakonikolasLMORSW07,GoldreichGR98,GopalanOSSW11,ParnasRS02}: Given black-box query access to a Boolean function $f$. In learning, for $f\in C$, we need to learn, with high probability, a hypothesis $h$ that is $\epsilon$-close to $f$. In testing, for any Boolean function $f$, we need to distinguish, with high probability, the case that $f$ is in $C$ versus the case that $f$ is $\epsilon$-far from every function in $C$.

In the {\it uniform-distribution property testing} (and learning) the distance between Boolean functions is measured with respect to the uniform distribution. In the {\it distribution-free property testing},~\cite{GoldreichGR98}, (and learning~\cite{Valiant84}) the distance between Boolean functions is measured with respect to an arbitrary and unknown distribution~${\cal D}$ over $\{0,1\}^n$. In the distribution-free model, the testing (and learning) algorithm is allowed (in addition to making black-box queries) to draw random $x\in\{0,1\}^n$ according to the distribution~${\cal D}$. This model is studied in~\cite{Bshouty19,ChenX16,DolevR11,GlasnerS09,HalevyK07,LiuCSSX18}.

\subsection{Results}

We give improved and almost optimal testers for several classes of Boolean functions on $n$ inputs that have concise representation in the uniform and distribution-free models. The classes studied here are $k$-Junta, $k$-Linear Functions, $k$-Term, $s$-Term DNF, $s$-Term Monotone DNF, $s$-Term Monotone $r$-DNF, $r$-DNF, Decision List, Length-$k$ Decision List, $r$-Decision List, size-$s$ Decision Tree, size-$s$ Branching Programs, size-$s$ Boolean Formula, size-$s$-Boolean Circuit, $s$-Sparse Polynomials over the binary field, $s$-Sparse Polynomials of Degree $d$ and functions with Fourier Degree at most $d$.

In Table~\ref{TABLE}, we list all the previous results and our results in this paper. In the table, $\tilde O(T)$ stands for $O(T\cdot poly(\log T))$, $U$ and $D$ stand for uniform and distribution-free model, and Exp and Poly stand for exponential and polynomial time.

It follows from the lower bounds of Saglam,~\cite{Saglam18}, that our query complexity is almost optimal (with log-factor) for the classes $k$-Junta, $k$-Linear, $k$-Term, $s$-Term DNF, $s$-Term Monotone DNF, $r$-DNF ($r$ constant), Decision List, $r$-Decision List ($r$ constant), size-$s$ Decision Tree, size-$s$ Branching Programs and size-$s$ Boolean Formula. For more details on the previous results and the results in this paper see Table~\ref{TABLE} and Sections~\ref{result1}, \ref{result2} and~\ref{result3}.

\begin{figure}[h!]
\begin{center}
\begin{tabular}{|c|c|c|c|c|}
\hline
{\bf Class of Functions}& {\bf Model} & {\bf $\#$Queries} & {\bf Time}&{\bf Reference}\\
\hline \hline
$s$-Term Monotone DNF & U & $\tilde O(s^2/\epsilon)$ &Poly. &\cite{ParnasRS02}\\
\cline{2-5}
$s$-Term Unate DNF & U & $\tilde O(s/\epsilon^2)$ &Exp. &\cite{ChakrabortyGM11}\\
\cline{2-5}
 & U & $\tilde O(s/\epsilon)$ &Poly. &This Paper\\
 \hline
 $s$-Term Monotone $r$-DNF & U & $\tilde O(s/\epsilon^2)$ &Exp. &\cite{ChakrabortyGM11}\\
\cline{2-5}
$s$-Term Unate $r$-DNF & U & $\tilde O(s/\epsilon)$ &Poly. &This Paper\\
\cline{2-5}
 & D & $\tilde O(s^2r/\epsilon)$ &Poly. &This Paper\\
\hline
$s$-Term DNF & U & $\tilde O(s^2/\epsilon)$ &Exp. &\cite{DiakonikolasLMORSW07}\\
\cline{2-5}
 & U & $\tilde O(s/\epsilon^2)$ &Exp. &\cite{ChakrabortyGM11}\\
 \cline{2-5}
 & U & $\tilde O(s/\epsilon)$ &Exp. &This Paper\\
\hline
$r$-DNF (Constant $r$)& U & $\tilde O(1/\epsilon)$ &Poly. &This Paper\\
\hline
Decision List & U & $\tilde O(1/\epsilon^2)$ &Poly. &\cite{DiakonikolasLMORSW07}\\
\cline{2-5}
 & U & $\tilde O(1/\epsilon)$ &Poly. &This Paper\\
 \hline
Length-$k$ Decision List & D & $\tilde O(k^2/\epsilon)$ &Poly. &This Paper\\
 \hline
$r$-DL (Constant $r$)& U & $\tilde O(1/\epsilon)$ &Poly. &This Paper\\
\hline
$k$-Linear & U & $\tilde O(k/\epsilon)$ &Poly. &\cite{Blais09,BlumLR93}\\
\cline{2-5}
 & D & $\tilde O(k/\epsilon)$ &Poly. &This Paper\\
\hline
$k$-Term & U & $ O(1/\epsilon)$ &Poly. &\cite{ParnasRS02}\\
\cline{2-5}
 & U & $\tilde O(1/\epsilon)$ &Poly. &This Paper\\
\cline{2-5}
 & D & $\tilde O(k/\epsilon)$ &Poly. &This Paper\\
 \hline
size-$s$ Decision Trees and& U & $\tilde O(s/\epsilon^2)$ &Exp. &\cite{ChakrabortyGM11}\\
\cline{2-5}
size-$s$ Branching Programs & U & $\tilde O(s/\epsilon)$ &Exp. &This Paper\\
\cline{2-5}
 & D & $\tilde O(s^2/\epsilon)$ &Exp. &This Paper\\
 \hline
size-$s$ Boolean Formulas &U &$\tilde O(s/\epsilon^2)$ &Exp.&\cite{ChakrabortyGM11}\\
\cline{2-5}
 & U & $\tilde O(s/\epsilon)$ &Exp. &This Paper\\
 \hline
 size-$s$ Boolean Circuit & U & $\tilde O(s^2/\epsilon^2)$ &Exp. &\cite{ChakrabortyGM11}\\
\cline{2-5}
 & U & $\tilde O(s^2/\epsilon)$ &Exp. &This Paper\\
 \hline
 Functions with & U & $\tilde O(2^{2d}/\epsilon^2)$ &Exp. &\cite{ChakrabortyGM11}\\
\cline{2-5}
Fourier Degree $\le d$&D & $\tilde O(2^d/\epsilon+2^{2d})$ &Poly. &This Paper\\
\hline
$s$-Sparse Polynomial & U & $poly(s/\epsilon)+\tilde O(2^{2d})$ &Poly. &\cite{AlonKKLR03,DiakonikolasLMSW11}\\
\cline{2-5}
over $F_2$ of Degree $d$ & U & $\tilde O(s^2/\epsilon+2^{2d})$ &Poly. &This Paper+\cite{AlonKKLR03}\\
\cline{2-5}
&U & $\tilde O(s/\epsilon+s2^d)$ &Poly. &This Paper\\
\cline{2-5}
&D & $\tilde O(s^2/\epsilon+s2^d)$ &Poly. &This Paper\\
\hline
$s$-Sparse Polynomial & U & $\tilde O(s/\epsilon^2)$ &Exp. &\cite{ChakrabortyGM11}\\
\cline{2-5}
over $F_2$&U & $Poly(s/\epsilon)$ &Poly. &\cite{DiakonikolasLMSW11}\\
\cline{2-5}
&U & $\tilde O(s^2/\epsilon)$  &Poly. &This Paper\\
\hline
\end{tabular}
\end{center}
	\caption{A table of the results. In the table, $\tilde O(T)$ stands for $O(T\cdot poly(\log T))$, $U$ and $D$ stand for uniform and distribution-free model, and Exp and Poly stand for exponential and polynomial time.}
	\label{TABLE}
	\end{figure}

\subsection{Notations}
\color{black}
In this subsection, we give some notations that we use throughout the paper.

Denote $[n]=\{1,2,\ldots,n\}$. For $S\subseteq [n]$ and $x=(x_1,\ldots,x_n)$ we denote $x(S)=\{x_i|i\in S\}$. For $X\subset [n]$ we denote by $\{0,1\}^X$
the set of all binary strings of
length $|X|$ with coordinates indexed by $i\in X$. For $x\in \{0,1\}^n$ and $X\subseteq [n]$ we write $x_X\in\{0,1\}^{X}$ to denote the projection of $x$ over coordinates in $X$. We denote by $1_X$ and $0_X$ the all-one and all-zero strings in $\{0,1\}^{X}$, respectively. When we write $x_I=0$ we mean $x_I=0_I$. For $X_1,X_2\subseteq [n]$ where $X_1\cap X_2=\emptyset$ and $x\in \{0,1\}^{X_1}, y\in \{0,1\}^{X_2}$ we write $x\circ y$  to denote their concatenation, i.e.,
the string in $\{0,1\}^{X_1\cup X_2}$ that agrees with $x$ over coordinates in $X_1$ and agrees with $y$ over coordinates in~$X_2$. Notice that $x\circ y=y\circ x$. When we write $u=\circ_{w\in W}w$ we mean that $u$ is the concatenation of all the strings in $W$. For $X\subseteq [n]$ we denote $\overline{X}=[n]\backslash X=\{x\in [n]|x\not\in X\}$. We say that two strings $x$ and $y$ are {\it equal on $I$} if $x_I=y_I$.

Given $f,g:\{0,1\}^n\to \{0,1\}$ and a probability distribution $\D$ over $\{0,1\}^n$, we say that $f$ is $\epsilon$-{\it close to $g$ with respect to} $\D$ if $\Pr_{x\in\D}[f(x)\not=g(x)]\le \epsilon$, where $x\in \D$ means $x$ is chosen from $\{0,1\}^n$ according to the distribution $\D$. We say that $f$ is $\epsilon$-{\it far from $g$ with respect to} $\D$ if $\Pr_{x\in\D}[f(x)\not=g(x)]\ge \epsilon$.
For a class of Boolean functions $C$, we say that $f$ is $\epsilon$-{\it far from every function in $C$ with respect to} $\D$ if for every $g\in C$, $f$ is $\epsilon$-far from $g$ with respect to $\D$. We will use $U$ to denote the uniform distribution over $\{0,1\}^n$ or over $\{0,1\}^X$ when $X$ in clear from the context.

For a Boolean function $f$ and $X\subset [n]$, we say that $X$ is a {\it relevant set} of $f$ if there
are $a,b\in \{0,1\}^n$ such that $f(a)\not= f(b_X\circ a_{\overline X})$. We call the pair $(a,b)$ (or just $a$ when $b=0$) a {\it witness} of $f$ for the relevant set $X$. When $X=\{i\}$ then we say that $x_i$ is a {\it relevant variable} of $f$ and $a$ is {\it a witness} of $f$ for $x_i$. Obviously, if $X$ is relevant set of $f$ then $x(X)$ contains at least one relevant variable of $f$.

We say that the Boolean function $f:\{0,1\}^n\to \{0,1\}$ is a literal if $f\in \{x_1,\ldots,x_n,\overline{x_1},\ldots,\overline{x_n}\}$ where $\overline{x}$ is the negation of $x$.

Let $C$ be a class of Boolean functions $f:\{0,1\}^n\to \{0,1\}$. We say that $C$ is {\it closed under variable projection} if for every {\it projection} $\pi:[n]\to [n]$ and every $f\in C$, we have $f(x(\pi))\in C$ where $x(\pi):=(x_{\pi(1)},\cdots,x_{\pi(n)})$. We say that $C$ is {\it closed under zero projection} (resp. {\it closed under one projection}) if for every $f\in C$ and every $i\in [n]$, $f(0_{\{i\}}\circ x_{\overline{\{i\}}})$  (resp. $f(1_{\{i\}}\circ x_{\overline{\{i\}}})\in C$). We say it is closed under zero-one projection if is closed under zero and one projection.

Throughout the paper, we assume that the class $C$ is closed under variable and zero projection. After section~3, we assume that it is also closed under one projection.

\subsection{The Model}
In this subsection, we define the testing and learning models.

In the testing model, we consider the problem of testing a class of Boolean function $C$ in the uniform and distribution-free testing models. In the distribution-free testing model (resp. uniform model), the algorithm has access to a Boolean function $f$ via a black-box that returns $f(x)$ when a string $x$ is queried. We call this query {\it membership query} (MQ$_f$ or just MQ). The algorithm also has access to unknown distribution $\D$ (resp. uniform distribution) via an oracle that returns $x\in\{0,1\}^n$ chosen randomly according to the distribution $\D$ (resp. according to the uniform distribution). We call this query {\it example query} (ExQ$_\D$ (resp. ExQ)).

A {\it distribution-free testing algorithm},~\cite{GoldreichGR98}, (resp. {\it testing algorithm}) ${\cal A}$ for $C$ is an algorithm that, given as input a distance parameter $\epsilon$ and the above two oracles to a Boolean function $f$,
\begin{enumerate}
\item if $f\in C$ then ${\cal A}$ outputs ``accept'' with probability at least $2/3$.
\item if $f$ is $\epsilon$-far from every $g\in C$ with respect to the distribution $\D$ (resp. uniform distribution) then ${\cal A}$ outputs ``reject'' with probability at least $2/3$.
\end{enumerate}

We will also call ${\cal A}$ {\it a tester (or $\epsilon$-tester) for the class $C$} and an algorithm for {\it $\epsilon$-testing $C$}.

We say that ${\cal A}$ is {\it one-sided} if it always accepts when $f\in C$; otherwise, it is called {\it two-sided} algorithm. The {\it query complexity of ${\cal A}$} is the maximum number of queries ${\cal A}$ makes on any Boolean function~$f$.

In the learning models, $C$ is a class of representations of Boolean functions rather than a class of Boolean functions. Therefore, we may have two different representations in $C$ that are logically equivalent. In this paper, we assume that this representation is verifiable, that is, given a representation $g$, one can decide in polynomial time on the length of this representation if $g\in C$.

A {\it distribution-free proper learning algorithm} (resp. proper learning algorithm under the uniform distribution) ${\cal A}$ for $C$ is an algorithm that, given as input an accuracy parameter $\epsilon$, a confidence parameter $\delta$ and an access to both MQ$_f$ for the {\it target function} $f\in C$ and ExQ$_\D$, with unknown $\D$, (resp. ExQ or ExQ$_U$), with probability at least $1-\delta$, ${\cal A}$ returns $h\in C$ that is $\epsilon$-close to $f$ with respect to $\D$ (resp. with respect to the uniform distribution). This model is also called {\it proper PAC-learning with membership queries} under any distribution (resp. under the uniform distribution)~\cite{Ang88,Valiant84}.
A {\it proper exact learning algorithm}~\cite{Ang88} for $C$ is an algorithm that given as input a confidence parameter $\delta$ and an access to MQ$_f$ for $f\in C$, with probability at least $1-\delta$, returns $h\in C$ that is equivalent to $f$. The {\it query complexity of ${\cal A}$} is the maximum number of queries ${\cal A}$ makes on any Boolean function~$f\in C$.

\section{Overview of the Distribution-Free Tester}\label{Oded}

\subsection{Preface}
Our approach refers to testing properties that are (symmetric)
sub-classes $C$ of $k$-juntas; that is, $f:\bitset^n\to\bitset$
has the property if there exists a function $f':\bitset^k\to\bitset$
that belongs to a predetermined class $C'$ of functions (over $k$-bit strings)
such that $f(x)=f'(x_\Gamma)$ for some $k$-subset~$\Gamma$.
Our new approach builds upon the ``testing by implicit sampling''
approach of Diakonikolas \etal~\cite{DiakonikolasLMORSW07},
while extending it from the case of uniform distribution to
the case of arbitrary unknown distributions
(i.e., the distribution-free model).

This allows us to present (almost optimal) {\em distribution-free}\/
testers for classes of properties that are sub-classes of $k$-juntas,
which correspond to classes of $k$-bit long Boolean functions.

While we follow Diakonikolas \etal~\cite{DiakonikolasLMORSW07}
in considering learning algorithms for the underlying classes,
our approach is also applicable to testing algorithms
(see~\cite[Sec.~6.2]{Goldreich17}).%

Let us again spell out our task.
For a class $ C$ of $n$-bit long Boolean functions
and a proximity parameter $\e$,
given samples from an unknown distribution $\D$
and oracle access to a function $f:\bitset^n\to\bitset$,
we wish to distinguish the case that $f\in C$ from
the case that $f$ is $\e$-far from $ C$.
Recall that $ C$ is a (symmetric) class consisting of
a symmetric subclass of $k$-juntas $ C'$; that is, $f\in C$
if and only if there exists a $k$-subset $\Gamma\subset[n]$
and $f'\in C'$ such that $f(x)=f'(x_\Gamma)$,
where $x_{\{i_1,\ldots,i_k\}}=(x_{i_1},\ldots,x_{i_k})$.
Actually, we also assume that $ C'$ is closed under
zero projection.

\subsection{A Bird's Eye View}
The basic strategy is to consider a random partition of $[n]$
to $r=O(k^2)$ parts, denoted $(X_1,\ldots,X_r)$, while relying on the
fact that, whp, each $X_i$ contains at most one influential variable.
Assuming that $f\in C$,
first we determine a set $I$ of at most $k$ indices
such that $\cup_{i\in[n]\setminus I}X_i$ contains
no ``significantly influential'' variables of $f$.
Suppose that $f':\bitset^k\to\bitset$, $f'\in C'$, is a function
that corresponds to the tested function $f:\bitset^n\to\bitset$,
and that $I\subset[n]$ is indeed the collection of all sets
that contain influential variables.
The crucial ingredient is devising a method that allows to
generate samples of the form $(x',f'(x'))$,
when given samples of the form $(x,f(x))$ (for $x\in\D$).
We stress that we cannot afford to find the influential variables,
and so this method works without determining these locations.
Using this method, we can test whether $f'$ belongs to the
underlying class $ C'$; hence, we test $f$ by implicitly sampling
the projection of $\D$ on the (unknown) influential variables.

The method employed by Diakonikolas \etal~\cite{DiakonikolasLMORSW07}
only handles the uniform distribution
(i.e., the case that $\D$ is uniform over $\bitset^n$),
and so it only yields testers for the standard testing model
(rather than for the distribution-free testing model).
Furthermore, their method as well as the identification of the set $I$
rely heavily on the notion of influence of sets, where the influence
of a set $S$ of locations on the value of a function is defined
as $\prob_{x',x''\in\bitset^n:x'_S=x''_S}[f(x')\!\neq\!f(x'')]$.
However, this notion refers to the uniform distribution
(over $\bitset^n$)
and does not seem adequate for the distribution-free context
(e.g., for\footnote{The addition operation in this paper is over the binary field $F_2$} $f(x)=x_1+x_2$ we may get
 $\prob_{x',x''\in\D:x'_1=x''_1}[f(x')\!\neq\!f(x'')]=0$).

We use a different way
of identifying the set $I$
and for generating samples for the underlying function $f'$.
Loosely speaking, we identifies $I$ as the set of indices $i$
for which $f(1_{X_i}\circ 0_{\overline{X_i}})\neq f(0^n)$, where (recall that) $1_S\circ 0_{\overline{S}}$ is a string
that is~1 on the locations in $S$ and is~0 on other locations.
({\em Be warned that this description is an over-simplification}!)
This means that for every $i\in I$ and $x\in\bitset^n$,
the value of $x$ at the influential variable in the set $X_i$
(a variable whose location is unknown to us!),
equals $f(x')+ f(0^n)$
where $x'=x_{X_i}\circ 0_{\overline X_i}$, i.e., $x'_j=x_j$ if $j\in X_i$ and $x'_j=0$ otherwise.%
\footnote{Indeed, if $\tau(i)\in X_i$ is the index of the
(unique) influential variable that resides in the set $X_i$,
then
$$f(x') = x_{\tau(i)}\cdot f(1_{X_i}\circ 0_{\overline{X_i}})
            +(x_{\tau(i)}+1)\cdot f(0^n)
        = x_{\tau(i)} + f(0^n)$$
since $f(1_{X_i}\circ 0_{\overline{X_i}})+ f(0^n)=1$.}
%
Note that the foregoing holds when $f\in C$;
in general, we can test whether $x\mapsto f(x')+ f(0^n)$
is close to a dictatorship (under the uniform distribution)
and reject otherwise, whereas if the mapping is close
to a dictatorship, we can self-correct it.

To sample the distribution $\D_{\Gamma}$,
where $\Gamma$ is the influential variables in $X_I=\cup_{i\in I}X_i$,
we sample $\D$ and determine the value of the influential
variable in each set $X_i$, for $i\in I$.
Queries to the function $f'$ are answered by querying $f$
such that the query $y=y_1\cdots y_k$ is mapped to
the query $\ext(y)$ such that\footnote{Notice that $\ext(y)=0_{\overline{X_I}}\circ \left(\underset{i\in I}{\circ} (y_i)_{X_i}\right)$ - Here $(y)_X=1_X$ if $y=1$ and $0_X$ if $y=0$.} $\ext(y)_j=y_i$
if $j$ belongs to the $i^\xth$ set in the collection $I$
(and $\ext(y)_j=0$ if $j\in[n]\setminus X_I$).
Effectively, we query the function $F:\bitset^n\to\bitset$
defined as $F(x)=f(\ext(x_\Gamma))$,
and this makes sense provided that $F$ is close to $f$
(under the distribution $\D$).
To test the latter hypothesis condition, we sample $\D$ and for each
sample point $x$ we compare $f(x)$ to $F(x)$,
where here we again use the ability to determine
the value of the influential variable in each set.
Specifically, $\ext(x_\Gamma)$ is computed by determining
the value of $x_\Gamma$ (without knowing $\Gamma$),
and using our knowledge of $(X_i)_{i\in I}$.

We warn that the foregoing description presumes
that we have correctly identified the collection $I$
of all sets containing an influential variable.
This leaves us with two questions:
The first question is, how do we identify the set $I$.
(Note that the influence of a variable may be as low as $2^{-k}$,
whereas we seek algorithms of $\poly(k)$-complexity.)
The solution (to be presented in Section~\ref{stage1:sec})
will be randomized, and will have one-sided error;
specifically, we may fail to identify some sets that
contain influential variables, but will never include
in our collection sets that have no influential variables.
Consequently, $f(1_{X_i}\circ 0_{\overline{X_i}})\neq f(0^n)$ may not hold for some $i\in I$,
and (over-simplifying again)
we shall seek instead some $v^{(i)}\in\bitset^n$
such that $f(v^{(i)})\neq f(w^{(i)})$,
where $w^{(i)}=v^{(i)}_{\overline{X_i}}\circ 0_{X_i}$ (i.e., $w^{(i)}_j=v^{(i)}_j$ if $j\in[n]\setminus X_i$
and $w^{(i)}_j=0$ otherwise).
Second, as before, for every $i\in I$ and $x\in\bitset^n$,
we wish to determine the value in $x$ of the influential variable
in the set $X_i$ (a variable whose location is unknown to us!).
This is done by observing that if $f\in C$ then this value
equals $f(x')+ f(v^{(i)})+1$
where $x'=x_{X_j}\circ v^{(i)}_{\overline{X_j}}$ (i.e., $x'_j=x_j$ if $j\in X_i$ and $x'_j=v^{(i)}_j$ otherwise).%
\footnote{Indeed, if $\tau(i)\in X_i$ is the index of the
(unique) influential variable that resides in the set $X_i$,
then
$$f(x') = x_{\tau(i)}\cdot f(v^{(i)})
            +(x_{\tau(i)}+1)\cdot f(w^{(i)})
        = x_{\tau(i)} + f(v^{(i)})+1$$
since $f(v^{(i)})+ f(w^{(i)})=1$.}
Again, we need to test whether $x\mapsto f(x')+ f(v^{(i)})+1$
is a dictatorship, and use self-correction.

\subsection{The Actual Tester}
As warned, the above description is an over-simplification,
and the actual way in which the set $I$ is identified and used
is more complex.

We fix a random partition of $[n]$ to $r=O(k^2)$ parts,
denoted $(X_1,\ldots, X_r)$.
If $f\in C$, then, with high probability, each $X_i$
contains at most one influential variable, denoted $\tau(i)$.
We assume that this is the case when providing intuition
throughout this section.

\subsubsection{Stage 1: Finding $I$ and corresponding $v^{(i)}$}\label{Stage01}
\label{stage1:sec}
Our goal is to find a collection $I$ of at most $k$ sets
such that the function $h_I$ is $\e/3$-close to $f$
(w.r.t distribution $\D$),
where $h_I$ is defined as $h_I(x)=f(x_{X_I}\circ 0_{\overline{X_I}})$
and $X_I=\cup_{i\in I} X_i$.
In addition, for each $i\in I$, we seek a witness $v^{(i)}$ for the
fact that $f$ depends on some variable in $X_i$;
that is, $f(v^{(i)}) \neq f(w^{(i)})$ for some $v^{(i)}$
that differ from $w^{(i)}$ only on $X_i$.

\noindent{\bf The procedure.}

We proceed in iterations, starting with $I=\emptyset$.

\BE
\item
We sample $\D$ for $O(1/\e)$ times,
trying to find $u\in\D$ such that $f(u) \neq h_I(u)$.

(Note that if $I=\emptyset$, then $h_I(u) = f(0^n)$.
In general, we seek $u$ such that $f(u) \neq f(u_{X_I}\circ 0_{\overline{X_I}})$.

If no such $u$ is found, then we set $h=h_I$
and proceed to Stage~2.
In this case, we may assume that $h_I$ is $\e/3$-close to $f$
(w.r.t $\D$).

\item
Otherwise (i.e., $f(u)\neq h_I(u)$),
we find an $i\in[m]\setminus I$ and $v^{(i)}$
such that $h_I(v^{(i)}) \neq h_{I\cup\{i\}}(v^{(i)})$,
which means that $X_i$ contains an influential variable
and $v^{(i)}$ is the witness for the sensitivity that we seek.
We set $I\gets I\cup\{i\}$ and proceed to the next iteration.

(We find this $i$ by binary search that seeks $i$ and $S$
such that $h_{I\cup S\cup\{i\}}(u) \neq h_{I\cup S}(u)$,
which means that $v^{(i)}$ equals $u$ in locations outside $S$
and is zero on $S$.)%
\footnote{By Step~1, we have $h_{S'\cup I}(u)\neq h_{S''\cup I}(u)$,
for $S'=[n]\setminus I$ and $S''=\emptyset$,
and in each iteration we cut $S'\setminus S''$ by half
while maintaining $h_{S'\cup I}(u)\neq h_{S''\cup I}(u)$.}
\EE
Once the iterations are suspended (due to not finding $u$),
we reject if $|I| > k$, and continue to the Stage~2 otherwise.
Recall that in the latter case $h=h_I$ is $\e/3$-close to $f$ (w.r.t $\D$).

Note that if $f\in C$, then $I$ contains only sets
that contain variables of the $k$-junta,
and so we never reject in this stage.
In general, if $i\in I$,
then $h_{I\setminus\{i\}}(v^{(i)}) \neq h_{I}(v^{(i)})$,
which implies that $f(x')\neq f(x'')$,
where $x'$ and $x''$ differ only on $X_i$
(e.g., $x''_{X_I}=v^{(i)}_{X_I}$ and $x''_j=0$ if $j\not\in X_I$).

\subsubsection{Stage 2: Extracting the value of an influential variable}
\label{stage2:sec}
Given a collection $I$ as found in Stage 1
(and a sensitivity witness $v^{(i)}$ for each $i\in I$),
let $h=h_I$ and recall that $h$ is close to $f$ w.r.t $\D$.
For each $i\in I$, given $x\in\bitset^{n}$,
we wish to determine the value of $x$
at the influential variable that resides in $X_i$.

For each $i\in I$, we define $\ev_i:\bitset^{|X_i|}\to\bitset$
such that $\ev_i(z) = h_I(y)$,
where $y_{X_i}=z$ and $y_{{\ov{X_i}}}=v^{(i)}_{{\ov{X_i}}}$.
Suppose that $f\in C$, and recall that $\tau(i)\in X_i$
denotes the location of the influential variable in $X_i$.
Let $\si(i)$ denote the index of $\tau(i)$ in $X_i$
(i.e., the $\si(i)^\xth$ element of $X_i$ is $\tau(i)$).
Then, in this case, $\ev_i$ is either a dictatorship
or an anti-dictatorship.
In particular, if $\ev_i$ is a dictatorship,
then $\ev_i(z)=z_{\si(i)}$ (and otherwise $\ev_i(z)=z_{\si(i)}+1$).

For each $i\in I$,
we test whether $\ev_i$ is a dictatorship or anti-dictatorship,
where testing is w.r.t the uniform distribution over $\{0,1\}^{|X_i|}$.
Note that we also check whether $\ev_i$ is a dictatorship
or anti-dictatorship.
If the tester (run with proximity parameter $0.1$) fails, we reject.
Otherwise (i.e., if we did not reject),
we can compute $\ev_i$ via self-correction on $h_I$;
that is, to compute $\ev_i$ at $z$,
we select $u\in\bitset^{|X_i|}$ at random,
and return $\ev_i(z+u)+\ev_i(u)$,
which (w.h.p.) equals $(z+u)_{\si(i)}+ u_{\si(i)}=z_{\si(i)}$.

Hence, we always continue to Stage~3 if $f\in C$,
and whenever we continue to Stage~3 we can compute
all $\ev_i$ (for $i\in I$) via self-correction.

\subsubsection{Stage 3: Emulating a tester of $ C'$}\label{thirds}
Recall that when reaching this stage,
we may assume that $h=h_I$ is $\e/3$-close to $f$ (w.r.t $\D$).
Also recall that $h_I(x)$ depends only on $x_{X_I}$,
where $X_I=\cup_{i\in I}X_i$,
and that by Stage~2 we may assume
that $\ev_i(z)=z_{\si(i)}$ (for every $i\in I$ and almost all $z$).
In light of the forgoing,
we define $F:\bitset^n\to\bitset$
such that $F(x)=h(x')$ where $x'_{X_i}=(x_{\si(i)},\ldots,x_{\si(i)})$
(i.e., $x'_j=(x_{X_i})_{\si(i)}=x_{\tau(i)}$ if $j\in X_i$)%
\footnote{In general, $\tau(i)$ denotes the location in $[n]$
of the $\si(i)^\xth$ element of $X_i$.}
and $x'_j=0$ otherwise.
(Indeed, if $f\in C$, then $F(x)=h(x)$,
since $h(y)$ depends only on $(y_{\tau(i)})_{i\in I}$.
Using hypothesis that $ C'$ (and so $ C$) is closed under
zero projection, it follows that $F\in C$.)

We observe that if $F$ is $\e/3$-close (w.r.t $\D$)
to both $h$ and $ C$,
then $f$ must be $\e$-close to $ C$
(since $f$ is $\e/3$-close to $h$).
Hence, we test both these conditions.
Specifically, using our ability to sample $\D$,
query $f$, and determine the value of the influential
variables in $X_I$, we proceed as follows:

\BE
\item Test whether $F=h$, where testing is w.r.t the distribution $\D$
and proximity parameter $\e/3$.

This is done by taking $O(1/\e)$ samples of $\D$,
and comparing the values of $F$ and $h$ on these sample points.
Recall that $h(u)=h_I(u)=f(u_{X_I}\circ 0_{\overline{X_I}})$.

The value of $F$ on $u$ is determined as follows.
\BE
\item For every $i\in I$, if $\ev_i$ is a dictatorship,
then set $v_i$ to equal the self-corrected value of $\ev_i(u_{X_i})$,
where $\ev_i$ is as defined in Stage~2.
Otherwise (i.e., when $\ev_i$ is an anti-dictatorship),
we set $v_i$ to equal the self-corrected value of $\ev_i(u_{X_i})+1$.
\item
Return the value $h(u')$, where $u'_j=v_i$ if $j\in X_i$
and $u'_j=0$ otherwise.
\EE
Indeed, $F=h$ always passes this test,
whereas $F$ that is $\e/3$-far from $h$ (w.r.t $\D$)
is rejected w.h.p.

\item Test whether $F$ is in $ C$,
where testing is w.r.t the distribution $\D$
and proximity parameter $\e/3$.
This is done by testing whether $F'$ is in $ C$,
where $F'(z)=F(x)$ such that $x_j=z_i$ if $j$ is in
the $i^\xth$ set in the collection $I$, and $x_j=0$ otherwise.
Here we use a distribution-free tester,
and analyze it w.r.t the distribution $\D_I$.
Toward this end, we need to samples $\D_I$
as well as answer queries to $F'$,
where both tasks can be performed as in the prior step.

Recall that if $f\in C$, then $F\in C$,
and this test will accept (w.h.p.),
whereas if $F$ is $\e/3$-far from $ C$
the test will reject (w.h.p.).
\EE
We conclude that if we reached Stage~3 and $f\in C$
(resp., $f$ is $\e$-far from $ C$),
then we accept (resp., reject) w.h.p.

\subsection{Digest: Our approach vs the original one~\cite{DiakonikolasLMORSW07}}
Our new approach differs from
the original approach of Diakonikolas \etal~\cite{DiakonikolasLMORSW07}
in two main aspects:
\BE
\item
In~\cite{DiakonikolasLMORSW07}, sets that contain influential variables
are identified according to their influence, which is defined
with respect to the uniform distribution. This definition
seems inadequate when dealing with arbitrary distributions.
Instead, we identify such a set by searching
for two assignments that differ only on this set and yield
different function values. The actual process is iterative
and places additional constraints on these assignments
(as detailed in Section~\ref{stage1:sec}).
\item
In~\cite{DiakonikolasLMORSW07}, given an assignment to the function,
the value of the unique influential variable that resides
in a given set $S$ is determined by approximating the influence
of two subsets of $S$ (i.e., the subsets of locations assigned
the value~0 and~1, respectively). In contrast, we
determines this value by defining an auxiliary function,
which depends on the unknown influential variable,
and evaluating this function (via self-correction w.r.t
the uniform distribution; see Section~\ref{stage2:sec}).
\EE

\subsection{More on Our Techniques}

In this section, we give a detailed overview of our techniques.

\subsubsection{Testing Subclasses of $k$-Junta}
For testing a subclass $C$ of $k$-Junta that is closed under variable and zero projections, we use {\bf Tester$C$} in Figure~\ref{Tester1}. We first note that {\bf Tester$C$} rejects if any procedure that it calls rejects.

First, {\bf Tester$C$} calls the procedure {\bf ApproxTarget}, in Figure~\ref{A31}. {\bf ApproxTarget} partitions the (indices of the) variables $[n]$ into $r=O(k^2)$ disjoint sets $X_1,\ldots,X_r$. Since $C\subseteq k-$Junta it follows that, with high probability (whp), if $f\in C$ then different relevant variables of $f$ fall into different sets. Therefore, if $f\in C$, whp, every $X_i$ contains at most one relevant variable of $f$. The procedure then binary searches for enough relevant sets $\{X_i\}_{i\in I}$ such that, whp, for $X=\cup_{i\in I} X_i$, $h=f(x_X\circ 0_{\overline{X}})$ is $(\epsilon/3)$-close to $f$ with respect to $\D$. If the procedure finds more than $k$ relevant sets of $f$ then there are more than $k$ relevant variables in $f$ and it rejects. If $f\in C$ then the procedure does not reject and, since $C$ is closed under zero projection, $h\in C$. Since, whp, $h$ is $(\epsilon/3)$-close to $f$ with respect to $\D$, it is enough to distinguish whether $h$ is in $C$ or $(2\epsilon/3)$-far from every function in $C$ with respect to $\D$.
{\bf ApproxTarget} also finds, for each relevant set $X_i$, $i\in I$, a witness  $v^{(i)}\in \{0,1\}^{n}$ of $h$ for $X_i$. That is, for every $i\in I$, $h(v^{(i)})\not=h(0_{X_i}\circ v^{(i)}_{\overline{X_i}})$. If $f\in C$, then $h\in C$ and, whp, for each $i\in I$, $h(x_{X_i}\circ v^{(i)}_{\overline{X_i}})$ is a literal. {\bf ApproxTarget} makes $\tilde O(k/\epsilon)$ queries.

In the second stage, the tester calls the procedure {\bf TestSets}, in Figure~\ref{A32}. {\bf TestSets} verifies, whp, that for every $i\in I$, $h(x_{X_i}\circ v^{(i)}_{\overline{X_i}})$ is $(1/30)$-close to some literal in $\{x_{\tau(i)},\overline{x_{\tau(i)}}\}$ for some $\tau(i)\in X_i$, with respect to the uniform distribution. If $f\in C$, then $h\in C$ and, whp, for each $i\in I$, $h(x_{X_i}\circ v^{(i)}_{\overline{X_i}})$ is a literal and therefore {\bf TestSets} does not reject. Notice that if $f\in C$, then, whp, $\Gamma:=\{x_{\tau(i)}\}_{i\in I}$ are the relevant variables of $h$. This test does not give $\tau(i)$ but the fact that $h(x_{X_i}\circ v^{(i)}_{\overline{X_i}})$ is close to $x_{\tau(i)}$ or $\overline{x_{\tau(i)}}$ can be used to find the value of $u_{\tau(i)}$ in every assignment $u \in \{0,1\}^{n}$ without knowing $\tau(i)$. The latter is done, whp, by the procedure {\bf RelVarValues}. See Figure~\ref{A3}. Both procedures make $\tilde O(k)$ queries.

Recall that for $\xi\in\{0,1\}$, $\xi_X$ is the all $\xi$ vector in $\{0,1\}^X$.
Then the tester defines the Boolean function $F=h(0_{\overline{X}}\circ\circ_{i\in I}(x_{\tau(i)})_{X_i})$ on the variables $\{x_{\tau(j)}\}_{j\in I}$, that is, the function $F$ is obtained by substituting in $h$ for every $i\in I$ and every $x_j\in x(X_i)$ the variable $x_{\tau(i)}$. Since $C\subseteq k$-Junta and $C$ is closed under variable and zero projections, $\tau(i)\in X_i$ and, whp, $\Gamma=\{x_{\tau(i)}\}_{i\in I}$ are the relevant variables of $h$ we have:
\begin{itemize}
\item If the function $f$ is in $C$ then, whp,  $F=h\in C$ and $F$ depends on all the variables in $\Gamma=\{x_{\tau(j)}\}_{j\in I}$.
\end{itemize}
If $h$ is $(2\epsilon/3)$-far from every function in $C$ with respect to $\D$ then either $h$ is $(\epsilon/3)$-far from $F$ with respect to $\D$ or $F$ is $(\epsilon/3)$-far from every function in $C(\Gamma)$ with respect to $\D$ where $C(\Gamma)$ is the set of all functions in $C$ that depends on all the variables in $\Gamma$.
Therefore,
\begin{itemize}
\item If the function $f$ is $\epsilon$-far from every function in $C$ then, whp, either
\begin{enumerate}
\item $h$ is $(\epsilon/3)$-far from $F$ with respect to $\D$ or
\item $F$ is $(\epsilon/3)$-far from every function in $C(\Gamma)$ with respect to $\D$.
\end{enumerate}
\end{itemize}
 Therefore, it remains to do two tests. The first is testing whether $h =F$ given that $h$ is either $(\epsilon/3)$-far from $F$ with respect to $\D$ or $h=F$. The second is testing whether $F\in C$ given that $F$ is either $(\epsilon/3)$-far from every function in $C(\Gamma)$ with respect to~$\D$ or $f\in C(\Gamma)$.

The former test, $h=F$, can be done, whp, by choosing $O(1/\epsilon)$ strings $u\in\{0,1\}^n$ according to the distribution $\D$ and testing whether $F(u)=h(u)$. To compute $F(u)$ we need to find $\{u_{\tau(i)}\}_{i\in I}$, which can be done by the procedure {\bf RelVarValues}. Therefore, each query to $F$ requires one call to the procedure {\bf RelVarValues} that uses $\tilde O(k)$ queries to $f$. Thus, the first test can be done using $\tilde O(k/\epsilon)$ queries. This is done in the procedure {\bf Close$fF$} in Figure~\ref{Close01}.

Notice that, thus far, all the above procedures run in polynomial time and make $\tilde O(k/\epsilon)$ queries.

Testing whether $F\in C$ can be done, whp, by choosing $O((\log |C(\Gamma)|)/\epsilon)$ strings $u\in\{0,1\}^n$ according to the distribution $\D$ and testing whether $F(u)=g(u)$ for every $g\in C(\Gamma)$. Notice here that the time complexity is $poly(|C(\Gamma)|)$ which is polynomial only when $C(\Gamma)$ contains polynomial number of functions.

If the distribution is uniform, we do not need to use {\bf RelVarValues} to find $\{u_{\tau(i)}\}_{i\in I}$ because when the distribution of $u$ is uniform the distribution of $\{u_{\tau(i)}\}_{i\in I}$ is also uniform. Therefore we can just test whether $F(u)=g(u)$ for every $g\in C(\Gamma)$ for uniform $\{u_{\tau(i)}\}_{i\in I}$. Then computing $F(u)$ for random uniform string $u$ can be done in one query to $h$. Thus, for the uniform distribution, the algorithm makes $\tilde O((\log |C(\Gamma)|)/\epsilon)$ queries to $f$. This is the procedure {\bf Close$FCU$} in Figure~\ref{Close02}.

If the distribution is unknown then each computation of $F(u)$ for a random string $u$ according to the distribution $\D$ requires choosing $u$ according to the distribution $\D$, then extracting $\{u_{\tau(i)}\}_{i\in I}$ from $u$ and then substituting the values $\{u_{\tau(i)}\}_{i\in I}$ in $F$. This can be done by the procedure {\bf RelVarValues} using $\tilde O(k)$ queries to $h$. Therefore, for unknown distribution the algorithm makes $\tilde O((k\log |C(\Gamma)|)/\epsilon)$ queries to $f$. This is the procedures {\bf Close$FCD$} in Figure~\ref{Close02}.

As we mentioned before the time complexity of {\bf Close$FCU$} and {\bf Close$FCD$} is polynomial only if $|C(\Gamma)|$ is polynomial.
When $|C(\Gamma)|$ is exponential, we solve the problem via learning theory. We find a
proper learning algorithm ${\cal A}$ for $C(\Gamma)$. We run ${\cal A}$ to learn $F$. If the algorithm fails, runs more time than it should, asks more queries than it should or outputs a hypothesis $g\not\in C$ then we know that, whp, $F\not\in C(\Gamma)$. Otherwise, it outputs a function $g\in C(\Gamma)$ and then, as above, we test whether $g=F$ given that $g$ is $(\epsilon/3)$-far from $F$ or $g=F$.

Therefore, for the uniform distribution, if the proper learning algorithm for $C$ makes $m$ MQs and $q$ ExQs then the tester makes $m+q+O(1/\epsilon)$ queries. If the distribution is unknown, then the tester makes $m+\tilde O(kq+k/\epsilon)$ queries.

\subsubsection{Testing Classes that are Close to $k$-Junta}
To understand the intuition behind the second technique, we demonstrate it for testing $s$-term DNF.

The tester first runs the procedure {\bf Approx$C$} in Figure~\ref{A3f}. This procedure is similar to the procedure {\bf ApproxTarget}. It randomly uniformly partitions the variables to $r=4c^2(c+1)s\log(s/\epsilon)$ disjoint sets $X_1,\ldots,X_r$ and finds relevant sets $\{X_i\}_{i\in I}$. Here $c$ is a large constant. To find a new relevant set, it chooses two random uniform strings $u,v\in \{0,1\}^n$ and verifies if $f(u_X\circ v_{\overline{X}})\not=f(u)$ where $X$ is the union of the relevant sets that it has found thus far. If $f(u_X\circ v_{\overline{X}})\not=f(u)$ then the binary search finds a new relevant set.

In the binary search for a new relevant set, the procedure defines a set $X'$ that is equal to the union of half of the sets in $\{X_i\}_{i\not\in I}$. Then either $f(u_{X\cup X'} \circ v_{\overline{X'}})\not=f(u)$ or $f(u_{X\cup X'} \circ v_{\overline{X'}})\not= f(u_X\circ v_{\overline{X}})$. Then it recursively does the above until it finds a new relevant set $X_\ell$.

It is easy to show that if $f$ is $s$-term DNF then, whp, for all the terms $T$ in $f$ of size at least $c^2\log(s/\epsilon)$, for all the random uniform strings $u,v$ chosen in the algorithm and for all the strings $w$ generated in the binary search, $T(u_X\circ v_{\overline{X}})=T(u)=T(w)=0$. Therefore, when $f$ is $s$-term DNF, the procedure, whp,  runs as if there are no terms of size greater than $c^2\log(s/\epsilon)$ in $f$. This shows that, whp, each relevant set that the procedure finds contains at least one variable that belongs to a term of size at most $c^2\log(s/\epsilon)$ in $f$. Therefore, if $f$ is $s$-term DNF, the procedure, whp,  does not generate more than $c^2s\log(s/\epsilon)$ relevant sets. If the procedure finds more than $c^2s\log(s/\epsilon)$ relevant sets then, whp, $f$ is not $s$-term DNF and therefore it rejects.

Let $R$ be the set of all the variables that belong to the terms in $f$ of size at most $c^2\log(s/\epsilon)$. The procedure returns $h=f(x_X\circ w_{\overline{X}})$ for random uniform $w$ where $X$ is the union of the relevant sets $X=\cup_{i\in I}X_i$ that is found by the procedure. If $f$ is $s$-term DNF then since $r=4c^2(c+1)s\log(s/\epsilon)$ and the number of relevant sets is at most $c^2s\log(s/\epsilon)$, whp, at least $(1/2)c\log(s/\epsilon)$ variables in each term of $f$ that contains at least $c\log(s/\epsilon)$ variables not in $R$ falls outside $X$ in the partition of $[n]$. Therefore, for random uniform $w$, whp, terms $T$ in $f$ that contains at least $c\log(s/\epsilon)$ variables not in $R$ satisfies $T(x_X\circ w_{\overline{X}})=0$ and therefore, whp, are vanished in $h=f(x_X\circ w_{\overline{X}})$. Thus, whp, $h$ contains all the terms that contains variables in $R$ and at most $cs\log(s/\epsilon)$ variables not in $R$. Therefore, whp, $h$ contains at most $c(c+1)s\log(s/\epsilon)$ relevant variables. From this, and using similar arguments as for the procedure {\bf ApproxTarget} in the previous subsection, we prove that,
{\bf Approx$C$} makes at most $\tilde O(s/\epsilon)$ queries and
\begin{enumerate}
\item If $f$ is $s$-term DNF then, whp, the procedure outputs $X$ and $w$ such that
\begin{itemize}
\item $h=f(x_X\circ w_{\overline{X}})$ is $s$-term DNF.
\item The number of relevant variables in $h=f(x_X\circ w_{\overline{X}})$ is at most $O(s\log(s/\epsilon))$.
\end{itemize}
\item If $f$ is $\epsilon$-far from every $s$-term DNF then the procedure either rejects or outputs $X$ and $w$ such that, whp, $h=f(x_X\circ w_{\overline{X}})$ is $(3\epsilon/4)$-far from every $s$-term DNF.
\end{enumerate}
We can now run {\bf Tester$C$} (with $3\epsilon/4$) on $h$ from the previous subsection for testing $C^*$ where $C^*$ is the set of $s$-term DNF with $k=O(s\log(s/\epsilon))$ relevant variables. All the procedures makes $\tilde O(s/\epsilon)$ queries except {\bf Close$FCU$} that makes $\tilde O(s^2/\epsilon)$ queries.  This is because that the size of the class $C^*(\Gamma)$ is $2^{\tilde O(s^2)}$ and therefore {\bf Close$FCU$} makes $\tilde O(s^2/\epsilon)$ queries. This gives a tester that makes $\tilde O(s^2/\epsilon)$ queries which is not optimal.

Instead, we consider the class $C'$ of $s$-term DNF with $O(s\log(s/\epsilon))$ variables and terms of size at most $c\log(s/\epsilon)$ and show that, in {\bf Close$FCU$}, whp,  all the terms $T$ of size greater than $c\log(s/\epsilon)$ and all the random strings $u$ chosen in the procedure satisfies $T(u)=0$ and therefore it runs as if the target function $h$ has only terms of size at most $c\log(s/\epsilon)$. This gives a tester that makes $\tilde O(s/\epsilon)$ queries.

As in the previous section, all the procedures run in polynomial time except {\bf Close$FCU$}. For some classes, we replace {\bf Close$FCU$} with polynomial time learning algorithms and obtains polynomial time testers.

\color{black}

\section{Preparing the Target for Accessing the Relevant Variables}

In this Section we give the three procedures {\bf ApproxTarget}, {\bf TestSets} and {\bf RelVarValues}.

\subsection{Preliminaries}

In this subsection, we give some known results that will be used in the sequel.

The following lemma is straightforward
\begin{lemma}\label{trivial01} If $\{X_i\}_{i\in [r]}$ is a partition of $[n]$ then for any Boolean function $f$ the number of relevant sets $X_i$ of $f$ is at most the number of relevant variables of $f$.
\end{lemma}

We will use the following folklore result that is formally proved in~\cite{LiuCSSX18}.
\begin{lemma}\label{BiSe} Let $\{X_i\}_{i\in [r]}$ be a partition of $[n]$. Let $f$ be a Boolean function and $u,w\in \{0,1\}^n$. If $f(u)\not= f(w)$ then a relevant set $X_\ell$ of $f$ with a string $v\in \{0,1\}^n$ that satisfies $f(v)\not=f(w_{X_\ell}\circ v_{\overline{X_\ell}})$ can be found using $\lceil \log_2 r\rceil$ queries.
\end{lemma}

The following is from \cite{Blais09}
\begin{lemma}\label{OneSide} There exists a one-sided adaptive algorithm, {\bf UniformJunta}$(f , k, \epsilon,\delta)$, for $\epsilon$-testing $k$-junta that makes $O(((k/\epsilon) + k \log k)\log(1/\delta))$ queries and rejects $f$ with probability at least $1-\delta$ when it is $\epsilon$-far from every $k$-junta with respect to the uniform distribution.

Moreover, it rejects only when it has found $k+1$ pairwise disjoint relevant sets and a witness of $f$ for each one.
\end{lemma}

\subsection{Approximating the Target}
In this subsection we give the procedure {\bf ApproxTarget} that returns $(X=\cup_{i\in I}X_i,V=\{v^{(i)}\}_{i\in I},I)$, $X\subseteq [n]$, $V\subseteq \{0,1\}^n$ and $I\subseteq [r]$ where, whp, each $x(X_i)$, $i\in I$, contains at least one relevant variable of $h:=f(x_X\circ 0_{\overline{X}})$ and exactly one if $f$ is $k$-junta. Each $v^{(i)}$, $i\in I$, is a witness of $f(x_{X}\circ 0_{\overline{X}})$ for the relevant set $X_i$. Also, whp,  $f(x_{X}\circ 0_{\overline{X}})$ is $(\epsilon/c)$-close to the target with respect to the distribution $\D$.

\newcounter{ALC}
\setcounter{ALC}{0}
\newcommand{\step}{\stepcounter{ALC}$\arabic{ALC}.\ $\>}
\newcommand{\steplabel}[1]{\addtocounter{ALC}{-1}\refstepcounter{ALC}\label{#1}}
\begin{figure}[h!]
  \begin{center}
  \fbox{\fbox{\begin{minipage}{28em}
  \begin{tabbing}
  xxx\=xxxx\=xxxx\=xxxx\=xxxx\=xxxx\= \kill
  {{\bf ApproxTarget}$(f,\D,\epsilon,c)$}\\
 {\it Input}: Oracle that accesses a Boolean function $f$ and \\
\>\>an oracle that draws $x\in \{0,1\}^n$ according to the distribution $\D$. \\
  {\it Output}: Either ``reject'' or $(X,V,I)$\\ \\
{\bf Partition $[n]$ into $r$ sets}\\
\step\steplabel{par11}
Set $r = 2k^2$.\\
\step\steplabel{par21}
Choose uniformly at random a partition $X_1,X_2,\ldots,X_r$ of $[n]$\\
\\
{\bf Find a close function and relevant sets} \\
\step\steplabel{Sett1}
Set $X=\emptyset$; $I=\emptyset$; $V=\emptyset$; $t(X)=0$.\\
\step\steplabel{two1}
Repeat $M=ck\ln(15k)/\epsilon$ times\\
\step\steplabel{Cho1}
\> Choose $u\in {\cal D}$. \\
\step  \> $t(X)\gets t(X)+1$\\
\step\steplabel{con11}
\> If $f(u_X\circ 0_{\overline{X}})\not=f(u)$ then\\
\step\steplabel{Wempty}\>\>\> $W\gets \emptyset$.  \\
\step\steplabel{Find1}
\>\>\> Binary Search to find a new relevant set from $(u,u_X\circ 0_{\overline{X}})\to \ell$;\\
\step\steplabel{Finddd1}    \>\>\>\> and a string $w^{(\ell)}\in\{0,1\}^n$ such that $f(w^{(\ell)})\not= f(w^{(\ell)}_{\overline{X_\ell}}\circ 0_{X_\ell})$;\\
\step\steplabel{Xadd}\>\>\> $X\gets X\cup X_\ell$; $I\gets I\cup \{\ell\}$.\\
\step\steplabel{Rej1}
\>\>\> If $|I|>k$ then Output(``reject'').\\
\step\>\>\> $W=W\cup\{w^{(\ell)}\}$.\\
\step\>\>\>\steplabel{ccch}Choose $w^{(r)}\in W$.\\
\step\>\>\>\steplabel{Tr}If $f(w^{(r)}_X\circ 0_{\overline{X}})\not=f(w^{(r)}_{X\backslash X_r}\circ 0_{\overline{X}\cup X_r})$ then \\
\>\>\>\>\> $W\gets W\backslash\{w^{(r)}\}; v^{(r)}\gets w^{(r)}_X\circ 0_{\overline{X}}; V\gets V\cup\{v^{(r)}\};$\\
\>\>\>\>\> If $W\not=\emptyset$ then Goto~\ref{ccch}\\
\step\>\>\>\steplabel{Fir01} Else If $f(w^{(r)}_X\circ 0_{\overline{X}})\not=f(w^{(r)})$ then $u\gets w^{(r)}$; Goto~\ref{Find1} \\
\step \>\>\>\steplabel{Fir02} Else $u\gets w^{(r)}_{\overline{X_r}}\circ 0_{X_r}$; Goto~\ref{Find1}\\
\step\steplabel{tx01}
\>\>\> $t(X)=0$.\\
\step\steplabel{EndRep1}
\>  If $t(X)=c\ln(15k)/\epsilon$ then Output($X,V,I$).
  \end{tabbing}
  \end{minipage}}}
  \end{center}
	\caption{A procedure that finds relevant sets $\{X_i\}_{i\in I}$ of $f$ and a witness $v^{(i)}$ for each relevant set~$X_i$ for $h:=f(x_X\circ 0_{\overline{X}})$ where $X=\cup_{i\in I}X_i$. Also, whp, $h$ is $(\epsilon/c)$-close to the target.}
	\label{A31}
	\end{figure}

Consider the procedure {\bf ApproxTarget} in Figure~\ref{A31}. In steps~\ref{par11}-\ref{par21} the procedure partitions the set $[n]$ into $r=2k^2$ disjoint sets $X_1,X_2,\ldots,X_r$. In step~\ref{Sett1} it defines the variables $X,I,V$ and $t(X)$. At each iteration of the procedure, $I$ contains the indices of some relevant sets of $f(x_X\circ 0_{\overline{X}})$ where $X=\cup_{i\in I}X_i$, i.e., each $X_i$, $i\in I$ is relevant set of $f(x_X\circ 0_{\overline{X}})$. The set $V$ contains, for each $i\in I$, a string $v^{(i)}\in \{0,1\}^n$ that satisfies $f(v^{(i)}_X\circ 0_{\overline{X}})\not=f(v^{(i)}_{X\backslash X_i}\circ 0_{X_i}\circ 0_{\overline{X}})$. That is, a witness of $f(x_X\circ 0_{\overline{X}})$ for the relevant set $X_i$, $i\in I$.

The procedure in steps~\ref{two1}-\ref{EndRep1} tests if $f(u_X\circ 0_{\overline{X}})=f(u)$ for at least $c\ln(15/k)/\epsilon$, independently and at random, chosen $u$ according to the distribution $\D$. The variable $t(X)$ counts the number of such $u$. If this happens then, whp, $f(x_X\circ 0_{\overline{X}})$ is $(\epsilon/c)$-close to $f$ with respect to $\D$ and the procedure returns $(X,V,I)$. If not then $f(u_X\circ 0_{\overline{X}})\not=f(u)$ for some $u$ and then a new relevant set is found. If the number of relevant sets is greater than $k$, it rejects. This is done in steps~\ref{Wempty}-\ref{tx01}.

In steps~\ref{Find1}-\ref{Finddd1}, the procedure uses Lemma~\ref{BiSe} to (binary) searches for a new relevant set. The search gives an index $\ell$ of the new relevant set $X_\ell$ and a  witness $w^{(\ell)}$ that satisfies $f(w^{(\ell)})\not= f(0_{X_\ell}\circ w^{(\ell)}_{\overline{X_\ell}})$. Then $\ell$ is added to $I$ and $X$ is extended to $X\cup X_\ell$. The binary search gives a witness that $X_\ell$ is relevant set of $f$, but not a witness that it is relevant set of $f(x_X\circ 0_{\overline{X}})$. This is why we need steps~\ref{ccch}-\ref{Fir02}. In those steps the procedure adds $w^{(\ell)}$ to $W$. Then for each $w^{(r)}\in W$ (at the beginning $r=\ell$) it checks if $w^{(r)}$ is a witness of $f(x_X\circ 0_{\overline{X}})$ for $X_r$. If it is then it adds it to $V$. If it isn't then we show in the discussion below that a new relevant set can be found. The procedure rejects when it finds more than $k$ relevant sets.

If the procedure does not reject then it outputs $(X,V,I)$ where $I$ contains the indices of some relevant sets of $f(x_{X}\circ 0_{\overline{X}})$, $X=\cup_{i\in I}X_i$ and the set $V$ contains for each $i\in I$ a string $v^{(i)}\in \{0,1\}^n$ that is a witness of $f(x_{X}\circ 0_{\overline{X}})$ for $X_i$, i.e.,  $f(v^{(i)}_X\circ 0_{\overline{X}})\not=f(v^{(i)}_{X\backslash X_i}\circ 0_{X_i}\circ 0_{\overline{X}})$. We will also show in Lemma~\ref{cloose} that, whp, $\Pr_\D[f(x_X\circ 0_{\overline{X}})\not=f(x)]\le \epsilon/c$.

\color{black}
We first prove
\begin{lemma}\label{dist} Consider steps~\ref{par11}-\ref{par21} in the {\bf ApproxTarget}. If $f$ is a $k$-junta then, with probability at least $2/3$, for each $i\in [r]$, the set $x(X_i)=\{x_j|j\in X_i\}$ contains at most one relevant variable of $f$.
\end{lemma}
\begin{proof} Let $x_{i_1}$ and $x_{i_2}$ be two relevant variables in $f$.
The probability that $x_{i_1}$ and $x_{i_2}$ are in the same set is equal to $1/r$.
By the union bound, it follows that the probability that some relevant variables $x_{i_1}$ and $x_{i_2}$, $i_1\not=i_2$, in $f$ are in the same set is at most ${k\choose 2}/r\le 1/3$.
\end{proof}

\begin{figure}[h!]
\begin{center}
\begin{tabular}{l | c | c || c | c |}
  & \multicolumn{2}{|c||}{$X^{(j)}=\bigcup_{i=1}^jX_{\ell_i}$}&\multicolumn{2}{|c|}{$\overline{X^{(j)}}$} \\
\cline{2-5}
  &$X^{(j)}\backslash X_r$ &$X_r$&$\bigcup_{i=j+1}^{q^{\ } }X_{\ell_i}$&$\overline{\overline{X^{(q)}}}$\\
\hline
$v^{(r)}$& *********** & ***** & 00000000000 & 000000\\
\hline
$w^{(r)}$& *********** & ***** & *********** & ******\\
\hline
$w^{(r)}_{X^{(j)}}\circ 0_{\overline{X^{(j)}}}=v^{(r)}$& *********** & ***** & 00000000000 & 000000\\
\hline
$w^{(r)}_{\overline{X_r}}\circ 0_{X_r}$&*********** & 00000 & *********** & ******\\
\hline
$w^{(r)}_{X^{(j)}\backslash X_r}\circ 0_{\overline{X^{(j)}}\cup X_r}$& *********** & 00000 & 00000000000 & 000000\\
\hline
\end{tabular}
\end{center}
	\caption{The value of $v^{(r)}$, $w^{(r)}$, $w^{(r)}_{X^{(j)}}\circ0_{\overline{X^{(j)}}}$, $w^{(r)}_{\overline{X_r}}\circ 0_{X_r}$ and $w^{(r)}_{X^{(j)}\backslash X_r}\circ 0_{\overline{X^{(j)}}\cup X_r}$ where * indicates any value.}
	\label{ValTable}
	\end{figure}

Recall that after the binary search in step~\ref{Find1} the procedure has a witness $w^{(\ell)}$ that satisfies $f(w^{(\ell)})\not= f(w^{(\ell)}_{\overline{X_\ell}}\circ 0_{X_\ell})$ that is not necessarily a witness of $f(x_X\circ 0_{\overline X})$ for $X_\ell$, i.e., does not necessarily satisfies $f(w^{(\ell)}_X\circ 0_{\overline{X}})\not= f(w^{(\ell)}_{X\backslash X_\ell}\circ 0_{X_\ell}\circ 0_{\overline{X}})$. This is why we first add $w^{(\ell)}$ to $W$ and not to $V$. We next will show that an element $w^{(r)}$ in $W$ is either a witness of $f(x_X\circ 0_{\overline X})$ for $X_\ell$, in which case we add it to $V$ and remove it from $W$, or, this element generates another new relevant set and then another witness of $f$ is added to $W$.

Suppose the variable $\ell$ in the procedure takes the values ${\ell_1},\ldots,{\ell_q}$. Then $X_\ell$ takes the values $X_{\ell_1},\ldots,X_{\ell_q}$ and $X$ takes the values $X^{(i)}$ where $X^{(i)}=X^{(i-1)}\cup X_{\ell_i}$ and $X^{(0)}=\emptyset$.
Notice that $X^{(0)}\subset X^{(1)} \subset \cdots \subset X^{(q)}$.

Suppose, at some iteration, the procedure chooses, in step \ref{ccch}, $w^{(r)}\in W$ where $r=\ell_i$. By step~\ref{Finddd1}, $f(w^{(r)})\not= f(w^{(r)}_{\overline{X_r}}\circ 0_{X_r})$. Suppose at this iteration $X=X^{(j)}$. Then $r\le j$, $X_{\ell_1},\ldots,X_{\ell_j}$ are the relevant sets that are discovered so far and $X_{\ell_{j+1}},\ldots,X_{\ell_q}\subseteq \overline{X^{(j)}}$. Since $w^{(r)}\in W$, by step~\ref{Xadd}, $X_r\subseteq X^{(j)}$. See the table in Figure~\ref{ValTable}. If in step~\ref{Tr}, $f(w^{(r)}_{X^{(j)}}\circ 0_{\overline{X^{(j)}}})\not=f(w^{(r)}_{X^{(j)}\backslash X_r}\circ 0_{\overline{X^{(j)}}\cup X_r})$ then $v^{(r)}= w^{(r)}_{X^{(j)}}\circ 0_{\overline{X^{(j)}}}$ is added to the set $V$. This is the only step that adds an element to $V$. Since $v^{(r)}= w^{(r)}_{X^{(j)}}\circ 0_{\overline{X^{(j)}}}$ and $X^{(j)}\subseteq X^{(q)}$ we have $v^{(r)}_{\overline{X^{(q)}}}=0$ and $f(v^{(r)})=f(w^{(r)}_{X^{(j)}}\circ 0_{\overline{X^{(j)}}})\not=f(w^{(r)}_{X^{(j)}\backslash X_r}\circ 0_{\overline{X^{(j)}}\cup X_r})=f(v^{(r)}_{\overline{X_r}} \circ 0_{X_r}).$

Therefore
\begin{lemma}\label{Witf0} If the procedure outputs $(X^{(q)},V,I)$ then for every $v^{(\ell)}\in V$ we have $v^{(\ell)}_{\overline{X^{(q)}}}=0$ and $f(v^{(\ell)})\not=f( v^{(\ell)}_{\overline{X_\ell}}\circ 0_{X_\ell})$. That is, $v^{(\ell)}\in V$ is a witness of $f(x_{X^{(q)}}\circ 0_{\overline{X^{(q)}}})$ for $X_\ell$.
\end{lemma}

We now show that if, in step~\ref{Tr}, $f(w^{(r)}_X\circ 0_{\overline{X}})=f(w^{(r)}_{X\backslash X_r}\circ 0_{\overline{X}\cup X_r})$ then the procedure finds a new relevant set.
\begin{lemma}
Consider step~\ref{Tr} in the procedure in the iteration where $X=X^{(j)}$. If $w^{(r)}$ is not a witness of $f(x_{X^{(j)}}\circ 0_{\overline{X^{(j)}}})$ for $X_r$, i.e., $f(w^{(r)}_{X^{(j)}}\circ 0_{\overline{X^{(j)}}})=f(w^{(r)}_{X^{(j)}\backslash X_r}\circ 0_{\overline{X^{(j)}}\cup X_r})$, then a new relevant set is found.
\end{lemma}
\begin{proof} See the table in Figure~\ref{ValTable} throughout the proof.
 Since by step ~\ref{Finddd1}, $f(w^{(r)})\not= f(w^{(r)}_{\overline{X_r}}\circ0_{X_r})$, then either $f(w^{(r)})\not=f(w^{(r)}_{X^{(j)}}\circ 0_{\overline{X^{(j)}}})$ or $f(w^{(r)}_{X^{(j)}\backslash X_r}\circ 0_{\overline{X^{(j)}}\cup X_r})\not= f(w^{(r)}_{\overline{X_r}}\circ0_{X_r})$. If $f(w^{(r)})\not=f(w^{(r)}_{X^{(j)}}\circ 0_{\overline{X^{(j)}}})$ then the procedure in step~\ref{Fir01} assign $u=w^{(r)}$ and goes to step~\ref{Find1} to find a relevant set in $\overline{X^{(j)}}$. Step~\ref{Find1} finds a new relevant set because $w^{(r)}$ and $w^{(r)}_{X^{(j)}}\circ 0_{\overline{X^{(j)}}}$ are equal on $X^{(j)}$. If $f(w^{(r)}_{X^{(j)}\backslash X_r}\circ 0_{\overline{X^{(j)}}\cup X_r})\not= f(w^{(r)}_{\overline{X_r}}\circ 0_{X_r})$ then the procedure in step~\ref{Fir02} assign $u=w^{(r)}_{\overline{X_r}}\circ0_{X_r}$ and goes to step~\ref{Find1} to find a relevant set in $\overline{X^{(j)}}$. Step~\ref{Find1} finds a new relevant set because $w^{(r)}_{X^{(j)}\backslash X_r}\circ 0_{\overline{X^{(j)}}\cup X_r}$ and $w^{(r)}_{\overline{X_r}}\circ 0_{X_r}$ are equal on $X^{(j)}$.
\end{proof}

Therefore, for every $w^{(r)}\in W$ the procedure either finds $v^{(r)}$ that satisfies the condition in Lemma~\ref{Witf0} or finds a new relevant set. If the number of relevant sets is greater than $k$, then the procedure rejects. This is because each relevant set contains a relevant variable, and the relevant sets are disjoint. So the function, in this case, is not $k$-junta and therefore not in $C$. If the number of relevant sets is less than or equal to $k$, then the algorithm eventually finds, for each $\ell\in I$, a witness $v^{(\ell)}$ of $f(x_X\circ 0_{\overline{X}})$ for $X_{i_\ell}$. This implies

\begin{lemma}\label{lklk} If {\bf ApproxTarget} does not reject then it outputs $(X=X^{(q)},V=\{v^{(\ell_1)},
\ldots,v^{(\ell_q)}\},I=\{\ell_1,\ldots,\ell_q\})$ that satisfies
\begin{enumerate}
\item\label{lklk1} $q=|I|\le k$.
\item\label{lklk2} For every $\ell\in I$, $v^{(\ell)}_{\overline{X}}=0$ and $f(v^{(\ell)})\not=f(0_{X_\ell}\circ v^{(\ell)}_{\overline{X_\ell}})$. That is, $v^{(\ell)}\in V$ is a witness of $f(x_{X}\circ 0_{\overline{X}})$ for $X_\ell$ .
\item\label{lklk3} Each $x(X_\ell)$, $\ell\in I$, contains at least one relevant variable of $f(x_X\circ 0_{\overline{X}})$.
\end{enumerate}
\end{lemma}

\begin{lemma}\label{kjun} If $f$ is $k$-junta and each $x(X_i)$ contains at most one relevant variable of $f$ then
\begin{enumerate}
\item {\bf ApproxTarget} outputs $(X,V,I)$.
\item Each $x(X_\ell)$, $\ell\in I$, contains exactly one relevant variable in $f(x_X\circ 0_{\overline{X}})$.
\item For every $\ell\in I$, $f(x_{X_\ell}\circ v_{\overline{X_\ell}}^{(\ell)})$ is a literal.
\end{enumerate}
\end{lemma}
\begin{proof} By {\it \ref{lklk3}} in Lemma~\ref{lklk}, $x(X_\ell)$, $\ell\in I$, contains exactly one relevant variable.  Thus, for every $\ell\in I$, $f(x_{X_\ell}\circ v_{\overline{X_\ell}}^{(\ell)})$ is a literal.

Since $f$ contains at most $k$ relevant variables, by Lemma~\ref{trivial01}, the number of relevant sets $|I|$ is at most $k$. Therefore,
{\bf ApproxTarget} does not halt in step~\ref{Rej1}.
\end{proof}

The following lemma shows that
\begin{lemma}\label{cloose}
If {\bf ApproxTarget} outputs $(X,V,I)$ then $|I|\le k$ and with probability at least $14/15$
$$\Pr_{u\in {\cal D}}[f(u_X\circ 0_{\overline{X}})\not= f(u)]\le \epsilon/c.$$
\end{lemma}
\begin{proof}
If $|I|>k$ then, from step~\ref{Rej1}, {\bf ApproxTarget} outputs ``reject''. Therefore, the probability that {\bf ApproxTarget} fails to output $(X,V,I)$ with $\Pr_{u\in {\cal D}}[f(u_X\circ 0_{\overline{X}})\not= f(u)]\le \epsilon/c$ is
the probability that for some $X^{(\ell)}$, $\Pr_{u\in {\cal D}}[f(x_{X^{(\ell)}}\circ 0_{\overline{X^{(\ell)}}})\not= f(x)]> \epsilon/c$ and $f(u_{X^{(\ell)}}\circ 0_{\overline{X^{(\ell)}}})= f(u)$ for $c\ln(15k)/\epsilon$ strings $u$ chosen independently at random according to the distribution $\D$.
This probability is at most $$k\left(1-\frac{c}{\epsilon}\right)^{c\ln(15k)/\epsilon}\le \frac{1}{15}.$$
\end{proof}

We now give the query complexity
\begin{lemma}\label{Query01} The procedure {\bf ApproxTarget} makes $O((k\log k)/\epsilon)$ queries.
\end{lemma}
\begin{proof} The condition in step~\ref{con11} requires two queries and is executed at most $M=ck\ln(15k)/\epsilon$ times. This is $2M=O((k\log k)/\epsilon)$ queries. Steps~\ref{Find1}-\ref{Fir02} are executed at most $k+1$ times. This is because each time it is executed, the value of $|I|$ is increased by one, and when $|I|=k+1$ the procedure rejects. By Lemma~\ref{BiSe}, to find a new relevant set the procedure makes $O(\log r)=O(\log k)$ queries. This gives another $O(k\log k)$ queries. Therefore, the query complexity is $O((k\log k)/\epsilon)$.
\end{proof}

\color{black}

\subsection{Testing the Relevant Sets}
In this subsection we give the procedure {\bf TestSets} that takes as an input $(X,V=\{v^{(\ell_1)},
\ldots,v^{(\ell_q)}\},I=\{\ell_1,\ldots,\ell_q\})$ and tests if for all $\ell\in I$, $f(x_{X_\ell}  \circ v^{(\ell)}_{\overline{X_\ell}})$ is $(1/30)$-close to some literal with respect to the uniform distribution.

\newcounter{ALC2}
\setcounter{ALC2}{0}
\newcommand{\stepb}{\stepcounter{ALC2}$\arabic{ALC2}.\ $\>}
\newcommand{\steplabelb}[1]{\addtocounter{ALC2}{-1}\refstepcounter{ALC2}\label{#1}}
\begin{figure}[h!]
  \begin{center}
  \fbox{\fbox{\begin{minipage}{28em}
  \begin{tabbing}
  xxx\=xxxx\=xxxx\=xxxx\=xxxx\=xxxx\= \kill
  {{\bf TestSets}$(X,V,I)$}\\
 {\it Input}: Oracle that accesses a Boolean function $f$ and $(X,V,I)$. \\
  {\it Output}: Either ``reject'' or ``OK''\\ \\
\stepb\steplabelb{LitT1}
For every $\ell\in I$ do\\
\stepb\steplabelb{Uni1}
\>  If {\bf UniformJunta}$(f(x_{X_\ell}\circ v^{(\ell)}_{\overline{X_\ell}}),1,1/30,1/15)$=``reject'' \\
\stepb\steplabelb{Rej211}
\>\>\>then Output(``reject'')\\
\stepb\steplabelb{ConB1}
\> Choose $b\in U$\\
\stepb\steplabelb{ConE11}
\> If $f(b_{X_\ell}\circ v^{(\ell)}_{\overline{X_\ell}})=f(\overline{b_{X_\ell}}\circ v^{(\ell)}_{\overline{X_\ell}})$ then Output(``reject'')\\
\stepb\
Return ``OK''
  \end{tabbing}
  \end{minipage}}}
  \end{center}
	\caption{A procedure that tests if for all $\ell\in I$, $f(x_{X_\ell}\circ v_{\overline{X_\ell}}^{(\ell)})$ is $(1/30)$-close to some literal with respect to the uniform distribution.}
	\label{A32}
	\end{figure}

We first prove

\begin{lemma}\label{kjun1} If $f$ is $k$-junta and each $x(X_i)$ contains at most one relevant variable of $f$ then
{\bf TestSets} returns ``OK''.
\end{lemma}
\begin{proof}
By Lemma~\ref{kjun}, for every $\ell\in I$, $f(x_{X_\ell}\circ v_{\overline{X_\ell}}^{(\ell)})$ is a literal.

If {\bf TestSets} rejects in step~\ref{Rej211} then, by Lemma~\ref{OneSide}, for some $X_\ell$, $\ell\in I$, $f(x_{X_\ell}\circ v^{(\ell)}_{\overline{X_\ell}})$ is not $1$-Junta (literal or constant function) and therefore $x(X_\ell)$ contains at least two relevant variables. If it rejects in step~\ref{ConE11}, then $f(b_{X_\ell}\circ v^{(\ell)}_{\overline{X_\ell}})=f(\overline{b_{X_\ell}}\circ v^{(\ell)}_{\overline{X_\ell}})$ and then $f(x_{X_\ell}\circ v^{(\ell)}_{\overline{X_\ell}})$ is not a literal. In all cases we get a contradiction.
\end{proof}

In the following lemma we show that if {\bf TestSets} returns ``OK'' then, whp, each $f(x_{X_\ell}\circ v_{\overline{X_\ell}}^{(\ell)})$ is close to a literal with respect to the uniform distribution.

\begin{lemma}\label{closelit} If for some $\ell\in I$, $f(x_{X_\ell}\circ v_{\overline{X_\ell}}^{(\ell)})$ is $(1/30)$-far from every literal with respect to the uniform distribution then, with probability at least $1-(1/15)$, {\bf TestSets} rejects.
\end{lemma}
\begin{proof} If $f(x_{X_\ell}\circ v_{\overline{X_\ell}}^{(\ell)})$ is $(1/30)$-far from every literal with respect to the uniform distribution then it is either (case 1) $(1/30)$-far from every $1$-Junta (literal or constant) or (case 2) $(1/30)$-far from every literal and $(1/30)$-close to $0$-Junta. In case 1, by Lemma~\ref{OneSide}, with probability at least $1-(1/15)$, ${\bf UniformJunta}$ $(f(x_{X_\ell}\circ v_{\overline{X_\ell}}^{(\ell)}),1,1/30,1/15)$ $=$ ``reject'' and then the procedure rejects. In case 2, if $f(x_{X_\ell}\circ v_{\overline{X_\ell}}^{(\ell)})$ is $(1/30)$-close to some $0$-Junta then it is either $(1/30)$-close to~$0$ or $(1/30)$-close to~$1$. Suppose it is $(1/30)$-close to $0$. Let $b$ be a random uniform string chosen in steps~\ref{ConB1}. Then $\overline{b}$ is random uniform and for $g(x)=f(x_{X_\ell}\circ v^{(\ell)}_{\overline{X_\ell}})$ we have
\begin{eqnarray*}
\Pr[\mbox{The procedure does not reject}]&=&
\Pr\left[g(b)\not=g(\overline{b})\right]\\
&=&\Pr[g(b)=1\wedge g(\overline{b})=0]+\Pr[g(b)=0\wedge g(\overline{b})=1]\\
&\le&\Pr[g(b)=1]+\Pr[g(\overline{b})=1]\\
&\le&\frac{1}{15}.
\end{eqnarray*}
\end{proof}

\begin{lemma}\label{Query02}
The procedure {\bf TestSets} makes $O(k)$ queries.
\end{lemma}
\begin{proof}
Steps~\ref{Uni1} and \ref{ConE11} are executed $|I|\le k$ times, and by Lemma~\ref{OneSide}, the total number of queries made is $O(1/(1/30)\log(15))k+2k=O(k)$.
\end{proof}

\subsection{Determining the Values of the Relevant Variables}
\newcounter{ALC3}
\setcounter{ALC3}{0}
\newcommand{\stepc}{\stepcounter{ALC3}$\arabic{ALC3}.\ $\>}
\newcommand{\steplabelc}[1]{\addtocounter{ALC3}{-1}\refstepcounter{ALC3}\label{#1}}
\begin{figure}[h!]
  \begin{center}
  \fbox{\fbox{\begin{minipage}{28em}
  \begin{tabbing}
  xxx\=xxxx\=xxxx\=xxxx\=xxxx\=xxxx\= \kill
  {{\bf RelVarValues}$(w,X,V,I,\delta)$}\\
 {\it Input}: Oracle that accesses a Boolean function $f$, $(X,V,I)$ and $w\in\{0,1\}^n$. \\
  {\it Output}: Either ``reject'' or for every $\ell\in I$, the value, $z_\ell=w_{\tau(\ell)}$ where $x_{\tau(\ell)}$ is one of the\\
\>\>  relevant variables of $f(x_X\circ 0_{\overline{X}})$ in $x(X_\ell)$\\ \\
  xxx\=xxxx\=xxxx\=xxxx\=xxxx\=xxxx\= \kill
\stepc\steplabelc{feld01}
  For every $\ell\in I$ do\\
\stepc  \>\>  For $\xi\in\{0,1\}$ set $Y_{\ell,\xi}=\{j\in X_\ell |  w_j=\xi\}$.\\
\stepc  \>\>  Set $G_{\ell,0}=G_{\ell,1}=0$;\\
\stepc\steplabelc{feld02}
\>\>  Repeat $h=\ln(k/\delta)/\ln(4/3)$ times\\
\stepc\steplabelc{feld03}
\>\>\> Choose $b\in U$; \\
\stepc\steplabelc{qqq1}\>\>\>\> If $f(b_{Y_{\ell,0}}\circ b_{Y_{\ell,1}}\circ v^{(\ell)}_{\overline{X_\ell}})\not= f(\overline{b_{Y_{\ell,0}}}\circ b_{Y_{\ell,1}}\circ v^{(\ell)}_{\overline{X_\ell}})$ then $G_{\ell,0}\gets G_{\ell,0}+1$\\
\stepc\steplabelc{qqq2}\>\>\>\> If $f(b_{Y_{\ell,1}}\circ b_{Y_{\ell,0}}\circ v^{(\ell)}_{\overline{X_\ell}})\not= f(\overline{b_{Y_{\ell,1}}}\circ b_{Y_{\ell,0}}\circ v^{(\ell)}_{\overline{X_\ell}})$ then $G_{\ell,1}\gets G_{\ell,1}+1$\\
\stepc\steplabelc{GGG5}
\>\>If ($\{G_{\ell,0},G_{\ell,1}\}\not=\{0,h\}$) then Output(``reject'')\\
\stepc\steplabelc{Gl}
\>\> If $G_{\ell,0}=h$ then $z_\ell \gets 0$ else $z_\ell \gets 1$\\
\stepc Output(``$\{z_\ell\}_{\ell\in I}$'')
  \end{tabbing}
  \end{minipage}}}
  \end{center}
	\caption{A procedure that takes as input $(X,V,I)$ and a string $w\in\{0,1\}^n$ and, with probability at least $1-\delta$, returns the values of $w_{\tau(i)}$, $i\in I$, where  $f(x_{X_i}\circ v_{\overline{X_i}}^{(i)})$ is $(1/30)$-close to one of the literals in $\{x_{\tau(i)},\overline{x_{\tau(i)}}\}$ with respect to the uniform distribution.}
	\label{A3}
	\end{figure}
\color{black}

In this subsection we give a procedure {\bf RelVarValue} that for an input $(w\in\{0,1\}^n,X,V,I,\delta)$ where $(X,V,I)$ satisfies all the properties in the previous two subsections, the procedure, with probability at least $1-\delta$, returns the values of $w_{\tau(i)}$, $i\in I$, where  $f(x_{X_i}\circ v_{\overline{X_i}}^{(i)})$ is $(1/30)$-close to one of the literals in $\{x_{\tau(i)},\overline{x_{\tau(i)}}\}$ with respect to the uniform distribution. When $f$ is $k$-junta and each $x(X_i)$ contains at most one relevant variable then $\{x_{\tau(i)}\}_{i\in I}$ is the set of the relevant variables of $f(x_X\circ 0_{\overline{X}})$ and $w_{\tau(i)}$, $i\in I$ are the values of the relevant variables. The procedure is in Figure~\ref{A3}.

We first prove
\begin{lemma}\label{kjun2} If $f$ is $k$-Junta and each $x(X_i)$ contains at most one relevant variable of $f$ then {\bf RelVarValues} outputs $z$ such that $z_\ell=w_{\tau(\ell)}$ where $f(x_{X_\ell}\circ 0_{\overline{X_\ell}})\in \{x_{\tau(\ell)},\overline{x_{\tau(\ell)}}\}$.
\end{lemma}
\begin{proof}
Since $Y_{\ell,0}, Y_{\ell,1}$ is a partition of $X_\ell$, $\ell\in I$ and, by Lemma~\ref{kjun}, $x(X_\ell)$ contains exactly one relevant variable $x_{\tau(\ell)}$ of $f(x_X\circ 0_{\overline{X}})$, this variable is either in $x(Y_{\ell,0})$ or in $x(Y_{\ell,1})$ but not in both. Suppose w.l.o.g. it is in $x(Y_{\ell,0})$ and not in $x(Y_{\ell,1})$. Then $w_{\tau(\ell)}=0$, $f(x_{Y_{\ell,0}}\circ b_{Y_{\ell,1}}\circ v^{(\ell)}_{\overline{X_\ell}})$ is a literal and $f(x_{Y_{\ell,1}}\circ b_{Y_{\ell,0}}\circ v^{(\ell)}_{\overline{X_\ell}})$ is a constant function. This implies that for any $b$, $f(b_{Y_{\ell,0}}\circ b_{Y_{\ell,1}}\circ v^{(\ell)}_{\overline{X_\ell}})\not= f(\overline{b_{Y_{\ell,0}}}\circ b_{Y_{\ell,1}}\circ v^{(\ell)}_{\overline{X_\ell}})$ and $f(b_{Y_{\ell,1}}\circ b_{Y_{\ell,0}}\circ v^{(\ell)}_{\overline{X_\ell}})= f(\overline{b_{Y_{\ell,1}}}\circ b_{Y_{\ell,0}}\circ v^{(\ell)}_{\overline{X_\ell}})$. Therefore, by steps~\ref{qqq1}-\ref{qqq2} in the procedure, $G_{\ell,0}=h$ and $G_{\ell,1}=0$ and the procedure does not output reject in step~\ref{GGG5}. Thus, by step~\ref{Gl}, $z_\ell=w_{\tau(\ell)}$.
\end{proof}

We now prove
\begin{lemma}\label{FinTest} If for every $\ell\in I$ the function $f(x_{X_\ell}\circ v^{(\ell)}_{\overline{X_\ell}})$ is $(1/30)$-close to a literal in $\{x_{\tau(\ell)},\bar{x}_{\tau(\ell)}\}$ with respect to the uniform distribution, where $\tau(\ell)\in X_\ell$,
and in {\bf RelVarValues}, for every $\ell\in I$, $\{G_{\ell,0},G_{\ell,1}\}=\{0,h\}$ then, with probability at least $1-\delta$, we have:
For every $\ell\in I$,
$z_\ell=w_{\tau(\ell)}$.
\end{lemma}
\begin{proof} Fix some $\ell$. Suppose $f(x_{X_\ell}\circ v^{(\ell)}_{\overline{X_\ell}})$ is $(1/30)$-close to $x_{\tau(\ell)}$ with respect to the uniform distribution. The case when it is $(1/30)$-close to $\overline{{x}_{\tau(\ell)}}$ is similar. Since $X_\ell=Y_{\ell,0}\cup Y_{\ell,1}$ and $Y_{\ell,0}\cap Y_{\ell,1}=\emptyset$ we have that $\tau(\ell)\in Y_{\ell,0}$ or $\tau(\ell)\in Y_{\ell,1}$, but not both. Suppose $\tau(\ell)\in Y_{\ell,0}$. The case where $\tau(\ell)\in Y_{\ell,1}$ is similar. Define the random variable $Z(x_{X_\ell})=1$ if $f(x_{X_\ell}\circ v^{(\ell)}_{\overline{X_\ell}})\not=x_{\tau(\ell)}$ and $Z(x_{X_\ell})=0$ otherwise. Then
$$\E_{x_{X_\ell}\in U}[Z(x_{X_\ell})]\le \frac{1}{30}.$$ Therefore
$$\E_{x_{Y_{\ell,1}}\in U}\E_{x_{Y_{\ell,0}}\in U}[Z(x_{Y_{\ell,0}}\circ x_{Y_{\ell,1}})]\le \frac{1}{30}$$ and by Markov's bound
$$\Pr_{x_{Y_{\ell,1}}\in U}\left[ \E_{x_{Y_{\ell,0}}\in U}[Z(x_{Y_{\ell,0}}\circ x_{Y_{\ell,1}})]\ge \frac{2}{15}\right]\le \frac{1}{4}.$$
That is, for a random uniform string $b\in \{0,1\}^n$, with probability at least $3/4$, $f(x_{Y_{\ell,0}}\circ b_{Y_{\ell,1}}\circ v^{(\ell)}_{\overline{X_\ell}})$ is $(2/15)$-close to $x_{\tau(\ell)}$ with respect to the uniform distribution. Now, given that $f(x_{Y_{\ell,0}}\circ b_{Y_{\ell,1}}\circ v^{(\ell)}_{\overline{X_\ell}})$ is $(2/15)$-close to $x_{\tau(\ell)}$ with respect to the uniform distribution the probability that $G_{\ell,0}=0$ is the probability that $f(b_{Y_{\ell,0}}\circ b_{Y_{\ell,1}}\circ v^{(\ell)}_{\overline{X_\ell}})= f(\overline{b_{Y_{\ell,0}}}\circ b_{Y_{\ell,1}}\circ v^{(\ell)}_{\overline{X_\ell}})$ for $h$ random uniform strings $b\in \{0,1\}^n$. Let $b^{(1)},\ldots,b^{(h)}$ be $h$ random uniform strings in $\{0,1\}^n$, $V(b)$ be the event $f(b_{Y_{\ell,0}}\circ b_{Y_{\ell,1}}\circ v^{(\ell)}_{\overline{X_\ell}})= f(\overline{b_{Y_{\ell,0}}}\circ b_{Y_{\ell,1}}\circ v^{(\ell)}_{\overline{X_\ell}})$ and $A$ the event that $f(x_{Y_{\ell,0}}\circ b_{Y_{\ell,1}}\circ v^{(\ell)}_{\overline{X_\ell}})$ is $(2/15)$-close to $x_{\tau(\ell)}$ with respect to the uniform distribution. Let $g(x_{Y_{\ell,0}})=f(x_{Y_{\ell,0}}\circ b_{Y_{\ell,1}}\circ v^{(\ell)}_{\overline{X_\ell}})$. Then
\begin{eqnarray*}
\Pr[V(b)|A]&=& \Pr[g(b_{Y_{\ell,0}})=g(\overline{b_{Y_{\ell,0}}})|A]\\
&=&\Pr[(g(b_{Y_{\ell,0}})=b_{\tau(\ell)}\wedge g(\overline{b_{Y_{\ell,0}}})=b_{\tau(\ell)})
\vee (g(b_{Y_{\ell,0}})=\overline{b_{\tau(\ell)}}\wedge g(\overline{b_{Y_{\ell,0}}})=\overline{b_{\tau(\ell)}})|A]\\
&\le& \Pr[g(\overline{b_{Y_{\ell,0}}})\not=\overline{b_{\tau(\ell)}}
\vee g({b_{Y_{\ell,0}}})\not={b_{\tau(\ell)}})|A]\\
&\le& \Pr[g(\overline{b_{Y_{\ell,0}}})\not=\overline{b_{\tau(\ell)}}|A]
+\Pr[g({b_{Y_{\ell,0}}})\not={b_{\tau(\ell)}})|A]\le \frac{4}{15}.
\end{eqnarray*}
Since $\tau(\ell)\in Y_{\ell,0}$, we have $w_{\tau(\ell)}=0$. Therefore, by step~\ref{Gl} and since $\tau(\ell)\in X_\ell$,
\begin{eqnarray*}
\Pr[z_\ell\not= w_{\tau(\ell)}]&=&\Pr[z_{\ell}=1]\\
&=&\Pr[G_{\ell,0}=0\wedge G_{\ell,1}=h]\\
&\le&\Pr[G_{\ell,0}=0]=\Pr[(\forall j\in[h]) V(b^{(j)})]\\
&=& (\Pr[V(b)])^h
\le \left( \Pr[V(b)|A]+\Pr[\overline{A}]\right)^h
\le (4/15+1/4)^h\le (3/4)^h
\end{eqnarray*}
Therefore, the probability that $z_{\ell}\not =w_{\tau(\ell)}$ for some $\ell\in I$ is at most $k(3/4)^h\le \delta$.
\end{proof}

The following is obvious
\begin{lemma}\label{kkkkkk}
The procedure {\bf RelVarValues} makes $O(k\log(k/\delta))$ queries.
\end{lemma}

\section{Testing Subclasses of $k$-Junta}
In this section, we give testers for subclasses of $k$-Junta that are closed under variable and zero projections.

Our tester will start by running the two procedures {\bf ApproxTarget} and {\bf TestSets} and therefore, by Lemmas~\ref{dist}, \ref{kjun} and~\ref{kjun1}, if $f\in C$ (and therefore is $k$-junta) then, with probability at least $2/3$, both procedures do not reject and item {\it \ref{f01}} in the following Assumption happens. By Lemmas~\ref{lklk}, \ref{cloose}, and \ref{closelit}, if $f$ is $\epsilon$-far from every function in $C$ and both procedures do not reject then, with probability at least $13/15$, item {\it \ref{f02}} in the following Assumption happens. Obviously, the above two probabilities can be changed to $1-\delta$ for any constant $\delta$ without changing the asymptotic query complexity.

\begin{assm}\label{assmp}
Throughout this section we assume that there are $X$, $q\le k$, $I=\{\ell_1,\ldots,\ell_q\}$ and $V=\{v^{(\ell_1)},\ldots,v^{(\ell_q)}\}$ such that: For every $\ell\in I$, $v^{(\ell)}_{\overline{X}}=0$ and $f(v^{(\ell)})\not=f(0_{X_\ell}\circ v^{(\ell)}_{\overline{X_\ell}})$. That is, $v^{(\ell)}\in V$ is a witness of $f(x_{X}\circ 0_{\overline{X}})$ for $X_\ell$ and
\begin{enumerate}

\item\label{f01}  If $f\in C$ (and therefore is $k$-junta)
\begin{itemize}
\item $f(x_X\circ 0_{\overline{X}})\in C$.
\item  Each $x(X_\ell)$, $\ell\in I$ contains exactly one relevant variable.
\item For every $\ell\in I$, $f(x_{X_\ell}\circ v^{(\ell)}_{\overline{X_\ell}})$ is a literal in $\{x_{\tau(\ell)},\overline{x_{\tau(\ell)}}\}$.
\end{itemize}
\item\label{f02}  If $f$ is $\epsilon$-far from every function in $C$ then
\begin{itemize}
\item $f(x_X\circ 0_{\overline{X}})$ is $(\epsilon/3)$-close to $f$ with respect to $\D$ and therefore $f(x_X\circ 0_{\overline{X}})$ is $(2\epsilon/3)$-far from every function in $C$ with respect to $\D$.
\item  For every $\ell\in I$, $f(x_{X_\ell}\circ v^{(\ell)}_{\overline{X_\ell}})$ is $(1/30)$-close to a literal in $\{x_{\tau(\ell)},\bar{x}_{\tau(\ell)}\}$ with respect to the uniform distribution.

\end{itemize}
\end{enumerate}
We will also use the set of indices $\Gamma:=\{\tau(\ell_1),\ldots,\tau(\ell_q)\}$. Notice that if $f$ is $k$-junta then $x(\Gamma)$ are the relevant variables of $f$.
\end{assm}

We remind the reader that for a projection $\pi:X\to X$ the string $x(\pi)$ is defined as $x(\pi)_j=x_{\pi(j)}$ for every $j\in X$. Define the projection $\pi_{f,I}:X\to X$ that satisfies: For every $\ell\in I$ and every $j\in X_\ell$, $\pi_{f,I}(j)=\tau(\ell)$.
Define the function $F(x_\Gamma)=F(x_{\tau(\ell_1)},\ldots,x_{\tau(\ell_q)}):=f(x(\pi_{f,I})\circ 0_{\overline{X}})$. That is, $F$ is the function that results from substituting in $f(x_X\circ 0_{\overline{X}})$ for every $\ell\in I$ and every $x_i$, $i\in X_\ell$, the variable $x_{\tau(\ell)}$.
Note here that the tester does not know $\tau(\ell_1),\ldots,\tau(\ell_q)$.

We now show how to query $F$ by querying $f$
\begin{lemma}\label{SimulateF} For the function $F$ we have
\begin{enumerate}
\item \label{Compp1}Given $(y_1,\ldots,y_q)$, computing $F(y_1,\ldots,y_q)$ can be done with one query to $f$.
\item \label{Compp3}Given $x\in\{0,1\}^n$ and $\delta$, there is an algorithm that makes $O(k\log (k/\delta))$ queries and, with probability at least $1-\delta$, either discovers that some $X_i$, $i\in I$ contains at least two relevant variables in $f$ (and therefore, whp,  $f$ is not $k$-junta) and then rejects or computes $z=(x_{\tau(\ell_1)},\ldots,x_{\tau(\ell_q)})$ and $F(z)$.
\end{enumerate}
\end{lemma}
\begin{proof}  {\it \ref{Compp1}} is immediate. To prove {\it \ref{Compp3}} we use Lemma~\ref{FinTest}. We run {\bf RelVarValues}$(x,X,V,I,\delta)$. If it rejects then $\{G_{\ell,0},G_{\ell,1}\}\not= \{0,h\}$ for some $\ell\in I$ and therefore $G_{\ell,0},G_{\ell,1}>0$. This implies that for some $b,b'\in\{0,1\}^n$, $f(b_{Y_{\ell,0}}\circ b_{Y_{\ell,1}}\circ v^{(\ell)}_{\overline{X_\ell}})\not= f(\overline{b_{Y_{\ell,0}}}\circ b_{Y_{\ell,1}}\circ v^{(\ell)}_{\overline{X_\ell}})$ and $f(b'_{Y_{\ell,1}}\circ b'_{Y_{\ell,0}}\circ v^{(\ell)}_{\overline{X_\ell}})\not= f(\overline{b'_{Y_{\ell,1}}}\circ b'_{Y_{\ell,0}}\circ v^{(\ell)}_{\overline{X_\ell}})$. Since $X_\ell=Y_{\ell,0}\cup Y_{\ell,1}$ and $Y_{\ell,0}\cap Y_{\ell,1}=\emptyset$, the set $x(X_\ell)$ contains at least two relevant variables in $f$.

If for every $\ell$ we have $\{G_{\ell,0},G_{\ell,1}\}= \{0,h\}$ then, by Lemma~\ref{FinTest}, with probability at least $1-\delta$, the procedure outputs $z$ where for every $\ell$, $z_\ell=x_{\tau(\ell)}$. Then using {\it \ref{Compp1}} we compute $F(z)$. Since by Lemma~\ref{kkkkkk}, {\bf RelVarValue} makes $O(k\log (k/\delta))$ queries, the result follows.
\end{proof}

We now give the key lemma for the first tester
\begin{lemma}\label{kLem} Let $C\subseteq k-$Junta be a class that is closed under variable and zero projections and $f$ be any Boolean function. Let $F(x_\Gamma)=f(x(\pi_{f,I})\circ 0_{\overline{X}})$ where $\Gamma=\{\tau(\ell)|\ell\in I\}$ and $C(\Gamma)$ be the set of all functions in $C$ that their relevant variables are $x(\Gamma)$. If Assumption~\ref{assmp} is true, then
\begin{enumerate}
\item\label{proj01} If $f\in C$ then $f(x_X\circ 0_{\overline{X}})=F\in C(\Gamma)$.
\item\label{proj02} If $f$ is $\epsilon$-far from every function in $C$ with respect to ${\cal D}$ then either
    \begin{enumerate}
    \item\label{proj02a} $f(x_X\circ 0_{\overline{X}})$ is $(\epsilon/3)$-far from $F$ with respect to ${\cal D}$,

        or
    \item\label{proj02b} $F$ is $(\epsilon/3)$-far from every function in $C(\Gamma)$ with respect to ${\cal D}$.
    \end{enumerate}
\end{enumerate}
\end{lemma}
\begin{proof} We first prove {\it \ref{proj01}}. If $f\in C$, then since $C$ is closed under variable and zero projection $f(x_X\circ 0_{\overline{X}})\in C$. We have $f(x_X\circ 0_{\overline{X}})=f(\circ_{\ell\in I}x_{X_\ell}\circ 0_{\overline{X}})$ and, by Assumption~\ref{assmp}, every $x(X_\ell)$, $\ell\in I$, contains exactly one relevant variable $x_{\tau(\ell)}$ of $f(x_X\circ 0_{\overline{X}})$. Therefore, $f(x_X\circ 0_{\overline{X}})$ is a function that depends only on the variables $x_{\tau(\ell)}$, $\ell\in I$. By the definition of $x(\pi_{f,{I}})$ and since $\tau(\ell)\in X_\ell$ we have $x(\pi_{f,{I}})_{\tau(\ell)}=x_{\tau(\ell)}$ and therefore $f(x_X\circ 0_{\overline{X}})=f(x(\pi_{f,{I}})\circ 0_{\overline{X}})=F$.

We now prove~{\it\ref{proj02}}. Suppose, for the contrary, $f(x_X\circ 0_{\overline{X}})$ is $(\epsilon/3)$-close to $F$ with respect to $\D$ and $F$ is $(\epsilon/3)$-close to some function $g\in C(\Gamma)$ with respect to $\D$. Then $f(x_X\circ 0_{\overline{X}})$ is $(2\epsilon/3)$-close to $g$ with respect to $\D$. Since, by Assumption~\ref{assmp}, $f(x_X\circ 0_{\overline{X}})$ is $(\epsilon/3)$-close to $f$ with respect to $\D$ we get that $f$ is $\epsilon$-close to $g\in C$ with respect to $\D$. A contradiction.

\end{proof}

In the following two subsections we discuss how to test the closeness of $f(x_X\circ 0_{\overline{X}})$ to $F$ and $F$ to $C(\Gamma)$. We will assume all the procedures in the following subsections have access to $X,V,I$ that satisfies Assumption~\ref{assmp}.

\subsection{Testing the Closeness of $f(x_X\circ 0_{\overline{X}})$ to $F$}
\newcounter{ALCc1}
\setcounter{ALCc1}{0}
\newcommand{\stepco}{\stepcounter{ALCc1}$\arabic{ALCc1}.\ $\>}
\newcommand{\steplabelco}[1]{\addtocounter{ALCc1}{-1}\refstepcounter{ALCc1}\label{#1}}
\begin{figure}[h!]
  \begin{center}
  \fbox{\fbox{\begin{minipage}{28em}
  \begin{tabbing}
  xxx\=xxxx\=xxxx\=xxxx\=xxxx\=xxxx\= \kill
  {{\bf Close$fF$}$(f,\D,\epsilon,\delta)$}\\
 {\it Input}: Oracle that accesses a Boolean function $f$ and $\D$. \\
  {\it Output}: Either ``reject'' or ``OK''\\ \\
  xxx\=xxxx\=xxxx\=xxxx\=xxxx\=xxxx\= \kill
  \stepco Define $F\equiv f(x(\pi_{f,I})\circ 0_{\overline{X}})$.\\
  \stepco\steplabelco{co01}
  Repeat $t=(3/\epsilon)\ln(2/\delta)$ times \\
  \stepco\steplabelco{co02}\> Choose $u\in \D$.\\
  \stepco\steplabelco{co03}\> $z\gets${\bf RelVarValue}$(u,X,V,I,\delta/(2t))$ .\\
  \stepco\steplabelco{co04}\> If $f(u_X\circ 0_{\overline{X}})\not=F(z)$ then Output(``reject'')\\
  \stepco\steplabelco{co05} Return ``OK''.
  \end{tabbing}
  \end{minipage}}}
  \end{center}
	\caption{A procedure that tests whether $f(x_X\circ 0_{\overline{X}})$ is $(\epsilon/3)$-far from $F$ with respect to ${\cal D}$.}
	\label{Close01}
	\end{figure}

We now give the procedure {\bf Close$fF$} that tests whether $f(x_X\circ 0_{\overline{X}})$ is $(\epsilon/3)$-far from $F$ with respect to ${\cal D}$. See Figure~\ref{Close01}.

\begin{lemma}\label{FirstApp} For any $\epsilon$, a constant $\delta$, and $(X,V,I)$ that satisfies Assumption~\ref{assmp}, procedure {\bf Close$fF$} makes $O((k/\epsilon)\log (k/\epsilon))$ queries and
\begin{enumerate}
\item\label{FirstApp1} If $f\in C$ then {\bf Close$fF$} returns OK.
\item\label{FirstApp2} If $f(x_X\circ 0_{\overline{X}})$ is $(\epsilon/3)$-far from $F$ with respect to $\D$ then, with probability at least $1-\delta$, {\bf Close$fF$} rejects.
\end{enumerate}
\end{lemma}
\begin{proof} {\bf Close$fF$} draws $t=(3/\epsilon)\ln(2/\delta)$ random $u^{(i)}\in\{0,1\}^n$, $i=1,\ldots,t$ according to the distribution~${\cal D}$. It finds $z^{(i)}=u^{(i)}_\Gamma$ and if $F(u^{(i)}_\Gamma)=f(u^{(i)}_X\circ 0_{\overline{X}})$ for all $i$ then it returns ``OK''. Otherwise it rejects.

If $f\in C$ then, by {\it \ref{proj01}} in Lemma~\ref{kLem}, $F(u^{(i)}_\Gamma)=f(u^{(i)}_X\circ 0_{\overline{X}})$ for every $i$. By Lemma~\ref{kjun2} and Assumption~\ref{assmp}, $z^{(i)}= u^{(i)}_\Gamma$ for all $i$, and therefore {\bf Close$fF$} returns OK.

Suppose now $f(x_X\circ 0_{\overline{X}})$ is $(\epsilon/3)$-far from $F$ with respect to $\D$. By {\it\ref{Compp3}} in Lemma~\ref{SimulateF}, {\bf RelVarValue} makes $O(k\log((kt)/\delta))$ queries and computes $F(u_\Gamma^{(i)})$, $i=1,\ldots,t$, with failure probability at most $\delta/2$. Then the probability that it fails to reject is at most
$(1-\epsilon/3)^t\le \delta/2$. This gives the result.

Therefore, {\bf Close$fF$} makes $O((k/\epsilon)\log(k/\epsilon))$ queries and satisfies~{\it \ref{FirstApp1}} and~{\it \ref{FirstApp2}}.
\end{proof}

\subsection{Testing the Closeness of $F$ to $C(\Gamma)$}
\newcounter{ALC8}
\setcounter{ALC8}{0}
\newcommand{\steps}{\stepcounter{ALC8}$\arabic{ALC8}.\ $\>}
\newcommand{\steplabels}[1]{\addtocounter{ALC8}{-1}\refstepcounter{ALC8}\label{#1}}
\begin{figure}[h!]
  \begin{center}
  \fbox{\fbox{\begin{minipage}{28em}
  \begin{tabbing}
  xxx\=xxxx\=xxxx\=xxxx\=xxxx\=xxxx\= \kill
  {{\bf Close$FC\D$}$(f,\D,\epsilon,\delta)$}\\
 {\it Input}: Oracles that access a Boolean function $f$ and $\D$. \\
  {\it Output}: Either ``reject'' or ``OK''\\ \\
  xxx\=xxxx\=xxxx\=xxxx\=xxxx\=xxxx\= \kill
  \steps\steplabels{FFs07}$C^*\gets C(\Gamma)$\\
  \steps\steplabels{FFs08}Repeat $\tau=(12/\epsilon)\ln(2|C^*|/\delta)$ times\\
  \steps\steplabels{FFs09}\> Choose $u\in \D$.\\
  \steps\steplabels{FFs10}\> $z\gets${\bf RelVarValue}$(u,X,V,I,1/2)$ .\\
  \steps\steplabels{FFs11}\> For every $g\in C^*$\\
  \steps\steplabels{FFs12}\>\>If $g(z)\not=F(z)$ then $C^*\gets C^*\backslash \{g\}$.\\
  \steps\steplabels{FFs13}\>If $C^*=\emptyset$ then Output(``Reject'')\\
  \steps\steplabels{FFs14} Return ``OK''\\
  \addtocounter{ALC8}{-8}
  ----------------------------------------------------------------------\\
  {{\bf Close$FCU$}$(f,\epsilon,\delta)$}\\
 {\it Input}: Oracle that accesses a Boolean function $f$. \\
  {\it Output}: Either ``reject'' or ``OK''\\ \\
  \steps\steplabels{FFss07}$C^*\gets C(\Gamma)$\\
  \steps\steplabels{FFs15}Repeat $\tau=(3/\epsilon)\ln(2|C^*|/\delta))$ times\\
  \steps\steplabels{FFs16}\> Choose $(z_1,\ldots,z_q)\in U$.\\
  \steps\steplabels{FFs17}\> For every $g\in C^*$\\
  \steps\steplabels{FFs18}\>\>If $g(z)\not=F(z)$ then $C^*\gets C^*\backslash \{g\}$.\\
  \steps\steplabels{FFs19}\>If $C^*=\emptyset$ then Output(``Reject'')\\
  \steps\steplabels{FFs20} Return ``OK''
  \end{tabbing}
  \end{minipage}}}
  \end{center}
	\caption{Two procedures that test whether $F$ is $(\epsilon/3)$-far from every function in $C(\Gamma)$ with respect to ${\cal D}$ and the uniform distribution, respectively.}
	\label{Close02}
	\end{figure}
In this section, we give the procedures {\bf Close$FC\D$} and {\bf Close$FCU$} that test whether $F$ is $(\epsilon/3)$-far from every function in $C(\Gamma)$ with respect to ${\cal D}$ and the uniform distribution, respectively.
We prove

\begin{lemma}\label{clooo} For any $\epsilon$ and any constant $\delta$  and $(X,V,I)$ that satisfies Assumption~\ref{assmp}, the procedures {\bf Close$FC\D$} and {\bf Close$FCU$} make $O((k\log|C(\Gamma)|)/\epsilon)$ and $O((\log|C(\Gamma)|)/\epsilon)$ queries to $f$, respectively, and
\begin{enumerate}
\item If $f\in C$ then {\bf Close$FC\D$} and {\bf Close$FCU$} output OK.
\item If $F$ is $(\epsilon/3)$-far from every function in $C(\Gamma)$ with respect to ${\cal D}$ then, with probability at least $1-\delta$, {\bf Close$FC\D$} rejects.
\item If $F$ is $(\epsilon/3)$-far from every function in $C(\Gamma)$ with respect to the uniform distribution, then with probability at least $1-\delta$, {\bf Close$FCU$} rejects.
\end{enumerate}
Both procedures run in time $poly(n,|C(\Gamma)|,1/\epsilon)$.
\end{lemma}
\begin{proof} The proof for {\bf Close$FCU$} is similar to the proof of Lemma~\ref{FirstApp} with union bound.

For {\bf Close$FC\D$}, notice that it calls {\bf RelVarValue}$(u,X,V,I,1/2)$, and therefore at each iteration, with probability $1/2$, $z=u_\Gamma$. By Chernoff's bound ((\ref{Chernoff2}) in Lemma~\ref{Chernoff}), with probability at least $1-\delta/2$, $(3/\epsilon)\ln(2|C^*|/\delta)$ of the chosen $u$s in the procedures satisfy $z=u_\Gamma$. Then again, by union bound, the result follows.
\end{proof}

\subsection{Testing the Closeness of $F$ to $C(\Gamma)$ via Learning $C(\Gamma)$}
In this subsection, we show how proper learning implies testing the closeness of $F$ to $C(\Gamma)$. The proofs are similar to the proof of Proposition 3.1.1 in~\cite{GoldreichGR98}.

Let $(X,V,I)$ be as in~Assumption~\ref{assmp} and $q=|I|\le k$. Let $Y=\{y_1,\ldots,y_q\}$ be a set of Boolean variables and $C(Y)$ be the set of all functions in $C$ that depend on all the variables of $Y$. Notice that instead of using $C(Y)$ we could have used $C(\{x_1,\ldots,x_q\})$ but here we use the new Boolean variables $y_i$ to avoid confusion with the variables $x_i$ of $f$.

\begin{remark}\label{remark}
In all the lemmas in this subsection and the following one, in addition to the fact that $F$ depends on all the variables of $Y$, the learning algorithms can also make use of $(X,V,I)$ that satisfies Assumption~\ref{assmp}. This may help for some classes. For example\footnote{See the definition of unate in Subsection~\ref{TMDD}}, if the target function is a unate monotone function, then from the witnesses in $V$, we can know if $F$ is positive or negative unate in $y_i$, for each variable $y_i$.
\end{remark}

The following is an immediate result that follows from the two procedures {\bf Close$FC\D$} and {\bf Close$FCU$} in the previous subsection
\begin{lemma}\label{LtoTTriv} If there is a polynomial time algorithm that given a set $${\cal Y}=\{(y^{(1)},\xi_1),\ldots,(y^{(t)},\xi_t)\}\subseteq \{0,1\}^q\times \{0,1\}$$ decides whether there is a function $F\in C(Y)$ that is consistent with ${\cal Y}$, i.e., $F(y^{(i)})=\xi_i$ for all $i=1,\ldots,t$, then there is a polynomial time algorithm ${\cal B}_\D$ (resp. ${\cal B}_U$) that makes
$O((k\log|C(\Gamma)|)/\epsilon)$ queries (resp. $O((\log|C(\Gamma)|)/\epsilon)$ queries) to $f$ and
\begin{enumerate}
\item If $f\in C$ then ${\cal B}_\D$ and ${\cal B}_U$ output OK.
\item If $F$ is $(\epsilon/3)$-far from every function in $C(\Gamma)$ with respect to ${\cal D}$ then, with probability at least $1-\delta$, ${\cal B}_\D$ rejects.
\item If $F$ is $(\epsilon/3)$-far from every function in $C(\Gamma)$ with respect to the uniform distribution then, with probability at least $1-\delta$, ${\cal B}_U$ rejects.
\end{enumerate}
\end{lemma}

We now give the reduction from exact learning
\begin{lemma}\label{LtoT01} If there is a polynomial time algorithm ${\cal A}$
that, given as an input any constant $\delta$, properly exactly learns $ C(Y)$ with confidence parameter $\delta$ and makes $M(\delta)$ membership queries to $F$ then there is a polynomial time algorithm ${\cal B}$ that, given as an input $\epsilon$, any constant $\delta$, makes $M(\delta/3)+O((k/\epsilon)\log(1/\epsilon))$ (resp. $M(\delta/3)+O(1/\epsilon)$) queries to $f$ and
\begin{enumerate}
\item If $f\in C$ then, with probability at least $1-\delta$, ${\cal B}$ outputs OK.
\item If $F$ is $(\epsilon/3)$-far from every function in $C(\Gamma)$ with respect to ${\cal D}$ (resp. with respect to the uniform distribution) then, with probability at least $1-\delta$, ${\cal B}$ rejects.
\end{enumerate}
\end{lemma}
\begin{proof} Algorithm ${\cal B}$ runs ${\cal A}$ with confidence parameter $\delta/3$ to learn $F(y_1,\ldots,y_q)$. By Lemma~\ref{SimulateF}, each membership query to $F$ can be simulated by one membership query to $f$. If algorithm ${\cal A}$ runs more than it should, asks more than $M(\delta/3)$ membership queries or outputs $h\not\in C(Y)$ then ${\cal B}$ rejects.

If ${\cal A}$ outputs $h\in C(Y)$ then the algorithm needs to distinguish whether $h(x_\Gamma)$ is equal to $F(x_\Gamma)$ or $(\epsilon/3)$-far from $F(x_\Gamma)$ with respect to the distribution $\D$ (resp. uniform distribution). When the distribution is uniform, the algorithm chooses $t=(3/\epsilon)\ln(3/\delta)$ strings $v^{(1)},\ldots,v^{(t)}\in \{0,1\}^q$ and if $F(v^{(i)})=h(v^{(i)})$ for all $i$ then it outputs ``OK''; otherwise it rejects.

In the distribution-free model, ${\cal B}$ chooses $t=(12/\epsilon)\ln(2/\delta)$ strings $u^{(i)}\in \{0,1\}^n$ according to the distribution $\D$. Then runs {\bf RelValValue}$(u^{(i)},X,V,I,1/2)$ to find, with probability $1/2$  the value of $u^{(i)}_\Gamma$, $i=1,\ldots,t$. If $F(u^{(i)}_\Gamma)=h(u^{(i)}_\Gamma)$ for all $i$, then it outputs ``OK''; otherwise it rejects.

The analysis and correctness of the algorithm are the same as in the above proofs and Proposition~3.1.1 in~\cite{GoldreichGR98}.
\end{proof}

We now give the reduction from learning from MQ and ExQ$_\D$

\begin{lemma}\label{LtoT02} If there is a polynomial time algorithm ${\cal A}$
that, given as an input a constant $\delta$, any $\epsilon$, learns $C(Y)$ with respect to the distribution $\D$ (resp. uniform distribution), with confident $\delta$, accuracy $\epsilon$, makes $M(\epsilon,\delta)$ $MQ$ to $F$ and $Q(\epsilon,\delta)$ ExQ$_\D$ (resp. ExQ$_U$) queries to $F$ then there is a polynomial time algorithm ${\cal B}_\D$ (resp. ${\cal B}_U$) that asks
$$O\left(M(\epsilon/12,\delta/3)+kQ(\epsilon/12,\delta/3)\log ({kQ(\epsilon/12,\delta/3)})+\frac{k}{\epsilon}\log\frac{1}{\epsilon}\right)$$ queries (resp. $$O\left(M(\epsilon/12,\delta/3)+Q(\epsilon/12,\delta/3)+\frac{1}{\epsilon}\right)$$ queries) to $f$ and
\begin{enumerate}
\item If $f\in C$ then with probability at least $1-\delta$, ${\cal B}_\D$ and ${\cal B}_U$ output OK.
\item If $F$ is $(\epsilon/3)$-far from every function in $C(\Gamma)$ with respect to ${\cal D}$ then, with probability at least $1-\delta$, ${\cal B}_\D$ rejects.
\item If $F$ is $(\epsilon/3)$-far from every function in $C(\Gamma)$ with respect to the uniform distribution then, with probability at least $1-\delta$, ${\cal B}_U$ rejects.
\end{enumerate}
\end{lemma}
\begin{proof} Algorithm ${\cal B}$ runs ${\cal A}$  with confidence parameter $\delta/3$ and accuracy $\epsilon/12$. By Lemma~\ref{SimulateF}, every membership query to $F(y)$ can be simulated with one membership query to $f$. Every ExQ$_{\D'}$ (resp. ExQ) for the induced distribution $\D'$ of $\D$ on the coordinates $\Gamma$, can be simulated with one ExQ$_{\D}$ and $k \log(3kQ(\epsilon/12,\delta/3)/\delta)$ membership queries (resp. one ExQ) with failure probability $\delta/(3Q(\epsilon/12,\delta/3))$, and therefore, with failure probability $\delta/3$ for all the ExQ$_{\D'}$ queries asked in the learning algorithm.

If algorithm ${\cal A}$ runs more than it should, asks more than $Q(\epsilon/12,\delta/3)$ ExQ$_{\D'}$, asks more than $M(\epsilon/12,\delta/3)$ MQ or outputs $h\not\in C(Y)$ then ${\cal B}_\D$ rejects. If ${\cal A}$ outputs $h\in C(Y)$ then, with probability at least $1-(2\delta/3)$,
\begin{enumerate}
\item If $F\in C(\Gamma)$ then $F$ is $(\epsilon/12)$-close to $h$
\item If $F$ is $(\epsilon/3)$-far from every function in $C(\Gamma)$ then $F$ is $(\epsilon/4)$-far from $h$.
\end{enumerate}
Now, using Chernoff's bound in Lemma~\ref{Chernoff}, algorithm ${\cal B}$, can estimate the distance between $F$ and $h$ with accuracy $\epsilon/24$ and confidence $\delta/6$ using $O((\log(1/\delta))/\epsilon)$ strings chosen according to the distribution $\D'$. This can be done using $O((\log(1/\delta))/\epsilon)$ queries in the uniform model and $O((k/\epsilon)(\log (1/\delta)\log(1/(\epsilon\delta)))$ with confidence $\delta/6$ in the distribution-free model.
\end{proof}

Define the oracle WExQ$_{\D}$ (Weak ExQ$_{\D}$) that returns with probability $1/2$ a $x\in\{0,1\}^n$ according to the distribution $\D$ and with probability $1/2$ an arbitrary $x\in\{0,1\}^n$. In some of the learning algorithms given in the sequel, the algorithms still work if we replace the oracle ExQ$_\D$ with WExQ$_{\D}$. In that case, we can save the factor of $\log(3kQ(\epsilon/12,\delta/3)/\delta)$ in the query complexity of Lemma~\ref{LtoT02} in the distribution-free setting. We will discuss this in Section~\ref{IFR}.

\subsection{The First Tester}
We are now ready to give the first tester.

Consider the tester {\bf Tester$C$} in Figure~\ref{Tester1}. Note that the tester rejects if any one of the procedures called by the tester rejects. We prove

\newcounter{ALC4}
\setcounter{ALC4}{0}
\newcommand{\stepd}{\stepcounter{ALC4}$\arabic{ALC4}.\ $\>}
\newcommand{\steplabeld}[1]{\addtocounter{ALC4}{-1}\refstepcounter{ALC4}\label{#1}}
\begin{figure}[h!]
  \begin{center}
  \fbox{\fbox{\begin{minipage}{28em}
  \begin{tabbing}
  xxx\=xxxx\=xxxx\=xxxx\=xxxx\=xxxx\= \kill
  {{\bf Tester$C$}$(f,\D,\epsilon)$}\\
 {\it Input}: Oracle that accesses a Boolean function $f$ and $\D$. \\
  {\it Output}: If any one of the procedures reject \\
  \>\>then ``reject'' or ``accept''\\ \\

\stepd\steplabeld{FF01}
  $(X,V,I)\gets ${\bf ApproxTarget}$(f,\D,\epsilon,1/3)$. \\
\stepd\steplabeld{FF02}
  {\bf TestSets}$(X,V,I)$.\\

  \stepd\steplabeld{FF04}
  Define $F\equiv f(x(\pi_{f,I})\circ 0_{\overline{X}})$\\
  \stepd\steplabeld{FF03}
  {{\bf Close$fF$}$(f,\D,\epsilon,1/15)$} \\ \\
  \ {\bf For any distribution}\\
  \stepd\steplabeld{FF08}{{\bf Close$FC\D$}$(f,\D,\epsilon,1/15)$}\\
  \stepd\steplabeld{FF14} Return ``accept''\\ \\
  \addtocounter{ALC4}{-2}

  {\bf For the uniform distribution}\\
  \stepd\steplabeld{FF19}{{\bf Close$FCU$}$(f,\epsilon,1/15)$}\\
  \stepd\steplabeld{FF20} Return ``accept''
  \end{tabbing}
  \end{minipage}}}
  \end{center}
	\caption{A tester for subclasses $C$ of $k$-Junta}
	\label{Tester1}
	\end{figure}

\begin{theorem}\label{FirstTester} Let $C\subseteq k-$Junta that is closed under zero and variable projections. Then
\begin{enumerate}
\item There is a $poly(|C(\Gamma)|,n,1/\epsilon)$ time two-sided adaptive algorithm, {\bf Tester$C$}, for $\epsilon$-testing $C$ that makes $\tilde O((1/\epsilon)(k+\log|C(\Gamma)|))$ queries. That is
    \begin{enumerate}
    \item If  $f\in C$ then, with probability at least $2/3$, {\bf Tester$C$} accepts.
    \item If $f$ is $\epsilon$-far from every function in $C$ with respect to the uniform distribution then, with probability at least $2/3$, {\bf Tester$C$} rejects.
    \end{enumerate}
\item There is a $poly(|C(\Gamma)|,n,1/\epsilon)$ time two-sided distribution-free adaptive algorithm, {\bf Tester$C$}, for $\epsilon$-testing $C$ that makes $\tilde O((k/\epsilon)\log (2|C(\Gamma)|))$ queries. That is
\begin{enumerate}
    \item If  $f\in C$ then, with probability at least $2/3$, {\bf Tester$C$} accepts.
    \item If $f$ is $\epsilon$-far from every function in $C$ with respect to the distribution $\D$ then, with probability at least $2/3$, {\bf Tester$C$} rejects.
    \end{enumerate}
\end{enumerate}
\end{theorem}
\begin{proof} We prove~{\it 1a} and {\it 2a}. Let $f\in C$. Consider step~\ref{FF01} in {\bf Tester$C$}. By Lemma~\ref{dist} and \ref{kjun}, with probability at least $2/3$, {\bf ApproxTarget} outputs $(X,V,I)$ that satisfies Assumption~\ref{assmp}. Now with this assumption we have: By Lemma~\ref{kjun1}, {\bf TestSets} in step~\ref{FF02} does not reject. By Lemma~\ref{FirstApp}, {\bf Close$fF$} in step~\ref{FF03} does not reject. By Lemma~\ref{clooo}, {\bf Close$FC\D$} and {\bf Close$FCU$} in step~\ref{FF08} do not reject. Therefore, with probability at least $2/3$ the tester accepts.

We now prove~{\it 1b} and {\it 2b}. Suppose $f$ is $\epsilon$-far from every function in $C$ with respect to $\D$. The proof for the uniform distribution is similar. If in  step~\ref{FF01} {\bf ApproxTarget} outputs $(X,V,I)$ then by Lemma~\ref{cloose}, with probability at least $14/15$, $f(x_X\circ 0_{\overline{X}})$ is $(\epsilon/3)$-close to $f$. If in step~\ref{FF02} {\bf TestSets} does not reject then, by Lemma~\ref{closelit}, with probability at least $14/15$, for all $\ell\in I$, $f(x_{X_\ell}\circ v_{\overline{X_\ell}}^{(\ell)})$ is $(1/30)$-close to a literal $\{x_{\tau(\ell)},\overline{x_{\tau(\ell)}}\}$. Therefore with probability at least $13/15$, Assumption~\ref{assmp} is true. Then by Lemma~\ref{kLem}, either $f(x_X\circ 0_{\overline{X}})$ is $(\epsilon/3)$-far from $F$ with respect to ${\cal D}$, or $F$ is $(\epsilon/3)$-far from every function in $C(\Gamma)$ with respect to ${\cal D}$. If $f(x_X\circ 0_{\overline{X}})$ is $(\epsilon/3)$-far from $F$ with respect to $\D$ then by Lemma~\ref{FirstApp}, with probability at least $14/15$, {\bf Close$fF$} rejects. If $F$ is $(\epsilon/3)$-far from every function in $C(\Gamma)$ with respect to ${\cal D}$ then by Lemma~\ref{clooo}, with probability at least $14/15$, {\bf Close$FC$} rejects. By the union bound the probability that the tester rejects is at least $2/3$.

The query complexity follows from Lemmas~\ref{Query01}, \ref{Query02}, \ref{FirstApp} and \ref{clooo}.
\end{proof}

If we replace {\bf Close$FC\D$} and {\bf Close$FCU$} with the testers in Lemma~\ref{LtoTTriv}, \ref{LtoT01} and \ref{LtoT02}, we get the following results

\begin{theorem}\label{LtoTTriv2} If there is a polynomial time algorithm that given a set $${\cal Y}=\{(y^{(1)},\xi_1),\ldots,(y^{(t)},\xi_t)\}\subseteq \{0,1\}^q\times \{0,1\}$$ decides whether there is a function $F\in C(Y)$ that is consistent with ${\cal Y}$, then
\begin{enumerate}
\item There is a polynomial time two-sided adaptive algorithm for $\epsilon$-testing $C$ that makes $\tilde O((1/\epsilon)(k+\log|C(\Gamma)|))$ queries.
\item There is a polynomial time two-sided distribution-free adaptive algorithm for $\epsilon$-testing $C$ that makes $\tilde O((k/\epsilon)\log (2|C(\Gamma)|))$ queries.
\end{enumerate}
\end{theorem}

\begin{theorem}\label{TesterL01} If there is a polynomial time algorithm ${\cal A}$
that, given as an input any constant $\delta$, properly exactly learns $ C(Y)$ with confidence parameter $\delta$ and makes $M(\delta)$ membership queries then
\begin{enumerate}
\item There is a polynomial time two-sided adaptive algorithm for $\epsilon$-testing $C$ that makes $M(1/24)+\tilde O(k/\epsilon)$ queries.
\item There is a polynomial time two-sided distribution-free adaptive algorithm for $\epsilon$-testing $C$ that makes $M(1/24)+\tilde O(k/\epsilon)$ queries.
\end{enumerate}
\end{theorem}

\begin{theorem}\label{TesterL02} If there is a polynomial time algorithm ${\cal A}$
that, given as an input a constant $\delta$ and any $\epsilon$, learns $C(Y)$, with confident $\delta$, accuracy $\epsilon$, makes $M(\epsilon,\delta)$ $MQ$ and $Q(\epsilon,\delta)$ ExQ$_U$ (resp. ExQ$_\D$) then
\begin{enumerate}
\item There is a polynomial time two-sided adaptive algorithm for $\epsilon$-testing $C$ that makes $$\tilde O\left(M(\epsilon/12,1/24)+Q(\epsilon/12,\delta/3)+\frac{k}{\epsilon}\right)$$ queries.
\item There is a polynomial time two-sided distribution-free adaptive algorithm for $\epsilon$-testing $C$ that makes $$\tilde O\left(M(\epsilon/12,1/24)+kQ(\epsilon/12,1/24)+\frac{k}{\epsilon}\right)$$ queries.
\end{enumerate}
\end{theorem}

Finally, one trivial but useful result is
\begin{theorem}\label{ThTriv} If there is a polynomial time algorithm ${\cal A}$
that, given as an input any constant $\delta$ and any $\epsilon$ makes $M(\epsilon,\delta)$ MQs and $Q(\epsilon,\delta)$ ExQ$_U$ (resp. ExD$_\D$) and distinguish between $F\in C(Y)$ and $F$ $\epsilon$-far from every function in $C(Y)$ with respect to the uniform distribution (resp. with respect to the distribution $\D$) then
\begin{enumerate}
\item There is a polynomial time two-sided adaptive algorithm for $\epsilon$-testing $C$ that makes $$\tilde O\left(M(\epsilon/12,1/24)+Q(\epsilon/12,1/24)+\frac{k}{\epsilon}\right)$$ queries.
\item There is a polynomial time two-sided distribution-free adaptive algorithm for $\epsilon$-testing $C$ that makes $$\tilde O\left(M(\epsilon/12,1/24)+kQ(\epsilon/12,1/24)+\frac{k}{\epsilon}\right)$$ queries.
\end{enumerate}
\end{theorem}

\section{Results}\label{result1}
In this section we define the classes and give the results for the classes $k$-Junta, $k$-Linear, $k$-Term, $s$-Term Monotone $r$-DNF, size-$s$ Decision Tree, size-$s$ Branching Program, Functions with Fourier Degree at most $d$, Length-$k$ Decision List and $s$-Sparse Polynomial of Degree $d$.

We will use words that are capitalized for classes and non-capitalized words for functions. For example, $k$-Junta is the class of all $k$-juntas.

\subsection{Testing $k$-Junta} For $k$-Junta in uniform distribution framework,
Ficher et al.~\cite{FischerKRSS02} introduced the junta testing problem and gave a non-adaptive algorithm that makes $\tilde O(k^2)/\epsilon$ queries. Blais in~\cite{Blais08} gave a non-adaptive algorithm that makes $\tilde{O}(k^{3/2})/\epsilon$ queries and in~\cite{Blais09} an adaptive algorithm that makes $O(k\log k+k/\epsilon)$ queries.
On the lower bounds side, Fisher et al.~\cite{FischerKRSS02} gave an $\Omega(\sqrt{k})$ lower bound for non-adaptive testing. Chockler and Gutfreund~\cite{ChocklerG04} gave an $\Omega(k)$ lower bound for adaptive testing and, recently, Sa\u{g}lam in~ \cite{Saglam18} improved this lower bound to $\Omega(k\log k)$. For the non-adaptive testing Chen et al.~\cite{ChenSTWX17} gave the lower bound $\tilde{\Omega}(k^{3/2})/\epsilon$.

For testing $k$-junta in the distribution-free model, Chen et al.~\cite{LiuCSSX18} gave a one-sided adaptive algorithm that makes $\tilde{O}(k^{2})/\epsilon$ queries and proved a lower bound $\Omega(2^{k/3})$ for any non-adaptive algorithm. The result of Halevy and Kushilevitz in~\cite{HalevyK07} gives a one-sided non-adaptive algorithm that makes $O(2^k/\epsilon)$ queries. The adaptive $\Omega(k\log k)$ uniform-distribution lower bound from~\cite{Saglam18} trivially extends to the distribution-free model.
Bshouty~\cite{Bshouty19} gave a two-sided adaptive algorithm that makes $\tilde O(1/\epsilon)k\log k$ queries.

Our algorithm in this paper gives
\begin{theorem} For any $\epsilon>0$, there is a polynomial time two-sided distribution-free adaptive algorithm for $\epsilon$-testing $k$-Junta that makes $\tilde O(k/\epsilon)$ queries.
\end{theorem}
\begin{proof} We use Theorem~\ref{ThTriv}. Since every $F(Y)$ is in $k$-Junta$(Y)$, the algorithm ${\cal A}$ always accepts. Therefore, we have $M=Q=0$, and the algorithm makes $\tilde O(k/\epsilon)$ queries.
\end{proof}

\subsection{Testing $k$-Linear}
The function is linear if it is a sum (over the binary field $F_2$) of variables. The class Linear is the class of all linear functions. The class $k$-Linear is Linear$\cap k$-Junta. That is, the class of functions that are the sum of at most $k$ variables.

Blum et al. \cite{BlumLR93} showed that there is an algorithm for testing Linear under the uniform distribution that makes $O(1/\epsilon)$ queries. For testing $k$-Linear under the uniform distribution, Fisher, et al.~\cite{FischerKRSS02} gave a tester that makes $O(k^2/\epsilon)$ queries. They also gave the lower bound $\Omega(\sqrt{k})$ for non-adaptive algorithms. Goldreich~\cite{Goldreich10}, proved the lower bound $\Omega(k)$ for non-adaptive algorithms and $\Omega(\sqrt{k})$ for adaptive algorithms. Then Blais et al.~\cite{BlaisBM11} proved the lower bound $\Omega(k)$ for adaptive algorithms. Blais and Kane, in~\cite{BlaisK12}, gave the lower bound $k-o(k)$ for adaptive algorithms and $2k-o(k)$ for non-adaptive algorithms.

Testing $k$-Linear can be done by first testing if the function is $k$-Junta and then testing if it is Linear. Therefore, there is an adaptive algorithm for $\epsilon$-testing $k$-Linear under the uniform distribution that makes $\tilde O(k/\epsilon)$ queries.

In this paper we prove
\begin{theorem} For any $\epsilon>0$, there is a polynomial time two-sided distribution-free adaptive algorithm for $\epsilon$-testing $k$-Linear that makes $\tilde O(k/\epsilon)$ queries.
\end{theorem}
\begin{proof} We use Theorem~\ref{TesterL02}. Here $C(Y)=\{y_1+\cdots+y_q\}$ contains one function and therefore the learning algorithm just outputs $y_1+\cdots+y_q$. Therefore $M=Q=0$ and the result follows.
\end{proof}

\subsection{Testing $k$-Term}
A {\it term} (or {\it monomial}) is a conjunction of literals and {\it Term} is the class of all terms. A {\it $k$-term} is a term with at most $k$ literals and {\it $k$-Term} is the class of all $k$-terms.

In the uniform distribution model, Pernas et al.~\cite{ParnasRS02}, gave a tester for $k$-terms that makes $O(1/\epsilon)$ queries in the uniform model. We give the same result in the next section. In this paper we prove the following result for the distribution-free model. When $k=n$, better results can be found in~\cite{GlasnerS09,DolevR11}.
\begin{theorem} For any $\epsilon>0$, there is a polynomial time two-sided distribution-free adaptive algorithm for $\epsilon$-testing $k$-Term that makes $\tilde O(k/\epsilon)$ queries.
\end{theorem}
\begin{proof} Recall that $x_i^0=x_i$ and $x_i^1=\overline{x_i}$. Here $C(Y)=\{y_1^{\xi_1}\wedge\cdots\wedge y_q^{\xi_q}|\xi\in\{0,1\}^q\}$ contains $2^q$ functions. We use Theorem~\ref{TesterL02} with Remark~\ref{remark}. Since $V$ contains witnesses for each variable it follows that $\xi_i$ are known. Just take any string $a$ that satisfies $F(a)=1$ and then $\xi_i=\overline{a_i}$. Therefore $M=Q=0$ and the result follows.
\end{proof}

\subsection{Testing $s$-Term Monotone $r$-DNF}\label{TMDD}
A {\it DNF} is a disjunction of terms. An {\it $r$-DNF} is a disjunction of $r$-terms.
The class {\it $s$-Term $r$-DNF} is the class of all $r$-DNFs with at most $s$ terms.
The class {\it $s$-Term Monotone $r$-DNF} is the class of all $r$-DNFs with at most $s$ terms with no negated variables. A DNF $f$ is called {\it unate DNF} if there is $\xi\in \{0,1\}^n$ such that $f(x_1^{\xi_1},\ldots,x_n^{\xi_n})$ is monotone DNF. If $\xi_i=0$ then we say that $f$ {\it is positive unate in $x_i$}; otherwise we say that $f$ {\it is negative unate in $x_i$}. Similarly, one can define the classes Unate DNF, Unate $s$-DNF etc.

We first give a learning algorithm for $s$-Term Monotone $r$-DNF. The algorithm is in Figure~\ref{MDNFL}. In the algorithm, we use $P_{1/r}$ for the probability distribution over the strings $b\in \{0,1\}^n$ where each coordinate $b_i$ is chosen randomly and independently to be $1$ with probability $1-1/r$ and $0$ with probability $1/r$. For two strings $x,y\in \{0,1\}^n$ we denote $x*y=(x_1y_1,\ldots,x_ny_n)$ where $x_iy_i=x_i\wedge y_i$. The procedure {\bf FindMinterm}$(f,a)$ flips bits that are one in $a$ to zero as long as $f(a)=1$.

\newcounter{ALCM}
\setcounter{ALCM}{0}
\newcommand{\stepM}{\stepcounter{ALCM}$\arabic{ALCM}.\ $\>}
\newcommand{\steplabelM}[1]{\addtocounter{ALCM}{-1}\refstepcounter{ALCM}\label{#1}}
\begin{figure}[h!]
  \begin{center}
  \fbox{\fbox{\begin{minipage}{28em}
  \begin{tabbing}
  xxx\=xxxx\=xxxx\=xxxx\=xxxx\=xxxx\= \kill
  {{\bf LearnMonotone}$(f,\D,\epsilon,\delta,s,r)$}\\
 {\it Input}: Oracle that accesses a Boolean function $f$\\
 \>\> that is $s$-term monotone $r$-DNF and $\D$. \\
  {\it Output}: $h$ that is $s$-term monotone $r$-DNF\\ \\
  xxx\=xxxx\=xxxx\=xxxx\=xxxx\=xxxx\= \kill
\stepM\steplabelM{MDNF00}
  $h\gets 0$. \\
  \stepM\steplabelM{MDNF01}
  Repeat $4(s/\epsilon)\log(1/\delta)$ times. \\
\stepM\steplabelM{MDNF02}\>
  Choose $a\in \D$.\\
  \stepM\steplabelM{MDNF03}\>
 If $f(a)=1$ and $h(a)=0$ then \\
 \stepM\steplabelM{MDNF04}\>\>
 $t\gets 0$\\
   \stepM\steplabelM{MDNF04b}\>\>
While $t\le \alpha:=4r\ln(2ns/\delta)$ and $wt(a)>r$ do\\
   \stepM\steplabelM{MDNF05a}\>\>\>
   $t\gets t+1;$ If $t=\alpha+1$ Output ``fail''\\
   \stepM\steplabelM{MDNF05}\>\>\>
   Choose $y\in P_{1/r}$\\
  \stepM\steplabelM{MDNF06}\>\>\>
   If $f(a*y)=1$ then $a\gets a*y$\\
   \stepM\steplabelM{MDNF07a}\>\>
   $a\gets${\bf FindMinterm}$(f,a)$\\
     \stepM\steplabelM{MDNF07}\>\>
   $h\gets h\vee \prod_{a_i=1}x_i$\\
   \stepM\steplabelM{MDNF08}
   Output $h$
  \end{tabbing}
  \end{minipage}}}
  \end{center}
	\caption{A learning algorithm for $s$-Term Monotone $r$-DNF}
	\label{MDNFL}
	\end{figure}

We now show
\begin{lemma}\label{LearnMDNF} If the target function $f$ is $s$-term monotone $r$-DNF then for any constant $\delta$, algorithm {\bf LearnMonotone} asks $O(s/\epsilon)$ ExQ$_D$ and $O(sr\log(ns))$ MQ and, with probability at least $1-\delta$, learns an $s$-term monotone $r$-DNF, $h$, that satisfies $\Pr_\D[h\not =f]\le \epsilon$.
\end{lemma}

\begin{proof} We first show that if in the $m$th iteration of the algorithm (steps \ref{MDNF02}-\ref{MDNF07}) the function $h$ contains $\ell$ terms of $f$ and $f(a)=1$ and $h(a)=0$ then, with probability at least $1-\delta/(2s)$, steps \ref{MDNF04} to \ref{MDNF07} adds to $h$ a new term of~$f$. This implies that in the $(m+1)$th iteration $h$ contains $\ell+1$ terms of $f$. Then, since the number of terms of $f$ is at most $s$, with probability at least $1-\delta/2$, all the terms in $h$ are terms in $f$. We then show that, with probability at least $1-\delta/2$, the procedure outputs $h$ that satisfies $\Pr_\D[h\not =f]\le \epsilon$.

First notice that if $f(a)=1$ and $h(a)=0$ then for every $y\in\{0,1\}^n$, $h(a*y)=0$. This follows from the fact that $h$ is monotone and $a*y\le a$. Therefore if $a$ receives the values $a^{(1)},\ldots,a^{(\tau)}$ in the While loop then $f(a^{(i)})=1$ and $h(a^{(i)})=0$ for all $i=1,\ldots,\tau$. We also have $a^{(i+1)}=a^{(i)}$ if $f(a^{(i)}*y)=0$ and $a^{(i+1)}=a^{(i)}*y$ if $f(a^{(i)}*y)=1$. Consider the random variable $W_i=wt(a^{(i)})-r$. We will now compute $\E[W_{i+1}|W_i]$. Since $f(a^{(i)})=1$ and $h(a^{(i)})=0$, there is a term $T$ in $f$ that is not in $h$ that satisfies $T(a^{(i)})=1$.  Suppose $T=x_{j_1}x_{j_2}\cdots x_{j_{r'}}$, $r'\le r$. Then $a^{(i)}_{j_1}=\cdots=a^{(i)}_{j_{r'}}=1$. Consider another $r-r'$ entries in $a^{(i)}$ that are equal to $1$, $a^{(i)}_{j_{r'+1}}=\cdots=a^{(i)}_{j_r}=1$. Such entries exist because of the condition $wt(a)>r$ of the While command. Note that $wt(a^{(i)})=W_i+r$. Let $j_{r+1},\ldots,j_{r+W_i}$ be the other entries of $a^{(i)}$ that are equal to $1$. Let $A$ be the event that, for the $y\in P_{1/r}$ chosen at this stage,  $y_{j_1}=\cdots=y_{j_r}=1$. Notice that if event $A$ happens then $T(a^{(i+1)})=f(a^{(i+1)})=1$ and $a^{(i+1)}=a^{(i)}*y$. Then
\begin{eqnarray}
\E[W_{i+1}|W_i]&=&\E[W_{i+1}|W_i,A]\Pr[A]+\E[W_{i+1}|W_i,\bar A]\Pr[\bar A]\nonumber\\
&\le &\E[W_{i+1}|W_i,A]\left(1-\frac{1}{r}\right)^r+W_i\left(1-\left(1-\frac{1}{r}\right)^r\right)\label{pr01}\\
&= & W_i\left(1-\frac{1}{r}\right)^{r+1}+W_i\left(1-\left(1-\frac{1}{r}\right)^r\right)\label{pr02}\\
&=& W_i\left(1-\frac{1}{r}\left(1-\frac{1}{r}\right)^r\right)\le W_i\left(1-\frac{1}{4r}\right).\nonumber
\end{eqnarray}
The inequality in (\ref{pr01}) follows from the fact that $W_{i+1}\le W_i$ and (\ref{pr02}) follows from the fact that the expected number of ones in $y_{j_{r+1}}a_{j_{r+1}},\ldots,y_{j_{r+W_i}}a_{j_{r+W_i}}$ is $(1-1/r)W_i$.

Therefore $\E[W_i]\le n(1-1/(4r))^i$. The probability that the algorithm fails is the probability that $t= 4r\ln(2ns/\delta)$. By Markov's Bound, Lemma~\ref{Markov}, this is bounded by
$$\Pr[wt(a^{(t)})>r]=\Pr[W_t>1]\le \E[W_t]\le n\left(1-\frac{1}{4r}\right)^t\le \frac{\delta}{2s}.$$
This completes the first part of the proof.

Now we show that, with probability at least $1-\delta/2$, the procedure outputs $h$ that satisfies $\Pr_\D[f\not= h]\le \epsilon$. Let $h^{(i)}$ be the function $h$ at iteration $i=1,2,\ldots,w$. Since $h^{(1)}\implies h^{(2)}\implies \cdots \implies h^{(w)}=h\implies f$, if $\Pr_\D[f\not= h]> \epsilon$ then $\Pr_\D[f\not= h^{(i)}]> \epsilon$ for all $i$. Therefore, the probability that $\Pr_\D[f\not= h]> \epsilon$ is less than the probability that for $v=4s\log(1/\delta)/\epsilon$ strings $a^{(1)},\ldots,a^{(v)}$ chosen independently at random according to the distribution $\D$, less than $s$ of them satisfies $g_i(a^{(i)})\not=f(a^{(i)})$ for Boolean functions $g_i$ that satisfy $\Pr_\D[g_i\not=f]\ge \epsilon$. By Chernoff's bound, Lemma \ref{Chernoff}, this probability is less than $\delta/2$.

The algorithm asks at most $4s\log(1/\delta)/\epsilon=O(s/\epsilon)$ ExQ$_\D$ and at most $s\cdot 4r\ln(2ns/\delta)=O(sr\log (ns))$ MQ.
\end{proof}

Now we show
\begin{theorem}\label{ThMon} For any $\epsilon>0$, there is a polynomial time two-sided distribution-free adaptive algorithm for $\epsilon$-testing $s$-Term Monotone $r$-DNF that makes $\tilde O(rs^2/\epsilon)$ queries.
\end{theorem}
\begin{proof} The number of relevant variables in any $s$-term monotone $r$-DNF is at most $q\le k=sr$. By Lemma~\ref{LearnMDNF}, $C(Y)$ can be learned with constant confidence $\delta$ and accuracy $\epsilon$ in $M=\tilde O(sr\log(qs))=\tilde O(sr)$ MQ and $O(s/\epsilon)$ ExQ$_\D$. By Theorem~\ref{TesterL02}, there is a distribution-free tester for $s$-Term Monotone $r$-DNF that makes $\tilde O(s^2r/\epsilon)$ queries.
\end{proof}

\begin{theorem}\label{MUnate} For any $\epsilon>0$, there is a polynomial time two-sided distribution-free adaptive algorithm for $\epsilon$-testing $s$-Term Unate $r$-DNF that makes $\tilde O(rs^2/\epsilon)$ queries.
\end{theorem}
\begin{proof} The set of witnesses tells us, for every variable $x_i$, if $f$ is positive unate in $x_i$ or negative unate. If $f(v^{(\ell)})=1$, $f(0_{X_\ell}\circ v^{(\ell)}_{\overline{X_\ell}})=0$ then $f$ is positive unate in  $x_{\tau(\ell)}$ and if $f(v^{(\ell)})=0$, $f(0_{X_\ell}\circ v^{(\ell)}_{\overline{X_\ell}})=1$ then $f$ is negative unate in  $x_{\tau(\ell)}$. Then the result immediately follows from Theorem~\ref{ThMon}.
\end{proof}

\subsection{Testing Size-$s$ Decision Tree and Size $s$ Branching Program}
A {\it decision tree} is a rooted binary tree in which each internal node is labeled with a variable $x_i$ and has two children. Each leaf is labeled with an output from $\{0,1\}$. A decision tree computes a Boolean function in an obvious way: given an input $x\in \{0,1\}^n$, the value of the function on $x$ is the output in the leaf reached by starting at the root and going left or right at each internal node according to whether the variable's value in $x$ is $0$ or $1$, respectively. The {\it size} of a decision tree is the number of leaves of the tree. The class size-$s$ Decision Tree is the class of all decision trees of size $s$.

A {\it branching program} is a rooted directed acyclic graph with two sink nodes labeled $0$ and $1$. As in the decision tree, each internal node is labeled with a variable $x_i$ and has two children. The two edges to the children are labeled with $0$ and $1$. Given an input $x$, the value of the branching program on $x$ is the label of the sink node that is reached as described above. The {\it size} of a branching program is the number of nodes in the graph. The class size-$s$ Branching Program is the class of all Branching Program of size $s$.

Diakonikolas et al.~\cite{DiakonikolasLMORSW07}, gave a tester for size-$s$ Decision Tree and size $s$ Branching Program under the uniform distribution that makes $\tilde O(s^4/\epsilon^2)$ queries. Chakraborty et al.~\cite{ChakrabortyGM11} improved the query complexity to $\tilde O(s/\epsilon^2)$. In this paper we prove

\begin{theorem} For any $\epsilon>0$, there is a two-sided adaptive algorithm for $\epsilon$-testing size-$s$ Decision Tree and size-$s$ Branching Program that makes $\tilde O(s/\epsilon)$ queries.

There is a two-sided distribution-free adaptive algorithm for $\epsilon$-testing size-$s$ Decision Tree and size $s$ Branching Program that makes $\tilde O(s^2/\epsilon)$ queries.
\end{theorem}
\begin{proof} For decision tree, $C(Y)$ contains the decision trees with $q=|Y|\le k=s$ relevant variables. It is shown in~\cite{DiakonikolasLMORSW07} that $|C(Y)|\le (8s)^s$. For branching programs $|C(Y)|\le (s+1)^{3s}$. Now by Theorem~\ref{FirstTester} the result follows.
\end{proof}

\subsection{Functions with Fourier Degree at most $d$}
For convenience here we take the Boolean functions to be $f:\{-1,1\}^n\to \{-1,1\}$. Then every Boolean function has a unique Fourier representation $f(x)=\sum_{S\subseteq [n]} \hat f_S\prod_{i\in S}x_i$ where $\hat f_S$ are the {\it Fourier coefficients} of $f$. The {\it Fourier degree} of $f$ is the largest $d=|S|$ with $\hat f_S\not=0$.

Let $C$ be the class of all Boolean functions over $\{-1,1\}^n$ with Fourier degree at most $d$. Nisan and Szegedy, \cite{NisanS92}, proved that any Boolean function with Fourier degree $d$ must have at most $k:=d2^d$ relevant variables. Diakinikolas et al.~\cite{DiakonikolasLMORSW07}, show that every nonzero Fourier coefficient of $f\in C$ is an integer multiple of $1/2^{d-1}$. Since $\sum_{S\subseteq [n]} \hat f_S^2=1$, there are at most $2^{2d-2}$ nonzero Fourier coefficients in $f\in C$.

Diakonikolas et al.~\cite{DiakonikolasLMORSW07}, gave an exponential time tester for Boolean functions with Fourier degree at most $d$ under the uniform distribution that makes $\tilde O(2^{6d}/\epsilon^2)$ queries. Chakraborty et al.~\cite{ChakrabortyGM11} improved the query complexity to $\tilde O(2^{2d}/\epsilon^2)$. In this paper we prove
\begin{theorem} For any $\epsilon>0$, there is a $poly(2^d,n)$ time two-sided distribution-free adaptive algorithm for $\epsilon$-testing for the class of Boolean functions with Fourier degree at most $d$ that makes $\tilde O(2^{2d}+2^d/\epsilon)$ queries.
\end{theorem}
\begin{proof}
Bshouty gives in~\cite{Bshouty18} an exact learning algorithm for such class\footnote{The class in~\cite{Bshouty18} is the class of decision trees of depth $d$ but the analysis is the same for the class of functions with Fourier degree at most $d$} that asks $M=\tilde O(2^{2d}\log n)$ membership queries for any constant confidence parameter $\delta$. Now since $q=|Y|\le k=d2^d$, by Theorem~\ref{TesterL01} the result follows.
\end{proof}

\subsection{Testing Length $k$ Decision List}
A decision list is a sequence $f=(x_{i_1},\xi_1,a_1),\ldots,(x_{i_s},\xi_s,a_s)$ for any $s$ where $\xi_i,a_i\in\{0,1\}$. This sequence represents the following function: $f(x):=$ If $x_{i_1}=\xi_1$ then output$(a_1)$ else if $x_{i_2}=\xi_2$ then output$(a_2)$ else if $\cdots$ else if $x_{i_s}=\xi_s$ then output$(a_s)$. Length-$k$ decision list is a decision list with $s\le k$. The class Decision List is the class of all decision lists and the class Length-$k$ Decision List is the class of all length-$k$ decision lists.

It is known that this class is learnable under any distribution with $O((k\log n+\log (1/\delta))/\epsilon)$ ExQ$_D$,~\cite{BlumerEHW87,Rivest87}. This implies

\begin{theorem}
For any $\epsilon>0$, there is a polynomial time two-sided distribution-free adaptive algorithm for $\epsilon$-testing Length-$k$ Decision List that makes $\tilde O(k^2/\epsilon)$ queries.
\end{theorem}
\begin{proof} The result follows from  Theorem~\ref{TesterL02}.
\end{proof}

\subsection{Testing $s$-Sparse Polynomial of Degree $d$}
A polynomial (over the field $F_2$) is a sum (in the binary field $F_2$) of monotone terms. An $s$-sparse polynomial is a sum of at most $s$ monotone terms. We say that the polynomial $f$ is of degree $d$ if its terms are monotone $d$-terms. The class $s$-Sparse Polynomial of Degree $d$ is the class of all $s$-sparse polynomials of degree~$d$. The class Polynomial of Degree~$d$ is the class of all polynomials of degree~$d$.

In the uniform distribution model, Diakonikolas et al.~\cite{DiakonikolasLMORSW07}, gave the first testing algorithm for the class $s$-Sparse Polynomial that runs in exponential time and makes $\tilde O(s^4/\epsilon^2)$ queries. Chakraborty et al.~\cite{ChakrabortyGM11} improved the query complexity to $\tilde O(s/\epsilon^2)$.  Diakonikolas et al. gave in~\cite{DiakonikolasLMSW11} the first  polynomial time testing algorithm that makes $poly(s,1/\epsilon)$ queries. In~\cite{AlonKKLR03}, Alon et al. gave a testing algorithm for Polynomial of Degree $d$ that makes $O(1/\epsilon+d2^{2d})$ queries. They also show the lower bound $\Omega(1/\epsilon+2^d)$ for the number of queries. Combining those results we get a polynomial time testing algorithm for $s$-Sparse Polynomial of Degree $d$ that makes $poly(s,1/\epsilon)+\tilde O(2^{2d})$ queries. Just run the Alon et al. algorithm in~\cite{AlonKKLR03} and then run Diakonikolas et al. algorithm in~\cite{DiakonikolasLMSW11} and accept if both algorithms accept.

Here we prove the following Theorem.

\begin{theorem}\label{ThSpa} For any $\epsilon>0$, there is a two-sided adaptive algorithm for $\epsilon$-testing $s$-Sparse Polynomial of Degree $d$ that makes $\tilde O(s/\epsilon+s2^d)$ queries.

For any $\epsilon>0$, there is a two-sided distribution-free adaptive algorithm for $\epsilon$-testing $s$-Sparse Polynomial of Degree $d$ that makes $\tilde O(s^2/\epsilon+s 2^d)$ queries.
\end{theorem}

We first give a learning algorithm {\bf LearnPolynomial} for $s$-sparse polynomial of degree $d$. See Figure~\ref{PFL}.

\newcounter{ALCMP}
\setcounter{ALCMP}{0}
\newcommand{\stepMP}{\stepcounter{ALCMP}$\arabic{ALCMP}.\ $\>}
\newcommand{\steplabelMP}[1]{\addtocounter{ALCMP}{-1}\refstepcounter{ALCMP}\label{#1}}
\begin{figure}[h!]
  \begin{center}
  \fbox{\fbox{\begin{minipage}{28em}
  \begin{tabbing}
  xxx\=xxxx\=xxxx\=xxxx\=xxxx\=xxxx\= \kill
  {{\bf LearnPolynomial}$(f,\D,\epsilon,\delta,s,d)$}\\
 {\it Input}: Oracle that accesses an $s$-sparse polynomial $f$\\
 \>\> of degree $d$ and $\D$. \\
  {\it Output}: An $s$-sparse polynomial of degree $d$, $h$, or ``fail''\\ \\
  xxx\=xxxx\=xxxx\=xxxx\=xxxx\=xxxx\= \kill
\stepMP\steplabelMP{nPF00}
  $h\gets 0$, $t(h)\gets 0$. \\
  \stepMP\steplabelMP{nPF01}
  Repeat $(s/\epsilon)\ln(3s/\delta)$ times. \\
\stepMP\steplabelMP{nPF02}\>
  Choose $a\in \D$.\\
\stepMP\steplabelMP{nPF02b}\>
  $t(h)\gets t(h)+1$.\\
  \stepMP\steplabelMP{nPF03}\>
 If $(f+h)(a)=1$ then \\
 \stepMP\steplabelMP{nPF04}\>\>
 $m\gets 0$;\\
   \stepMP\steplabelMP{nPF04b}\>\>
While $m\le \alpha:=16\cdot 2^d(2\ln(s/\delta)+\ln n)$ and $wt(a)>d$ do\\
   \stepMP\steplabelMP{nPF05a}\>\>\>
   $m\gets m+1;$ \\
   \stepMP\steplabelMP{nPF05}\>\>\>
   Choose $y\in U$\\
  \stepMP\steplabelMP{nPF06}\>\>\>
   If $(f+h)(a*y)=1$ then $a\gets a*y$\\
  \stepMP\steplabelMP{nPF07x}\>\>
   If $wt(a)>d$ then ``fail''\\
    \stepMP\steplabelMP{nPF07a}\>\>
   $M\gets$ Find a monotone term in $(f+h)(a*x)$\\
     \stepMP\steplabelMP{nPF07}\>\>
   $h\gets h+ M$\\
   \stepMP\steplabelMP{nPF07b}\>\>
   $t(h)\gets 0$.\\
   \stepMP\steplabelMP{nPF08}\>
   If $t(h)=(1/\epsilon)\ln(3s/\delta)$ then Output $h$
  \end{tabbing}
  \end{minipage}}}
  \end{center}
	\caption{A learning algorithm for $s$-sparse polynomial of degree $d$}
	\label{PFL}
	\end{figure}

\begin{lemma}\label{LearnPolynomial} Let $f$ be an $s$-sparse polynomial of degree $d$. For any constant $\delta$, algorithm {\bf LearnPolynomial} asks $O((s/\epsilon)\log s)$ ExQ$_\D$ and $O(s 2^d$ $\log(ns))$ MQ and, with probability at least $1-\delta$, learns an $s$-sparse polynomial of degree $d$, $h$, that satisfies $\Pr_\D[h\not =f]\le \epsilon$.
\end{lemma}
\begin{proof} Suppose $f=\sum_{M\in F}M$,  where $F$ are the set of monotone $d$-terms of $f$ and $|F|=s'\le s$. Suppose, at some stage of the algorithm $h=\sum_{M\in F'} M$ where $F'\subset F$. Then $f+h=\sum_{M\in F\backslash F'}M$. Notice that $F'$ is the set of terms of $f$ that is found by the algorithm up to this stage and $F\backslash F'$ is the set of terms that is yet to be found. Since the number of terms of $f$ is at most $s$,  all we need to show is that:
\begin{enumerate}
\item Each time the algorithm executes steps~\ref{nPF04}-\ref{nPF07b}, with probability at least~$1-\delta/(2s)$, it finds a term of $f+h$, and therefore, a new term of $f$.
\item Assuming 1., the algorithm, with probability at least $1-\delta/2$, outputs an $s$-sparse polynomial of degree $d$, $h$ that satisfies $\Pr_D[f\not=h]\le \epsilon$.
\end{enumerate}
Then, by the union bound, the success probability of the algorithm is at least $1-\delta$ and the result follows.

 We first prove 1. Let $g=f+h$. Suppose that the algorithm finds a string  $a$ such that $g(a)=1$. Then $a$ satisfies at least one term in $g$. Let $M'\in F\backslash F'$ be one of them and let $d'\le d$ be the degree of $M'$ . Then $g(a*x)=\sum_{M\in F\backslash F', M(a)=1}M$ contains $M'$ and therefore $g(a*x)$ is not zero.

 We first show that the probability that after $\alpha_1:=16\cdot 2^d\ln(sn/\delta)$ iterations of steps~\ref{nPF04b}-\ref{nPF06}, the weight of $a$ does not drop below $24d$ is less than $\delta/(4s)$. Then we show that, if the weight of $a$ is less than or equal to $24d$ then the probability that after $\alpha_2:=\alpha-\alpha_1$ more iterations of steps~\ref{nPF04b}-\ref{nPF06} the weight of $a$ does not drop below $d+1$ is less than $\delta/(4s)$. If these two facts are true then after the algorithm finishes executing the While command, with probability at least $1-\delta/(2s)$, the weight of $a$ is less than $d+1$.

 It is known that for any non-zero polynomial $H$ of degree at most $d$, $\Pr_U[H(x)=1]\ge 1/2^d$,~\cite{BshoutyM02}. Since $g(a*x)$ is of degree at most $d$,  for a random uniform string $y$ we get $\Pr[g(a*y)=1]\ge 1/2^d$. Now suppose $wt(a)\ge 24d$. By Chernoff's bound, Lemma~\ref{Chernoff}, the probability that $wt(a*y)>(3/4)wt(a)$ is at most $e^{-wt(a)/24}\le 2^{-d-1}$. Therefore, by the union bound,
\begin{eqnarray}
\Pr[g(a*y)=1 \mbox{\ and\ } wt(a*y)\le (3/4)wt(a)]\ge 1-\left(1-\frac{1}{2^d}+\frac{1}{2^{d+1}}\right)\ge \frac{1}{2^{d+1}}.\label{cO}
\end{eqnarray}
The probability that after $\alpha_1=16\cdot 2^d\ln(sn/\delta)$ iterations of steps~\ref{nPF04b}-\ref{nPF06}, the weight of $a$ does not drop below $24d$ is less than the probability that for $\alpha_1=16\cdot 2^d\ln(sn/\delta)$ random uniform strings~$y$, less than $\log(n)/\log(4/3)$ of them satisfies $g(a*y)=1$ and $wt(a*y)\le (3/4)wt(a)$  given that $a$ satisfies $wt(a)\ge 24d$ and $g(a*x)\not=0$. By (\ref{cO}) and Chernoff's bound this probability is less than $\delta/(4s)$.

We now show that if $wt(a)<24d$, then after $\alpha_2=\alpha-\alpha_1=16\cdot 2^d\ln(s/\delta)$ iterations of steps~\ref{nPF04b}-\ref{nPF06}, with probability at least $1-\delta/(4s)$, the weight of $a$ drops below $d+1$. Take $a$ that satisfies $d+1\le wt(a)<24d$. Then
\begin{eqnarray*}
\Pr[g(a*y)=1\mbox{\ and\ }wt(a*y)< wt(a)]&\ge& \Pr[g(a*y)=1]-\Pr[wt(a*y)=wt(a)]\\
&\ge& \frac{1}{2^d}-\frac{1}{2^{d+1}}=\frac{1}{2^{d+1}}.
\end{eqnarray*}
Then as before, with an additional $\alpha_2$ iterations of steps~\ref{nPF04b}-\ref{nPF06}, with probability at least $1-\delta/(4s)$,  the weight of $a$ drops below $d+1$.

Once the weight of $a$ is less or equal to $d$, the algorithm finds in step~\ref{nPF07a} a monotone term in $g(a*x)$ by building a truth table of $g(a*x)$ using at most $2^d$ queries and learning one of its terms. This term is in $g$ because all the terms of $g(a*x)$ are terms of $g(x)$.

The proof that the algorithm, with probability at least $1-\delta/2$, outputs $h$ such that $\Pr_D[f\not=h]\le \epsilon$ is identical to the proof  in Lemma~\ref{cloose} for the output of {\bf ApproxTarget}.
\end{proof}

We are now ready to prove Theorem~\ref{ThSpa}.
\begin{proof} By Lemma~\ref{LearnPolynomial}, Theorem~\ref{TesterL02} and since $n=|Y|=sd$, $Q=(s/\epsilon)\log s$ and $M=s2^d\log(sd)$ the result follows.
\end{proof}

\section{Testing Classes that are Close to $k$-Junta}

In this section, we show the result for $s$-term DNF in the uniform distribution model. Then in the following section, we show how to extend it to other classes.

The main idea is the following. We first run the procedure {\bf Approx$C$} in Figure~\ref{A3f} that finds $X\subset [n]$ and $w\in\{0,1\}^n$ such that, with high probability,
\begin{enumerate}
\item\label{pp01} The projection $x_X\circ w_{\overline{X}}$ removes variables from $f$ that appear only in terms of $f$ of size at least $c\log(s/\epsilon)$ for some large constant $c$.
\item\label{pp02} $h=f(x_X\circ w_{\overline{X}})$ is $(\epsilon/8)$-close to $f$.
\end{enumerate}
From~(\ref{pp01}) we conclude that the terms of size at most $c\log(s/\epsilon)$ in $f$ contain all the variables of~$h$. Since the number of terms in $f$ is at most $s$ the number of variables that remain in $h$ is at most $k:=cs\log(s/\epsilon)$. From~(\ref{pp02}) we conclude that if $f$ is $\epsilon$-far from every $s$-term DNF then $h$ is $(7\epsilon/8)$-far from every $s$-term DNF and therefore $h$ is $(7\epsilon/8)$-far from every $s$-term DNF with at most $k$ variables. Therefore, it is enough to distinguish whether $h$ is an $s$-term DNF with at most $k$ variables or $(7\epsilon/8)$-far from every $s$-term DNF with at most $k$ variables. This can be done by the algorithm {\bf Tester$C$} in the previous section

Note that removing variables that only appears in large size terms does not necessarily remove large terms in $f$. Therefore, $h$ may still contain large terms even after running {\bf ApproxTarget} in {\bf Tester$C$}. To handle large terms, we can use any learning algorithm that learns $h$ with accuracy $\epsilon/12$ and use Theorem~\ref{TesterL02}.

This gives a tester for $s$-term DNF that makes $\tilde O(s^2/\epsilon)$ queries, which is not optimal. This is because the number of $s$-term DNF with at most $k$ variables is $m:=2^{O(ks)}$ (and therefore the number of queries in~{\bf Tester$C$} is at least $O((\log m)/\epsilon)=\tilde O(s^2/\epsilon)$).
To get an optimal query tester, we do the following. We build a tester that uses only random uniform queries for the class $s$-term DNF with at most $k$ variables and terms of size at most $r=c'\log(s/\epsilon)$ where $c'$ is a large constant and show that this tester, with high probability, works well for $h$. The reason for that is that when the algorithm uses random uniform queries, with high probability, all the terms of $h$ that are of size greater than $r$ are zero for every query. Since the number of  $s$-term DNF with at most $k$ variables and terms of size at most $r$ is at most $m=2^{O(rs\log k)}$  the number of queries in~{\bf Tester$C$} is at most $O(k/\epsilon+(\log m)/\epsilon)=\tilde O(s/\epsilon)$.

In the next subsection, we give the procedure {\bf Approx$C$} that removes variables that only appear in large size terms, and in Subsection~\ref{RVAST}, we give the tester for $s$-Term DNF. Then in Section~\ref{result2}, we extend the above to other classes.

\subsection{Removing Variables that only Appears in Large Size Terms}

\newcounter{ALCf}
\setcounter{ALCf}{0}
\newcommand{\stepf}{\stepcounter{ALCf}$\arabic{ALCf}.\ $\>}
\newcommand{\steplabelf}[1]{\addtocounter{ALCf}{-1}\refstepcounter{ALCf}\label{#1}}
\begin{figure}[h!]
  \begin{center}
  \fbox{\fbox{\begin{minipage}{28em}
  \begin{tabbing}
  xxx\=xxxx\=xxxx\=xxxx\=xxxx\=xxxx\= \kill
  {\bf Algorithm Approx$C(f,\epsilon,\lambda)$}\\
 {\it Input}: Oracle that accesses a Boolean function $f$ and \\
  {\it Output}: Either ``$X\subseteq [n],w\in\{0,1\}^n$'' or ``reject''\\ \\
{\bf Partition $[n]$ into $r$ sets}\\
\stepf\steplabelf{par1f}
Set $m=c\log(s/\epsilon)$; $r=8ms$.\\
\stepf\steplabelf{par2f}
Choose uniformly at random a partition $X_1,X_2,\ldots,X_r$ of $[n]$\\
\\
{\bf Find a close function and relevant sets} \\
\stepf\steplabelf{Settf}
Set $X=\emptyset$; $I=\emptyset$; $t(X)=0$; $k=3ms$\\
\stepf\steplabelf{twof}
Repeat $M=100\lambda k\ln(100k)/\epsilon$ times\\
\stepf\steplabelf{Chof}
\> Choose $u,v\in U$. \\
\stepf  \> $t(X)\gets t(X)+1$\\
\stepf\steplabelf{con1f}
\> If $f(u_X\circ v_{\overline{X}})\not=f(u)$ then\\
\stepf\steplabelf{Findf}
\>\>\> Binary Search to find a new relevant set $X_\ell$; $X\gets X\cup X_\ell$; $I\gets I\cup \{\ell\}$.\\
\stepf\steplabelf{Rejf}
\>\>\> If $|I|>k$ then output ``reject'' and halt.\\
\stepf\steplabelf{tx0f}
\>\>\> $t(X)=0$.\\
\stepf\steplabelf{EndRepf}
\>  If $t(X)=100\lambda\ln(100k)/\epsilon$ then \\
\stepf\>\>\>\> Choose a random uniform $w$; \\
\stepf\>\>\>\> Output$(X,w,f(x_X\circ w_{\overline{X}}))$.
  \end{tabbing}
  \end{minipage}}}
  \end{center}
	\caption{A procedure that removes variables from $f$ that only appears in large size terms.}
	\label{A3f}
	\end{figure}

We explain our technique by proving the result for $s$-term DNF.

We remind the reader that for a term $T$, the size of $T$ is the number of variables that are in it. For a variable $x$ and $\xi\in\{0,1\}$, $x^\xi=x$ if $\xi=0$ and $x^\xi=\overline{x}$ if $\xi=1$. For a term $T=x_{i_1}^{c_1}\wedge \cdots\wedge  x_{i_v}^{c_v}$ we denote by $\Va(T)=\{x_{i_1},\ldots,x_{i_v}\}$, the set of variables that appears in $T$. For a set of terms $\T$ we denote $\Va(\T)=\cup_{T\in\T}\Va(T)$. Here $\lambda > 1$ is any constant and we use $c = O(\log \lambda)$ to denote a large
constant.

Consider the procedure {\bf Approx$C$} in Figure~\ref{A3f}. We will prove the following two Lemmas

\begin{lemma}\label{DNF1} Let $f$ be an $s$-term DNF. {\bf Approx$C$} makes $\tilde O(s/\epsilon)$ queries and, with probability at least $9/10$, outputs $X$ and $w$ such that
\begin{enumerate}
\item $f(x_X\circ w_{\overline{X}})$ is $s$-term DNF.
\item The number of relevant variables in $f(x_X\circ w_{\overline{X}})$ is at most $3cs\log(s/\epsilon)=O(s\log(s/\epsilon))$.
\end{enumerate}
\end{lemma}

\begin{lemma}\label{DNF2} Let $f$ be $\epsilon$-far from every $s$-term DNF. {\bf Approx$C$} makes $\tilde O(s/\epsilon)$ queries and either rejects, or with probability at least $9/10$, outputs $X$ and $w$ such that $f(x_X\circ w_{\overline{X}})$ is $(1-1/\lambda)\epsilon$-far from every $s$-term DNF.
\end{lemma}

We first prove Lemma~\ref{DNF1}
\begin{proof}
Consider an $s$-term DNF, $f=T_1\vee T_2\vee \cdots\vee T_{s'}$, where $s'\le s$.
Let $\T=\{T_1,\ldots,T_{s'}\}$. Let $\T_1=\{T\in {\cal T} : |\Va(T)|\le m\}$ be the set of terms in $\T$ of size at most $m:=c\log(s/\epsilon)$ and let $R_1=\Va(\T_1)$ be the set of variables that appear in the terms in $\T_1$. Let $\T_2=\{T\in \T\ :\ |\Va(T)\backslash R_1|\le m\}$ be the set of terms $T\in\T$ that contain at most $m$ variables not in $R_1$. Let $R_2=R_1\cup \Va(\T_2)$ be the variables in $R_1$ and of the terms in $\T_2$. Let $\T_3=\{T\in \T\ :\ |\Va(T)\backslash R_1|> m\}$ be the set terms $T\in \T$ that contain more than $m$ variables that are not in $R_1$. Let $\T_4=\{T\in \T\ :\ |\Va(T)\backslash R_2|\ge m\}$ be the set of terms in $T\in \T$ that contain at least $m$ variables not in $R_2$. Let $R_3=R_2\cup \Va(\T_3\backslash \T_4)$ be the set of variables in $R_2$ and of the terms in $\T_3\backslash \T_4$. Then $$|R_1|\le ms,|R_2|\le 2ms, |R_3|\le 3ms \mbox{\ and\ } \T_4\subseteq \T_3.$$

 In steps~\ref{par1f}-\ref{par2f}, procedure {\bf Approx$C$} uniformly at random partitions $[n]$ into $r=8ms$ sets $X_1,\ldots,X_r$. Suppose that the variables in $R_2$ are distributed to the sets $X_{j_1},\ldots,X_{j_q}$, $q\le |R_2|\le 2ms$.

For each $T\in \T_4$, the expected number of variables in $\Va(T)\backslash R_2$ that are not distributed to one of the sets $X_{j_1},\ldots,X_{j_q}$ is greater or equal to
$\left(1-{q}/{r}\right)m = (3/4)m$. By Hoeffding's bound, Lemma~\ref{Hoeffding}, and the union bound, the probability that it is greater than $m/2$ in every term $T\in \T_4$ is at least $1-s \cdot exp(-m/8)\ge 99/100$.

In steps~\ref{Chof}-\ref{con1f} procedure {\bf Approx$C$} are repeated $M$ times and it finds relevant sets using two random uniform strings $u$ and $v$. If $f(u_X\circ v_{\overline{X}})\not=f(u)$ then a new relevant set is found. Consider any $T\in \T_3$. The size of $T$ is at least $m$. Suppose $T=x_{a_1}^{\xi_1}\wedge \cdots \wedge x_{a_{m'}}^{\xi_{m'}}$, $m'\ge m$. For random uniform $u,v\in\{0,1\}^n$, the probability that there is no $j\in[m']$ such that $u_{a_j}= (u_X\circ v_{\overline{X}})_{a_j}=\xi_j$ is at most $(3/4)^{m}$. The probability that this happens for at least one $T\in \T_3$ and at least one of the $M$ randomly uniformly chosen $u$ and $v$ in the procedure is at most $(3/4)^{m}sM\le 1/100$. Notice that if $u_{a_j}= (u_X\circ v_{\overline{X}})_{a_j}=\xi_j$ then $T(u)=T(u_X\circ v_{\overline{X}})=0$ and $T(w)=0$ for every string $w$ in the binary search that is made to find a new relevant set. Therefore, with probability at least $99/100$, the procedure runs as if $f$ contains no terms in $\T_3$. Let $f'=\vee_{T\in\T\backslash \T_3}T$. With probability at least $99/100$ the procedure runs as if $f=f'$. The number of relevant variables in $f'$ is at most $|R_2|\le 2ms$ and all those variables are distributed to $X_{j_1},\ldots,X_{j_q}$. Therefore, with probability at least $99/100$, the procedure generates at most $2ms<k$ relevant sets and therefore, by step~\ref{Rejf}, it does not reject and those relevant sets are from $X_{j_1},\ldots,X_{j_q}$.

The output of the procedure is $X\subseteq X_{j_1}\cup\cdots\cup X_{j_q}$ and a random uniform $w$. We now show that with high probability $f(x_X\circ w_{\overline{X}})$ contains at most $k=3ms$ relevant variables.

We have shown above that with probability at least $99/100$ every term $T\in \T_4$ contains at least $m/2$ variables that are not distributed to $X_{j_1},\ldots,X_{j_q}$. Therefore, for a fixed term $T\in\T_4$ and for a random uniform $w$, the probability that $T(x_X\circ w_{\overline{X}})=0$ is at least $1-(1/2)^{m/2}$. The probability that $T(x_X\circ w_{\overline{X}})=0$ for every $T\in \T_4$ is at least $1-s(1/2)^{m/2}\ge 99/100$.
Therefore, when we randomly uniformly choose $w\in \{0,1\}^n$, with probability at least $99/100$, the function $f(x_X\circ w_{\overline{X}})$ does not contain terms from $\T_4$. Thus, with probability at least $99/100$, $f(x_X\circ w_{\overline{X}})$ contains at most $|R_3|\le 3ms$ variables.

Now as in the proof of Lemma~\ref{Query01}, the query complexity is
$2M+k\log r=\tilde O(s/\epsilon)$.
\end{proof}

We now prove Lemma~\ref{DNF2}.

\begin{proof}
As in the proof of Lemma~\ref{cloose}, the probability that the algorithm fails to output $X$ that satisfies $\Pr_{x,y\in U}[f(x_X\circ y_{\overline{X}})=f(x)]\le \epsilon/(100\lambda)$ is at most
$$k\left(1-\frac{\epsilon}{100\lambda}\right)^{(100\lambda\ln(100/k))/\epsilon}=\frac{1}{100}.$$ If $\Pr_{x,y\in U}[f(x_X\circ y_{\overline{X}})=f(x)]\le \epsilon/(100\lambda)$ then, by Markov's inequality, Lemma~\ref{Markov}, for a random uniform $w$, with probability at least $99/100$, $\Pr_{x\in U}[f(x_X\circ w_{\overline{X}})=f(x)]\le \epsilon/\lambda$.

Now as in the proof of Lemma~\ref{Query01}, the query complexity is
$2M+k\log r=\tilde O(s/\epsilon)$.
\end{proof}

In the next subsection, we give the result for $s$-term DNF, and then we show how to use the above technique to other classes.

\subsection{Testing $s$-term DNF}\label{RVAST}
We have shown in Lemma~\ref{DNF1} and \ref{DNF2} that, using $\tilde O(s/\epsilon)$ queries, the problem of testing $s$-Term DNF can be reduced to the problem of testing $s$-Term DNF with $k=O(s\log(s/\epsilon))$ relevant variables. We then can use {\bf Tester$C$} for the latter problem. This gives a tester for $s$-term DNF that makes $\tilde O(s^2/\epsilon)$. This is because the number of $s$-term DNF with $k=O(s\log (s/\epsilon))$ relevant variables is $2^{\tilde O(s^2)}$. We will now show how to slightly change the tester {\bf Tester$C$} and get one that makes $\tilde O(s/\epsilon)$ queries.

In {\bf Tester$C$} the procedures {\bf ApproxTarget}, {\bf TestSet} and {\bf Close$fF$} make $O(k\log k)$, $O(k)$ and $O((k/\epsilon)\log (k/\epsilon))$ queries, respectively. This is $\tilde O(s/\epsilon)$ queries. Therefore, the only procedure that makes $\tilde O(s^2/\epsilon)$ queries in {\bf Tester$C$} is {\bf Close$FCU$}. So we will change this procedure.

Let $h$ be an $s$-term DNF. Notice that in step \ref{FFs16} in {\bf Close$FCU$}, for a random uniform $z=(z_1,\ldots,z_q)\in\{0,1\}^q$, the probability that $z$ satisfies a term $T$ in $h$ of size at least $c\log (s/\epsilon)$ (i.e., $T(z)=1$), is at most $(\epsilon/s)^c$. Therefore, if in {\bf Close$FCU$} we define $C^*$ to be the class of all $s$-term DNF with terms of size at most $c\log (s/\epsilon)$, then the probability that for at least one of the $\tau=(3/\epsilon)\log(|C^*|/\delta)$ random uniform $z=(z_1,\ldots,z_q)$ and  at least one of the terms $T$ in $h$ of size at least $c\log (s/\epsilon)$, $T$ satisfies $z$, is at most
$s\tau (\epsilon/s)^c\le 1/100$. Therefore, with probability at least $99/100$ the algorithm runs as if $h$ is $s$-term DNF with terms of size at most $c\log (s/\epsilon)$ and then accept. Since $\log|C^*|=\tilde O(s)$ we get

\begin{theorem} For any $\epsilon>0$, there is a two-sided adaptive algorithm for $\epsilon$-testing $s$-Term DNF that makes $\tilde O(s/\epsilon)$ queries.
\end{theorem}
For completeness we wrote the tester. See {\bf TesterApprox$C$} in Figure~\ref{yTester}. In the tester $C(\{y_1,\ldots,y_q\},c\log(s/\epsilon))$ is the class of all $s$-term DNFs with terms of size at most $c\log(s/\epsilon)$ over $q$ variables.

\section{Results}\label{result2}

In this section, we extend the technique used in the previous section to other classes.

\newcounter{yALC4}
\setcounter{yALC4}{0}
\newcommand{\ystepd}{\stepcounter{yALC4}$\arabic{yALC4}.\ $\>}
\newcommand{\ysteplabeld}[1]{\addtocounter{yALC4}{-1}\refstepcounter{yALC4}\label{#1}}
\begin{figure}[h!]
  \begin{center}
  \fbox{\fbox{\begin{minipage}{28em}
  \begin{tabbing}
  xxx\=xxxx\=xxxx\=xxxx\=xxxx\=xxxx\= \kill
  {{\bf TesterApprox$C$}$(f,\D,\epsilon)$}\\
 {\it Input}: Oracle that access a Boolean function $f$ and $\D$. \\
  {\it Output}: Either ``reject'' or ``accept''\\ \\
  xxx\=xxxx\=xxxx\=xxxx\=xxxx\=xxxx\= \kill
\ystepd\ysteplabeld{yFF001}
  $(X,w)\gets ${\bf Approx$C$}$(f,\epsilon,1/6)$. \\
  \ystepd\ysteplabeld{yFF002}
  $h:= f(x_X\circ w_{\overline{X}})$.\\
\ystepd\ysteplabeld{yFF01}
  $(X,V,I)\gets ${\bf ApproxTarget}$(h,U,\epsilon,1/6)$. \\
\ystepd\ysteplabeld{yFF02}
  {\bf TestSets}$(X,V,I)$.\\
  \ystepd {\bf Close$fF$}$(f,U,\epsilon,1/15)$\\
  \ystepd\ysteplabeld{yFF07}$C^*\gets C(\{y_1,\ldots,y_q\},c\log(s/\epsilon))$ where $q=|V|$\\
  {\bf Test closeness to $C^*$}\\
  \ystepd\ysteplabeld{yFF15}Repeat $\tau=(3/\epsilon)\log(30|C^*|)$ times\\
  \ystepd\ysteplabeld{yFF16}\> Choose $(z_1,\ldots,z_q)\in U$.\\
  \ystepd\ysteplabeld{yFF17}\> For every $g\in C^*$\\
  \ystepd\ysteplabeld{yFF18}\>\>If $g(z)\not=F(z)$ then $C^*\gets C^*\backslash \{g\}$.\\
  \ystepd\ysteplabeld{yFF19}\>\>If $C^*=\emptyset$ then ``reject''\\
  \ystepd\ysteplabeld{yFF20} Return ``accept''
  \end{tabbing}
  \end{minipage}}}
  \end{center}
	\caption{A procedure for testing $s$-Term DNF}
	\label{yTester}
	\end{figure}

\subsection{Testing $s$-Term Monotone DNF}
We first use the algorithm {\bf LearnMonotone} in Figure~\ref{MDNFL} to show
\begin{lemma}\label{LearnMDNF2} Let $f:\{0,1\}^n\to \{0,1\}$ be an $s$-term Monotone DNF. For constant $\delta$, algorithm {\bf LearnMonotone}$(f,U,\epsilon/2,\delta/2,s,2(\log (s/\epsilon)+\log(1/\delta)))$ asks $O(s/\epsilon)$ ExQ$_U$ and $O(s(\log n+\log s)\cdot$ $\log (s/\epsilon))$ MQ and, with probability at least $1-\delta$, learns an $s$-term monotone DNF $h$ that satisfies $\Pr_U[h\not =f]\le \epsilon$.
\end{lemma}
\begin{proof} Let $\hat f$ be the function $f$ without the terms of size greater than $2(\log(s/\epsilon)+\log(1/\delta))$. Then $$\Pr_U[f\not=\hat f]\le s2^{-2(\log(s/\epsilon)+\log(1/\delta)))}\le \frac{\epsilon}{2}.$$

In the algorithm {\bf LearnMonotone} the probability that one of the assignments in step~\ref{MDNF02} (that is, $a$ where $a\in U$) satisfies one of the terms in $f$ of size greater than $2(\log(s/\epsilon)+\log(1/\delta))$ is less than $$(4s/\epsilon)(\log(1/\delta))s2^{-2(\log(s/\epsilon)+\log(1/\delta))}\le \frac{\delta}{2}.$$ Also for a monotone term $T$, if $T(a)=0$ then for any $y$, $T(a*y)=0$. Therefore, with probability at least $1-\delta/2$, the algorithm runs as if $f$ is $\hat f$ (which is $s$-term monotone $(2(\log(s/\epsilon)+\log(1/\delta)))$-DNF). By Lemma~\ref{LearnMDNF}, if the target is $\hat f$ then, with probability at least $1-\delta/2$, {\bf LearnMonotone} outputs $h$ that is $(\epsilon/2)$-close to $\hat f$.
Since $\hat f$ is $(\epsilon/2)$-close to $f$ and $h$ is $(\epsilon/2)$-close to $\hat f$ we have that $h$ is $\epsilon$-close to $f$. This happens with probability at least $1-\delta$.

The number of queries follows from Lemma~\ref{LearnMDNF}.
\end{proof}

We now prove
\begin{theorem} For any $\epsilon>0$, there is a polynomial time two-sided adaptive algorithm for $\epsilon$-testing $s$-Term Monotone DNF that makes $\tilde O(s/\epsilon)$ queries.
\end{theorem}
\begin{proof} We first run {\bf Approx$C$} and get an $s$-term monotone DNF $h$ with $O(s\log(s/\epsilon))$ variables that is $(\epsilon/6)$-close to $f$. We then use Theorem~\ref{TesterL02} with Lemma~\ref{LearnMDNF2}.
\end{proof}

We also have
\begin{theorem} For any $\epsilon>0$, there is a polynomial time two-sided adaptive algorithm for $\epsilon$-testing $s$-Term Unate DNF that makes $\tilde O(s/\epsilon)$ queries.
\end{theorem}
\begin{proof}
The proof is similar to the proof of Theorem~\ref{MUnate}.\end{proof}

\subsection{Testing Size-$s$ Boolean Formula and Size-$s$ Boolean Circuit}
A Boolean formula is a rooted tree in which each internal node has arbitrarily many children and is labeled with AND or OR. Each leaf is labeled with a Boolean variable $x_i$ or its negation $\bar x_i$. The size of a Boolean formula is the number of AND/OR gates it contains. The class size-$s$ Boolean Formula is the class of all Boolean formulas of size at most $s$.

A Boolean circuit is a rooted directed acyclic graph with internal nodes labeled with an AND, OR or NOT gate. Each AND/OR gate is allowed to have arbitrarily many descendants. Each directed path from the root ends in one of the nodes $x_1,x_2,\ldots,x_n,0,1$.

The same analysis that we did for $s$-term DNF also applies to size-$s$ Boolean formulas and size-$s$ Boolean circuit. Analogous to the size of terms, we take the number of distinct literals a gate has. If the gate is labeled with AND (respectively, OR) and the number of distinct literals it has is more than $c\log(s/\epsilon)$, then we replace the gate with a node labeled with $0$ (respectively, $1$) and remove all the edges to its children.

Therefore we have

\begin{lemma} Lemma~\ref{DNF1} and \ref{DNF2} are also true for size-$s$ Boolean formulas and size-$s$ Boolean circuit.
\end{lemma}

We now prove
\begin{theorem} For any $\epsilon>0$, there is a two-sided adaptive algorithm for $\epsilon$-testing size-$s$ Boolean Formula that makes $\tilde O(s/\epsilon)$ queries.
\end{theorem}
\begin{proof} Similar to testing $s$-term DNF, we can ignore gates that have more than $c\log (s/\epsilon)$ distinct literals. Just replace it with $0$ if its label is AND and with $1$ if it is OR.

The number of Boolean formulas of size $s$ that have at most $c\log (s/\epsilon)$ distinct literal in each gate and at most $k=O(s\log (s/\epsilon))$ variables is $2^{\tilde O(s)}$,~\cite{DiakonikolasLMORSW07}. The rest of the proof goes along with the proof of testing $s$-term DNF.
\end{proof}

The number of Boolean circuits of size $s$ that have at most $c\log (s/\epsilon)$ distinct literals in each gate and at most $k=O(s\log (s/\epsilon))$ variables is $2^{\tilde O(s^2)}$,~\cite{DiakonikolasLMORSW07}. Then similar to the above proof one can show

\begin{theorem} For any $\epsilon>0$, there is a two-sided adaptive algorithm for $\epsilon$-testing size-$s$ Boolean Circuit that makes $\tilde O(s^2/\epsilon)$ queries.
\end{theorem}

\subsection{Testing $s$-Sparse Polynomial}
In the literature, the first testing algorithm for the class $s$-Sparse Polynomial runs in exponential time~\cite{DiakonikolasLMORSW07} and makes $\tilde O(s^4/\epsilon^2)$ queries. Chakraborty et al., \cite{ChakrabortyGM11}, then gave another exponential time algorithm that makes $\tilde O(s/\epsilon^2)$ queries. Diakonikolas et al. gave in~\cite{DiakonikolasLMSW11} the first polynomial time testing algorithm that makes $poly(s,1/\epsilon)$ queries. Here we prove
\begin{theorem} For any $\epsilon>0$, there is a two-sided adaptive algorithm for $\epsilon$-testing $s$-Sparse Polynomial that makes $\tilde O(s^2/\epsilon)$ queries.
\end{theorem}

\newcounter{ALCMp}
\setcounter{ALCMp}{0}
\newcommand{\stepMp}{\stepcounter{ALCMp}$\arabic{ALCMp}.\ $\>}
\newcommand{\steplabelMp}[1]{\addtocounter{ALCMp}{-1}\refstepcounter{ALCMp}\label{#1}}
\begin{figure}[h!]
  \begin{center}
  \fbox{\fbox{\begin{minipage}{28em}
  \begin{tabbing}
  xxx\=xxxx\=xxxx\=xxxx\=xxxx\=xxxx\= \kill
  {{\bf LearnPolyUnif}$(f,\epsilon,\delta,s)$}\\
 {\it Input}: Oracle that accesses a Boolean function $f$ that is $s$-sparse polynomial. \\
  {\it Output}: $h$ that is $s$-sparse polynomial\\ \\
  xxx\=xxxx\=xxxx\=xxxx\=xxxx\=xxxx\= \kill
\stepMp\steplabelMp{pF00}
  $h\gets 0$, $t(h)\gets 0$, $w\gets 0$. \\
  \stepMp\steplabelMp{pF01}
  Repeat $(s/\epsilon)\ln(3s/\delta)\log(3s/\delta)$ times. \\
\stepMp\steplabelMp{pF02}\>
  Choose $a\in U$.\\
\stepMp\steplabelMp{pF02b}\>
  $t(h)\gets t(h)+1$.\\
  \stepMp\steplabelMp{pF03}\>
 If $(f+h)(a)=1$ then \\
 \stepMp\steplabelMp{pF04}\>\>
 $m\gets 0$, $w\gets w+1$\\
 \stepMp\steplabelMp{pF04x}\>\>
 If $w=\log(3s/\delta)$ then Output $h$\\
   \stepMp\steplabelMp{pF04b}\>\>
While $m\le \alpha:=16\cdot (8s/\epsilon)(2\ln(s/\delta)+\ln n)$ and $wt(a)>\log(s/\epsilon)+3$ do\\
   \stepMp\steplabelMp{pF05a}\>\>\>
   $m\gets m+1;$ \\
   \stepMp\steplabelMp{PF05}\>\>\>
   Choose $y\in U$\\
  \stepMp\steplabelMp{PF06}\>\>\>
   If $(f+h)(a*y)=1$ then $a\gets a*y$\\
      \stepMp\steplabelMp{PF07x}\>\>
   If $wt(a)>\log(s/\epsilon)+3$ then Goto~\ref{PF07b}\\
    \stepMp\steplabelMp{PF07y}\>\>
   $w\gets 0$.\\
   \stepMp\steplabelMp{PF07a}\>\>
   $a\gets$ Find a monotone term in $(f+h)(a*x)$\\
     \stepMp\steplabelMp{PF07}\>\>
   $h\gets h+ \prod_{a_i=1}x_i$\\
   \stepMp\steplabelMp{PF07b}\>\>
   $t(h)\gets 0$.\\
   \stepMp\steplabelMp{PF08}\>
   If $t(h)=(1/\epsilon)\ln(3s/\delta)$ then Output $h$
  \end{tabbing}
  \end{minipage}}}
  \end{center}
	\caption{A learning algorithm for $s$-sparse polynomial under the uniform distribution}
	\label{pFL}
	\end{figure}

We have shown in Lemma~\ref{DNF1} and \ref{DNF2} that the problem of testing $s$-Term DNF can be reduced to the problem of testing $s$-Term DNF with $k=O(s\log(s/\epsilon))$ relevant variables. The same reduction and analysis show that the problem of testing $s$-sparse polynomials can be reduced to the problem of testing $s$-sparse polynomials with $k=O(s\log(s/\epsilon))$ relevant variables. The reduction makes $\tilde O(s/\epsilon)$ queries. Thus

\begin{lemma} Lemma~\ref{DNF1} and \ref{DNF2} are also true for $s$-Sparse Polynomials.
\end{lemma}

We can then use {\bf Tester$C$} for the latter problem.
Therefore, all we need to do in this section is to find a learning algorithm for the class of $s$-sparse polynomials with $k=O(s\log(s/\epsilon))$ relevant variables.

Consider the algorithm {\bf LearnPolyUnif} in Figure~\ref{pFL}. We prove

\begin{lemma}\label{LearnPolyUnif} Let $f$ be a $s$-sparse polynomial with $n$ variables. For constant $\delta$, algorithm {\bf LearnPolyUnif} asks $\tilde O(s/\epsilon)$ ExQ$_U$ and $\tilde O((s^2/\epsilon)\log n)$ MQ and, with probability at least $1-\delta$, learns an $s$-sparse polynomial $h$ that satisfies $\Pr_U[h\not =f]\le \epsilon$.
\end{lemma}
\begin{proof} Let $f=M_1+M_2+\cdots+M_{s'}$, $s'\le s$, where $deg(M_1)\le \cdots\le \deg(M_{s''})\le \log(s/\epsilon)+3 < deg(M_{s''+1})\le \cdots \le deg(M_{s'})$. Let $f_1=M_1+\cdots+M_{s''}$ and $f_2=M_{s''+1}+\cdots+M_{s'}$. Then $f=f_1+f_2$ and $\Pr_U[f_2=1]\le s2^{-\log(s/\epsilon)-3}=\epsilon/8$. If $f(a)\not=f_1(a)$ then $f_2(a)=1$ and therefore $\Pr_U[f\not=f_1]\le \Pr_U[f_2=1]\le \epsilon/8$ and $\Pr_U[f=f_1]\ge 1-\epsilon/8$. In fact, if $A(x)$ is the event that $T_i(x)=0$ for all $i>s''$ then $\Pr_U[A]\ge 1-\epsilon/8$.

The algorithm {\bf LearnPolyUnif} is similar to the algorithm {\bf LearnPolynomial} with $d=\log(s/\epsilon)+3$ with the changes that is described below.

First notice that in the algorithm the hypothesis $h$ contains only terms from $f_1$, that is, terms of size at most $\log(s/\epsilon)+3$. This is because in step~\ref{PF07x} the algorithm skips the command that adds a term to $h$ when $wt(a)\ge \log(s/\epsilon)+3$. Suppose at some stage of the algorithm, $h=\sum_{i\in B\subseteq [s'']}M_i$ contains some terms of $f_1$ and let $g=f+h=\sum_{i\in [s']\backslash B}M_i$. Suppose $\Pr_U[f\not=h]\ge \epsilon$. Then $\Pr_U[g=1]\ge \epsilon$. We want to compute the probability that the algorithm finds a term in $f_1+h=\sum_{i\in [s'']\backslash B}M_i$ and not in $f_2$. That is, the probability that it finds a term of $f$ of degree at most $\log(s/\epsilon)+3$. We first have
\begin{eqnarray*}
\Pr_{a\in U}[(f_1+h)(a)=1\wedge A |(f+h)(a)=1]&=&\Pr_{a\in U}[g(a)=1\wedge A |g(a)=1]\\
&=&\frac{\Pr_{a\in U}[g(a)=1\wedge A]}{\Pr_{a\in U}[g(a)=1]}\\
&\ge& \frac{\Pr_{a\in U}[g(a)=1]-\Pr_{a\in U}[\overline{A}]}{\Pr_{a\in U}[g(a)=1]}\\
&\ge& 1-\frac{\epsilon/8}{\epsilon}\ge \frac{7}{8}.
\end{eqnarray*}
That is, if $\Pr_U[f\not=h]\ge \epsilon$ and $(f+h)(a)=1$ then with probability at least $7/8$, $(f_1+h)(a)=1$ and for every term $T$ in $f_2$, $T(a)=0$. If  the event $A(a)$ happens then for any $y$ and for every term $T$ in $f_2$, $T(a*y)=0$, and therefore, for such $a$, the algorithm runs as if the target is $f_1$.

When the algorithm reaches step~\ref{pF03} and finds a string $a\in\{0,1\}^n$ such that $(f+h)(a)=1$, we have three cases:
\begin{enumerate}
\item The event $B\equiv [\ ((f_1+h)(a)=1$ and $A(a)]$ happens.
\item $\Pr_U[f\not=h]\ge \epsilon$ and $\overline{B}$.
\item $\Pr_U[f\not=h]< \epsilon$ and $\overline{B}$.
 \end{enumerate}
 {\bf Case 1.} Notice that steps~\ref{pF04b}-\ref{PF06} are identical to steps~\ref{nPF04b}-\ref{nPF06} in {\bf LearnPolynomial} in Figure~\ref{PFL} with $d=\log(s/\epsilon)+3$. Therefore, as in the proof of Lemma~\ref{LearnPolynomial}, with probability at least $1-\delta/3$ every assignment $a$ that satisfies $B$ gives a term in $f_1$ that is not in $h$.

\noindent
 {\bf Case 2.} This case can happen with probability at most $1/8$. So the probability that it happens $\log(3s/\delta)$ consecutive times is at most $\delta/(3s)$. The probability that it does happen for some of the at most $s$ different hypothesis $h$ generated in the algorithm is at most $\delta/3$. Notice that $w$ counts the number of consecutive times that this case happens and step~\ref{pF04x} outputs $h$ when it does happen $\log(3s/\delta)$ consecutive times. Therefore, with probability at most $\delta/3$ the algorithm halts in step~\ref{pF04x} and output $h$ that satisfies $\Pr[f\not=h]\ge \epsilon$.
 This is also the reason that the algorithm repeats the search for $a$ in step~\ref{pF01}, $(s/\epsilon)\ln(3s/\delta)\log(3s/\delta)$ times which is $\log(3s/\delta)$ times more than in algorithm {\bf LearnPolynomial}.

 When this case happens, the algorithm either ends up with a string $a$ of weight that is greater than $\ell:=\log(s/\epsilon)+3$ and then it ignores this string, or, it ends up with a string of weight less than or equal to $\ell$ and then,  Step~\ref{PF07a} finds a new term in $f_1$. This is because that a string of weight less than or equal to $\ell$ cannot satisfy a term of degree more than $\ell$.

\noindent
{\bf Case 3.} This case cannot happen more than $\log(3s/\delta)$ consecutive times because if it does step~\ref{pF04x} outputs $h$ which is a good hypothesis.
\end{proof}

\section{A General Method for Other Testers}\label{result3}
In this section, we generalize the method we have used in the previous section and then prove some more results

\newcounter{ALCG}
\setcounter{ALCG}{0}
\newcommand{\stepG}{\stepcounter{ALCG}$\arabic{ALCG}.\ $\>}
\newcommand{\steplabelG}[1]{\addtocounter{ALCG}{-1}\refstepcounter{ALCG}\label{#1}}
\begin{figure}[h!]
  \begin{center}
  \fbox{\fbox{\begin{minipage}{28em}
  \begin{tabbing}
  xxx\=xxxx\=xxxx\=xxxx\=xxxx\=xxxx\= \kill
  {\bf Algorithm ApproxGeneral$C$$(f,\epsilon,\delta)$}\\
 {\it Input}: Oracle that accesses a Boolean function $f$\\
  {\it Output}: Either ``$X\subseteq [n],w\in\{0,1\}^n$'' or ``reject''\\ \\
{\bf Partition $[n]$ into $r$ sets}\\
\stepG\steplabelG{par1G}
Set $r=k^{c+1}$.\\
\stepG\steplabelG{par2G}
Choose uniformly at random a partition $X_1,X_2,\ldots,X_r$ of $[n]$\\
\\
{\bf Find a close function and relevant sets} \\
\stepG\steplabelG{SettG}
Set $X=\emptyset$; $I=\emptyset$; $t(X)=0$\\
\stepG\steplabelG{twoG}
Repeat $M=(4c_1k/(\delta\epsilon))\ln(4k/\delta)$ times\\
\stepG\steplabelG{ChoG}
\> Choose $u,v\in U$. \\
\stepG  \> $t(X)\gets t(X)+1$\\
\stepG\steplabelG{con1G}
\> If $f(u_X\circ v_{\overline{X}})\not=f(u)$ then\\
\stepG\steplabelG{FindG}
\>\>\> Find a new relevant set $X_\ell$; $X\gets X\cup X_\ell$; $I\gets I\cup \{\ell\}$.\\
\stepG\steplabelG{RejG}
\>\>\> If $|I|>k$ then Output ``reject''\\
\stepG\steplabelG{tx0G}
\>\>\> $t(X)=0$.\\
\stepG\steplabelG{EndRepG}
\>  If $t(X)=(4c_1/(\delta\epsilon))\ln(4k/\delta)$ then \\
\stepG\>\>\> Choose a random uniform $w$; \\
\stepG\>\>\> Output$(X,w)$.
  \end{tabbing}
  \end{minipage}}}
  \end{center}
	\caption{An algorithm that removes variables from $f$ that have a small influence.}
	\label{A3G}
	\end{figure}

We define the distribution $\D[p]$ to be over $\prod_{i=1}^n\{0,1,x_i\}$ where each coordinate $i$ is chosen to be $x_i$ with probability $p$, $0$ with probability $(1-p)/2$ and $1$ with probability $(1-p)/2$. We will denote by $|f|$ the {\it size} of $f$ in $C$ which is the length of the representation of the function $f$ in $C$.

Consider the algorithm {\bf ApproxGeneral$C$} in Figure~\ref{A3G}. We start with the following result
\begin{lemma}\label{General} Let $\delta<1/2$, $c_1\ge 1,\lambda>1$ and $c\ge 1$ be any constants, $k:=k(\epsilon,\delta,|f|)$ be an integer and $M=(4c_1k/(\delta\epsilon))\ln(4k/\delta)$. Let $C$ be a class of functions where for every $f\in C$ there is $h\in (C\cap k$-Junta$)$ with relevant variables $x(Y)=\{x_i|x\in Y\}$, $Y\subseteq [n]$, and $h'\in (C\cap (\lambda k)$-Junta$)$ that satisfy the following:
\begin{enumerate}
\item\label{ConG1} $\Pr_{z\in \D[1/2]}[f(z)\not=h(z)]\le \delta/(4M)$.
\item\label{ConG2} $\Pr_{y\in \D[1/k^c]}[f(x_Y\circ y_{\overline{Y}})\not= h'(y)]\le \delta/4$.
\end{enumerate}
The algorithm {\bf ApproxGeneral$C$} makes $\tilde O(k/\epsilon)$ queries and,
\begin{enumerate}
\item If $f\in C$ then, with probability at least $1-\delta$, the algorithm does not reject and outputs $X$ and $w$ such that $f(x_X\circ w_{\overline{X}})\in C$ has at most $\lambda k$ relevant variables.
\item For any $f$, if the algorithm does not reject then, with probability at least $1-\delta$, $\Pr[f(x)\not =f(x_X\circ w_{\overline{X}})]\le \epsilon/c_1$.
\end{enumerate}
\end{lemma}
\begin{proof} Let $f\in C$. Let $h\in C\cap k-$Junta and $h'\in (C\cap (\lambda k)$-Junta$)$ be functions that satisfies \ref{ConG1} and~\ref{ConG2}. The algorithm in step~\ref{ChoG} chooses two random uniform strings $u$ and $v$. Define $z=(z_1,\ldots,z_n)$ such that $z_i=0$ if $u_i=v_i=0$, $z_i=1$ if $u_i=v_i=1$, and $z_i=x_i$ if $u_i\not=v_i$. Since $u$ and $v$ are chosen uniformly at random we have that $z\in \D[1/2]$ and therefore, with probability at least $1-\delta/(4M)$, $f(z)=h(z)$. If $f(z)=h(z)$ then $f(u)=h(u)$ and $f(u_X\circ v_{\overline{X}})=h(u_X\circ v_{\overline{X}})$ for any $X\subseteq [n]$. In Step~\ref{FindG} {\bf ApproxGeneral} does a binary search to find a new relevant set. In the binary search it queries strings $a$ that satisfy $a_i=z_i=u_i=v_i$ for all $i$ that satisfies $z_i\in\{0,1\}$. Therefore, $f(a)=h(a)$ for all the strings $a$ generated in the binary search for finding a relevant set. Therefore, with probability at least $1-\delta/(4M)$, $f(a)=h(a)$ for all the queries used in one iteration and, with probability at least $1-\delta/4$, $f(a)=h(a)$ for all the queries used in the algorithm. That is, with probability at least $1-\delta/4$, the algorithm runs as if the target is $h$.

Therefore, if $f\in C$, then with probability at least $1-\delta/4$, each one of the relevant sets discovered in the algorithm contains at least one relevant variable of $h$. Then since $h\in k$-Junta, the algorithm does not reject in Step~\ref{RejG}, that is, $|I|\le k$.

Now we show that, if $f\in C$ then with probability at least $1-\delta/4$, $f(x_X\circ w_{\overline{X}})$ contains at most $\lambda k$ relevant variables. Consider the partition in steps~\ref{par1G}-\ref{par2G} in the algorithm and let $X_{i_1},\ldots,X_{i_{k'}}$, $k'\le k$, be the sets where the indices of the relevant variables $x(Y)$ of $h$ are distributed. Let $X'=X_{i_1}\cup \cdots \cup X_{i_{k'}}$. Notice that for a random uniform $w\in\{0,1\}^n$, $(x_{X'}\circ w_{\overline{X'}})=(x_Y\circ y_{\overline{Y}})$ where $y\in \D[k'/k^{c+1}]$. That is, given that the relevant variables of $h$ are distributed to $k'$ different sets that their union is $X'$, the probability distributions of $(x_{X'}\circ w_{\overline{X'}})$ and of $(x_Y\circ y_{\overline{Y}})$ are identical. Choosing a string in $\D[k'/k^{c+1}]$ can be done by first choosing a string $b$ in $\D[1/k^c]$ and then substitute in each variable $x_i$, $i\not\in Y$, in $b$, $0$, $1$ or $x_i$ with probability $1/2-k'/(2k)$, $1/2-k'/(2k)$ or $k'/k$, respectively. Therefore, since $\Pr_{y\in \D[1/k^c]}[f(x_Y\circ y_{\overline{Y}})\not= h'(y)]\le \delta/4$, we have $\Pr_{y\in \D[k'/k^{c+1}]}[f(x_Y\circ y_{\overline{Y}})\not= h'(y)]\le \delta/4$ and, thus,
with probability at least $1-\delta/4$, we have that $f(x_{X'}\circ w_{\overline{X'}})=h'(x)$. In particular, with probability at least $1-\delta/4$, $f(x_{X'}\circ w_{\overline{X'}})$ has at most $\lambda k$ relevant variables. Since $X\subseteq X'$ we also have with the same probability $f(x_{X}\circ w_{\overline{X}})$ has at most $\lambda k$ relevant variables. This completes the proof of {\it 1.}

Now let $f$ be any boolean function. If the algorithm does not reject then $|I|\le k$. Since for the final $X$, $f(u_X\circ v_{\overline{X}})\not=f(u)$ for $(4c_1/(\delta\epsilon))\ln(4k/\delta)$ random uniform $u$ and $v$, we have that the probability that the algorithm fails to output $X$ that satisfies $\Pr_{x,y\in U}[f(x_X\circ y_{\overline{X}})\not=f(x)]\le \delta\epsilon/(4c_1)$ is at most
$$k\left(1-\frac{\delta\epsilon}{4c_1}\right)^{(4c_1/(\delta\epsilon))\ln(4k/\delta)}=\frac{\delta}{4}.$$ If $\Pr_{x,y\in U}[f(x_X\circ y_{\overline{X}})\not=f(x)]\le \delta\epsilon/(4c_1)$ then, by Markov's inequality, for a random uniform $w$, with probability at least $1-\delta/4$, $\Pr_{x\in U}[f(x_X\circ w_{\overline{X}})\not=f(x)]\le \epsilon/c_1$. This completes the proof of~{\it 2.}
\end{proof}

In the next subsections, we give some results that follow from Lemma~\ref{General}.

\subsection{Testing Decision List}

In \cite{DiakonikolasLMORSW07}, Diakonikolas et al. gave a polynomial time tester for Decision List that makes $\tilde O(1/\epsilon^2)$ queries. In this paper, we give a polynomial time tester that makes $\tilde O(1/\epsilon)$ queries.

We show
\begin{theorem}\label{kdl} For any $\epsilon>0$, there is a two-sided adaptive algorithm for $\epsilon$-testing Decision List that makes $\tilde O(1/\epsilon)$ queries.
\end{theorem}
\begin{proof} Let $f=(x_{i_1},\xi_1,a_1),\ldots,(x_{i_s},\xi_s,a_s)$ be any decision list. We first use Lemma~\ref{General}.

Define $k=\min(s,c'\log(1/(\epsilon\delta)))$ for some large constant $c'$ and $h=(x_{i_1},\xi_1,a_1),\ldots,(x_{i_k},\xi_k,a_k)$. For the distribution $\D[1/2]$, the probability that $f(z)=h(z)$ is $1$ when $k=s$ and at least $1-(3/4)^k\ge 1-1/(3M)$ where $M=(4c_1k/(\delta\epsilon))\ln(4k/\delta)$. For the distribution $\D[1/k^2]$ and $Y=\{i_1,\ldots,i_k\}$, the probability that $f(x_Y\circ y_{\overline{Y}})$ has more than $2 k$ relevant variables is less than $(3/4)^{k}\le \delta/3$.

The query complexity of {\bf ApproxGeneral} is $\tilde O(k/\epsilon)=\tilde O(1/\epsilon)$. Therefore all we need to do to get the result is to give a tester for decision list of size $k=O(\log(1/\epsilon))$ that makes $\tilde O(1/\epsilon)$ queries. The learnability of this class with $\tilde O(1/\epsilon)$ ExQs follows from~\cite{Rivest87,BlumerEHW87}.
\end{proof}

\subsection{Testing $r$-DNF and $r$-Decision List for Constant $r$}
An $r$-decision list is a sequence $f=(T_{1},\xi_1,a_1),\ldots,(T_{s},\xi_s,a_s)$ for any $s$ where $\xi_i,a_i\in\{0,1\}$ and $T_i$ are $r$-terms. This sequence represents the following function: $f(x):=$``If $T_{1}=\xi_1$ then output$(a_1)$ else if $T_{2}=\xi_2$ then output$(a_2)$ else if $\cdots$ else if $T_{s}=\xi_s$ then output$(a_s)$''. The class $r$-Decision List is the class of all $r$-decision lists and the class of Length-$s$ $r$-Decision List is the class of all $r$-decision lists $f=(T_{1},\xi_1,a_1),\ldots,(T_{m},\xi_m,a_m)$ with $m\le s$.

In this subsection, we show
\begin{theorem} Let $r$ be any constant. For any $\epsilon>0$, there is a two-sided adaptive algorithm for $\epsilon$-testing $r$-Decision List and $r$-DNF that makes $\tilde O(1/\epsilon)$ queries.
\end{theorem}

It is known that the class Length-$s$ $r$-Decision List is learnable under any distribution in time $O(n^r)$ using $O((sr\log n+\log (1/\delta))/\epsilon)$ ExQ$_D$,~\cite{BlumerEHW87,Rivest87}. Also we may assume that $T_1,\ldots ,T_s$ are distinct and therefore $s$ is less than $\sum_{i=1}^r{n\choose i}2^i$ , the number of terms of size at most $r$. Thus, for constant $r$, it is enough to prove that $r$-Decision List and $r$-DNF satisfies \ref{ConG1} and \ref{ConG2} in Lemma~\ref{General} with $k=poly(\log(1/\epsilon))$.

We now prove the result for $r$-Decision List when $r$ is constant. The same analysis shows that the result is also true for $r$-DNF.

Consider an $r$-decision list $f=(T_1,\xi_1,a_1)\cdots (T_s,\xi_s,a_s)$. If $\xi_i=0$ then we change $(T_i,0,a_i)$ to the equivalent expression $(x_{i_1}^{{1-c_{i_1}}},1,a_i)\cdots (x_{i_\ell}^{{1-c_{i_\ell}}},1,a_i)$ where $T_i=x_{i_1}^{c_1}\cdots x_{i_\ell}^{c_\ell}$. Therefore we may assume that $\xi_j=1$ for all $j$. In that case we just write $f=(T_1,a_1)\cdots (T_s,a_s)$.

For an $r$-decision list $f=(T_1,a_1)\cdots (T_s,a_s)$, a {\it sublist of $f$} is an $r$-decision list $g=(T_{i_1},a_{i_1})\cdots $ $(T_{i_\ell},a_{i_\ell})$ such that $1\le i_1<i_2<\cdots<i_\ell\le s$.

We first prove
\begin{lemma}
Let $r$ be a constant. For any $r$-decision list $f$ there is a $k_r:=O(\log^r(1/\epsilon))$-length $r$-decision list $h$ that is a sublist of $f$ and satisfies \ref{ConG1} and \ref{ConG2} in Lemma~\ref{General}.
\end{lemma}
\begin{proof} We show that it satisfies \ref{ConG1} in Lemma~\ref{General}. The proof that it also satisfies~\ref{ConG2} is similar.

We give a stronger result as long as $r$ is constant. We prove by induction that for any $r$-decision list $f$ there is a $k_r=O(\log^{r}(1/\epsilon))$-length $r$-decision list $h$ that is a sublist of $f$ and satisfies $\Pr_{z\in \D[1/2]}[f(z)\not=h(z)]\le poly(\epsilon)$.

The proof is by induction on $r$. For $r=1$, the result follows from the proof of Theorem~\ref{kdl} in the previous subsection. Assume the result is true for $r$-decision list. We now show the result for $(r+1)$-decision list.

Let $c$ be a large constant. Let $f=(T_1,a_1)\cdots (T_s,a_s)$ be $(r+1)$-decision list. Let $s_1=1$ and $T_{s_1},\ldots,T_{s_w}$ be a sequence of terms such that $s_1<s_2<\cdots<s_w\le s$ and for every $i$, $s_i$ is the minimal integer that is greater than $s_{i-1}$ such that the variables in $T_{s_i}$ do not appear in any one of the terms $T_{s_1},T_{s_2},\ldots,T_{s_{i-1}}$. Define $h_0=(T_1,a_1)(T_2,a_2)\cdots (T_{s_{w'}},a_{s_{w'}})$ if $w':=c2^{r+1}\ln(1/\epsilon)\le w$ and $h_0=f$ if $w'>w$. Then
\begin{eqnarray*}
\Pr_{z\in \D[1/2]}[f(z)\not=h_0(z)]&\le&\Pr[T_{1}(z)=0\wedge T_{2}(z)=0\wedge \cdots\wedge T_{s_{w'}}(z)=0]\\\\
&\le&\Pr[T_{s_1}(z)=0\wedge T_{s_2}(z)=0\wedge \cdots\wedge T_{s_{w'}}(z)=0]\\
&\le& \left(1-\frac{1}{2^{r+1}}\right)^{w'}\le
poly(\epsilon).
\end{eqnarray*}
Let $S=\{x_{j_1},\ldots,x_{j_t}\}$ be the set of the variables in $T_{s_1},\ldots,T_{s_{w'}}$. Then $t\le (r+1)w'=c(r+1)2^{r+1}\ln(1/\epsilon)$ and every term $T_i$ in $h_0$ contains at least one variable in $S$. Consider all the terms that contains the variable $x_{j_1}$, $T_{i_1}=x_{j_1}T_{i_1}',T_{i_2}=x_{j_1}T_{i_2}',\ldots,T_{i_\ell}=x_{j_1}T_{i_\ell}'$, $i_1<i_2<\cdots<i_\ell$. Consider the $r$-decision list $g:=(T'_{i_1},a_1)(T'_{i_2},a_2)\cdots(T'_{i_\ell},a_\ell)$. By the induction hypothesis there is an $r$-decision list $g'$ that is a sublist of $g$ of length at most $k_{r}=O(\log^{r}(1/\epsilon))$ such that $\Pr_{z\in \D[1/2]}[g(z)\not=g'(z)]\le poly(\epsilon)$. Let $h_1$ be $h_0$ without all the terms $(T_{i_w},a_{i_w})$ that correspond to the terms $(T_{i_w}',a_{i_w})$ that do not occur in $g'$. It is easy to see that $\Pr_{z\in \D[1/2]}[h_0(z)\not=h_1(z)]\le poly(\epsilon)$. We do the same for all the other variables $x_{j_2},\ldots,x_{j_t}$ of $S$ and get a sequence of $r$-decision lists $h_2,h_3,\ldots,h_t$ that satisfies $\Pr_{z\in \D[1/2]}[h_w(z)\not=h_{w+1}(z)]\le poly(\epsilon)$. Therefore $\Pr_{z\in \D[1/2]}[f(z)\not=h_{t}(z)]\le t \cdot poly(\epsilon)=poly(\epsilon)$ and the length of $h_t$ is at most $tk_{r}=O(\log^{r+1}(1/\epsilon))=k_{r+1}$.
\end{proof}
\subsection*{Acknowledgements}

The author would like to express heartfelt gratitude to Oded Goldreich for his invaluable contributions to this paper. Goldreich's insightful discussions, thoughtful comments, and generosity in allowing the inclusion of his overview of the first tester (Section 2) have greatly improved the exposition and presentation of the technique employed in this paper.

\bibliography{TestingRef}

\section{Appendix A}

In this Appendix, we give some bounds used in the paper.

\begin{lemma}\label{Markov} {\bf Markov's Bound}.
Let $X\ge 0$ be a random variable with a finite expected value $\mu=\E[X]$. Then for any real numbers $\kappa,K > 0$,

\begin{eqnarray}
\Pr(X\ge \kappa)\leq {\frac {\E[X]}{\kappa}}.\label{Markov1}\\
\Pr(X\ge K\E[X])\leq {\frac {1}{K}}.\label{Markov2}
\end{eqnarray}
\end{lemma}

\begin{lemma}\label{Chebyshev} {\bf Chebyshev's Bound}.
Let $X$ be a random variable with a finite expected value $\mu=\E[X]$ and finite non-zero variance $\Var[X]=\E[X^2]-\E[X]^2$. Then for any real numbers $\kappa,K > 0$,

\begin{eqnarray}
\Pr(|X-\mu |\geq \kappa\sqrt{\Var[X]} )\leq {\frac {1}{\kappa^{2}}}.\label{Chebyshev1}\\
\Pr(|X-\mu |\geq K)\leq {\frac {\Var[X]}{K^{2}}}.\label{Chebyshev2}
\end{eqnarray}

\end{lemma}

\begin{lemma}\label{Chernoff}{\bf Chernoff's Bound}. Let $X_1,\ldots, X_m$ be independent random variables taking values in $\{0, 1\}$. Let $X=\sum_{i=1}^mX_i$ denotes their sum, and let $\mu = \E[X]$ denotes the sum's expected value. Then
\begin{eqnarray}\Pr[X>(1+\lambda)\mu]\le \left(\frac{e^{\lambda}}{(1+\lambda)^{(1+\lambda)}}\right)^{\mu}\le e^{-\frac{\lambda^2\mu}{2+\lambda}}\le  \begin{cases} e^{-\frac{\lambda^2\mu}{3}} &\mbox{if\ } 0< \lambda\le 1 \\
e^{-\frac{\lambda \mu}{3}} & \mbox{if\ } \lambda>1 \end{cases} .\label{Chernoff1}
\end{eqnarray}
For $0\le \lambda\le 1$ we have
\begin{eqnarray}
\Pr[X<(1-\lambda)\mu]\le \left(\frac{e^{-\lambda}}{(1-\lambda)^{(1-\lambda)}}\right)^{\mu}\le e^{-\frac{\lambda^2\mu}{2}}.\label{Chernoff2}
\end{eqnarray}
\end{lemma}

\begin{lemma}\label{Hoeffding}{\bf Hoeffding's Bound}. Let $X_1,\ldots, X_m$ are independent random variables taking values in $\{0, 1\}$. Let $X=\sum_{i=1}^mX_i$ denote their sum and let $\mu = \E[X]$ denote the sum's expected value. Then for $0\le \lambda\le 1$ we have
\begin{eqnarray}\Pr[X>\mu+\lambda m]\le  e^{-{2\lambda^2m}}
\end{eqnarray}
and
\begin{eqnarray}
\Pr[X<\mu-\lambda m]\le  e^{-{2\lambda^2m}}
\end{eqnarray}
\end{lemma}

\end{document}